\documentclass[11pt]{article}

\title{Fully Dynamic $k$-Median with  
Near-Optimal \\  Update Time and Recourse}

\author{
Sayan Bhattacharya 
\\University of Warwick 
\\\texttt{s.bhattacharya@warwick.ac.uk}  
\and 
Mart\'{i}n Costa 
\\University of Warwick
\\\texttt{martin.costa@warwick.ac.uk}  
\and
Ermiya Farokhnejad
\\University of Warwick
\\\texttt{ermiya.farokhnejad@warwick.ac.uk} 
}

\date{}

\usepackage[utf8]{inputenc}

\usepackage[ruled,vlined,linesnumbered,noend]{algorithm2e}
\SetKwInOut{Parameter}{Parameters}

\usepackage{amsthm}
\usepackage{amsmath}
\usepackage{amssymb}
\usepackage{enumerate}
\usepackage{graphicx}
\usepackage[margin=1in]{geometry}
\usepackage{dsfont}
\usepackage{booktabs} 
\usepackage{tabularx}
\usepackage{multirow}
\usepackage{enumitem}
\usepackage{tablefootnote}

\usepackage[dvipsnames]{xcolor}
\definecolor{ForestGreen}{rgb}{0.1333,0.5451,0.1333}
\definecolor{DarkRed}{rgb}{0.65,0,0}

\usepackage{hyperref}
\hypersetup{
    colorlinks=true,
    linkcolor=DarkRed,
    filecolor=magenta,      
    urlcolor=cyan,
    citecolor=ForestGreen,
}

\usepackage{cleveref}

\usepackage{framed}
\usepackage[framemethod=tikz]{mdframed}
\usepackage{tikz-cd}

\DeclareMathAlphabet{\mathmybb}{U}{bbold}{m}{n}

\newtheorem{theorem}{Theorem}[section]
\newtheorem{lemma}[theorem]{Lemma}

\newtheorem{corollary}[theorem]{Corollary}

\newtheorem{definition}[theorem]{Definition}
\newtheorem{invariant}[theorem]{Invariant}
\newtheorem{remark}[theorem]{Remark}

\newtheorem{claim}[theorem]{Claim}

\newcommand{\ground}[0]{\mathbf{P}}
\newcommand{\calP}[0]{\mathcal{P}}
\newcommand{\calU}[0]{\mathcal{U}}
\newcommand{\calV}[0]{\mathcal{V}}
\newcommand{\calW}[0]{\mathcal{W}}
\newcommand{\calR}[0]{\mathcal{R}}
\newcommand{\calT}[0]{\mathcal{T}}

\newcommand{\calS}[0]{\mathcal{S}}
\newcommand{\calF}[0]{\mathcal{F}}

\newcommand{\calC}[0]{\mathcal{C}}

\newcommand{\cost}[1]{\textnormal{\textsf{Cost}} 
\left( #1 \right)}
\newcommand{\avcost}[1]{\textnormal{\textsf{AverageCost}} \left( #1 \right)}
\newcommand{\costd}[1]{\textsf{Cost}^d\left( #1 \right)}
\newcommand{\costdp}[1]{\textsf{Cost}^{d'}\left( #1 \right)}
\newcommand{\OPT}[0]{\textnormal{\textsc{OPT}}}
\newcommand{\FOPT}[0]{\textnormal{\textsc{FOPT}}}

\newcommand{\init}{\textnormal{\texttt{init}}}
\newcommand{\final}{\textnormal{\texttt{final}}}
\newcommand{\uo}[0]{\calU^{(0)}}
\newcommand{\tuo}[0]{\tilde{\calU}^{(0)}}
\newcommand{\tu}[0]{\tilde{\calU}}
\newcommand{\po}[0]{\calP^{(0)}}
\newcommand{\pl}[0]{\calP^{(l+1)}}

\newcommand{\MainEp}[0]{\textsc{Main-Epoch}}
\newcommand{\RemoveCenters}[0]{\textsc{Remove-Centers}}
\newcommand{\RandLocalSearch}[0]{\textsc{Rand-Local-Search}}
\newcommand{\LazyUpdates}[0]{\textsc{Lazy-Updates}}
\newcommand{\DevelopCenters}[0]{\textsc{Develop-Centers}}

\newcommand{\Robustify}[0]{\textsc{Robustify}}
\newcommand{\EmptyS}[0]{\textsc{Empty-S}}
\newcommand{\MakeRbst}[0]{\textsc{Make-Robust}}
\newcommand{\FindSuspects}[0]{\textsc{Find-Suspects}}
\newcommand{\FastOneMedian}[0]{\textsc{Fast-One-Median}}

\begin{document}

\maketitle

\pagenumbering{gobble}

\begin{abstract}
In metric $k$-clustering, we are given as input a set of $n$ points in a general metric space, and we have to pick $k$ {\em centers} and  {\em cluster} the input points around these chosen centers, so as to minimize an appropriate objective function. In recent years, significant effort has been devoted to the study of metric $k$-clustering problems in a {\em dynamic setting}, where the input keeps changing via updates (point insertions/deletions), and we have to maintain a good clustering throughout these updates [Fichtenberger, Lattanzi, Norouzi-Fard and Svensson, SODA'21; Bateni, Esfandiari, Fichtenberger, Henzinger, Jayaram, Mirrokni and Weise, SODA'23; Lacki, Haeupler, Grunau, Rozhon and Jayaram, SODA'24; Bhattacharya, Costa, Garg, Lattanzi and Parotsidis, FOCS'24; Forster and Skarlatos, SODA'25]. The performance of such a dynamic algorithm is measured in terms of three parameters: (i) Approximation ratio, which signifies the quality of the maintained solution, (ii) Recourse, which signifies how stable the maintained solution is, and (iii) Update time, which signifies the efficiency of the algorithm.

We consider a textbook metric $k$-clustering problem,  {\em metric $k$-median}, where the objective is the sum of the distances of the points to their nearest centers. We design the first dynamic algorithm for this problem with near-optimal guarantees across all three performance measures (up to a constant factor in approximation ratio, and polylogarithmic factors in recourse and update time). Specifically, we obtain a $O(1)$-approximation algorithm for dynamic metric $k$-median with $\tilde{O}(1)$ recourse and $\tilde{O}(k)$ update time. Prior to our work, the state-of-the-art here was the recent result of [Bhattacharya, Costa, Garg, Lattanzi and Parotsidis, FOCS'24], who obtained  $O(\epsilon^{-1})$-approximation ratio with $\tilde{O}(k^{\epsilon})$ recourse and $\tilde{O}(k^{1+\epsilon})$ update time.

We achieve our results by carefully synthesizing the concept of {\em robust centers} introduced in [Fichtenberger, Lattanzi, Norouzi-Fard and Svensson, SODA'21] along with the {\em randomized local search} subroutine from [Bhattacharya, Costa, Garg, Lattanzi and Parotsidis, FOCS'24], in addition to several key technical insights of our own.

\end{abstract}

\newpage

\tableofcontents

\newpage

\pagenumbering{arabic}

    \part{Extended Abstract}\label{part:0}

\section{Introduction}

Consider a metric space $(\calP, d)$ over a set $\calP$ of $n$ points, with a distance function $d : \calP \times \calP \to \mathbb{R}^{\geq 0}$, and a positive integer $k \leq n$. In the {\em metric $k$-median} problem, we have to pick a set $\calU \subseteq \calP$ of $k$ {\em centers}, so as to minimize the objective function $\cost{\calU, \calP} := \sum_{p \in \calP} d(p, \calU)$, where $d(p, \calU) := \min_{q \in \calU} d(p, q)$ denotes the distance between a point $p$ and its nearest center in $\calU$. We assume that we have access to the function $d$ via a {\em distance oracle}, which returns the value of $d(p, q)$ for any two points $p, q \in \calP$ in $O(1)$ time. We further assume that $1 \leq d(p, q) \leq \Delta$ for all $p, q \in \calP, p \neq q$, where $\Delta$ is an upper bound on the {\em aspect ratio} of the metric space.

Metric $k$-median is a foundational problem in clustering, is known to be NP-hard, and approximation algorithms for this problem are taught in standard textbooks~\cite{WilliamsonShmoys2011}. In particular, it has a $O(1)$-approximation algorithm that runs in $\tilde{O}(kn)$ time~\cite{MettuP02},\footnote{Throughout this paper, we use the $\tilde{O}(\cdot)$ notation to hide polylogarithmic factors in $k, n$ and $\Delta$.} and it is known that we {\em cannot} have any $O(1)$-approximation algorithm for metric $k$-median with $o(kn)$ runtime~\cite{BadoiuCIS05}. 

In recent years, substantive effort have been devoted to the study of this problem in a {\em dynamic setting}, when the underlying input changes over time~\cite{icml/LattanziV17,nips/Cohen-AddadHPSS19,soda/FichtenbergerLN21,esa/HenzingerK20,ourneurips2023,esa/TourHS24,FOCS24kmedian}. To be more specific, here the input changes by a sequence of updates; each update inserts/deletes a point in $\calP$.
Throughout these updates, we have to maintain a set of $k$ centers $\calU \subseteq \calP$ which form an approximate $k$-median solution to the current input $\calP$. Such a {\em dynamic} algorithm's performance is measured in terms of its: (i) Approximation ratio, (ii) {\em Recourse}, which is the number of changes (i.e., point insertions/deletions) in the maintained solution $\calU$ per update, and (iii) {\em Update time}, which is the time taken by the algorithm to process an update. In a sense, approximation ratio and recourse respectively measures the ``quality'' and the ``stability'' of the maintained solution, whereas update time measures the ``efficiency''
  of the algorithm. 

We design a dynamic algorithm for this problem with almost optimal performance guarantees with respect to all these measures. Our main result is summarized in the theorem below.

\begin{theorem}\label{main-theorem-2}
    There is a randomized dynamic algorithm for the metric $k$-median problem that has $O(1)$-approximation ratio, $O(\log^2 \Delta)$  recourse and $\tilde{O}(k)$ update time, w.h.p.\footnote{Both our recourse and update time bounds are amortized. Throughout the paper, we do not make any distinction between amortized vs worst-case bounds.}
\end{theorem}

\medskip
\noindent {\bf Remarks.} A few important remarks are in order. First, note that there {\em cannot} exist a dynamic $O(1)$-approximation algorithm for our problem with $o(k)$ update time, for otherwise we would get a {\em static} algorithm for metric $k$-median with $O(1)$-approximation ratio and $o(kn)$ runtime:  Simply let the dynamic algorithm handle a sequence of $n$ insertions corresponding to the points in the static input, and return the solution maintained by the dynamic algorithm at the end of this update sequence. This would contradict the $\Omega(kn)$ lower bound on the runtime of any such static algorithm, derived in~\cite{BadoiuCIS05}. Furthermore, it is easy to verify that we cannot achieve $o(1)$ recourse in the fully dynamic setting, and hence, our dynamic algorithm is almost optimal (up to a $O(1)$ factor in approximation ratio and polylogarithmic factors in recourse and update time).

Second, in this extended abstract  we focus the {\em unweighted} metric $k$-median problem, only to ease notations. In the full version (see \Cref{part:full}), we show that \Cref{main-theorem-2} seamlessly extends to the {\em weighted} setting, where each point $p \in \calP$ has a weight $w(p) > 0$ associated with it, and we have to maintain a set $\calU$ of $k$ centers that (approximately) minimizes  $\sum_{p \in \calP} w(p) \cdot d(p, \calU)$. Moreover, our result extends to the related {\em metric $k$-means} problem as well, where we have to pick a set $\calU \subseteq \calP$ of $k$ centers so as to minimize $\sum_{p \in \calP} \left( d(p, \calU)\right)^2$. We can get a dynamic $O(1)$-approximation algorithm for (weighted) 
metric $k$-means  that has $\tilde{O}(1)$  recourse and $\tilde{O}(k)$  update time, w.h.p.

 Finally, \Cref{tab:kmedian} compares our result against prior  state-of-the-art. Until very recently, all known  algorithms~\cite{nips/Cohen-AddadHPSS19,esa/HenzingerK20,ourneurips2023} for fully dynamic metric $k$-median had a trivial recourse bound of $\Omega(k)$,
 which can be obtained by computing a new set of $k$ centers from scratch after every update (at the expense of $\Omega(kn)$ update time). Then, in FOCS 2024, \cite{FOCS24kmedian} took a major step towards designing an almost optimal algorithm for this problem, by achieving $O(\epsilon^{-1})$-approximation ratio, $\tilde{O}(k^{1+\epsilon})$ update time and $\tilde{O}(k^{\epsilon})$ recourse. To achieve truly polylogarithmic recourse and $\tilde{O}(k)$ update time using the algorithm of \cite{FOCS24kmedian}, we have to set $\epsilon = O\left(\frac{\log \log k}{\log k}\right)$.
This, however, increases the approximation guarantee to $\Omega\left(\frac{\log k}{\log \log k}\right)$. In contrast, we achieve $O(1)$-approximation ratio, $\tilde{O}(k)$ update time and $\tilde{O}(1)$ recourse.

\begin{table}[h!]
  \centering
  \begin{tabular}{cccc}
    \toprule
    Approximation Ratio & Update Time & Recourse & Paper \\
    \midrule
     $O(1)$ & $\tilde O(n + k^2)$ & $O(k)$ & \cite{nips/Cohen-AddadHPSS19} \\
     $O(1)$ & $\tilde O(k^2)$ & $O(k)$ & \cite{esa/HenzingerK20} \\ 
     $O(1)$ & $\tilde O(k^2)$ & $O(k)$ & \cite{ourneurips2023}\tablefootnote{We remark that \cite{ourneurips2023} actually maintain a $O(1)$ approximate ``sparsifier'' of size $\tilde{O}(k)$ in $\tilde{O}(k)$ update time, and they need to run the static algorithm of~\cite{MettuP00} on top of this sparsifier after every update. This leads to an update time of $\tilde{O}(k^2)$ for the dynamic $k$-median and $k$-means problems.}\\
     $O\left(\epsilon^{-1}\right)$ & $\tilde O(k^{1 + \epsilon})$ & $\tilde O(k^\epsilon)$ & \cite{FOCS24kmedian} \\
     $O(1)$ & $\tilde{O}(k)$ & $O(\log^2 \Delta)$ & {\bf Our Result} \\
    \bottomrule
  \end{tabular}
  \caption{State-of-the-art for {\bf fully dynamic metric $k$-median}. The table for {\bf fully dynamic metric $k$-means} is identical, except that the approximation ratio of \cite{FOCS24kmedian} is $O\left(\epsilon^{-2}\right)$.}
  \label{tab:kmedian}
\end{table}

\medskip
\noindent {\bf Related Work.} In addition to metric $k$-median, two related clustering problems have been extensively studied in the dynamic setting: (i) {\em metric $k$-center}~\cite{ChanGS18,BateniEFHJMW23,LHRJ23,FOCS24kmedian,SkarlatosF25,simple-kcenter} and (ii) {\em metric facility location}~\cite{nips/Cohen-AddadHPSS19,nips/BhattacharyaLP22}. Both these problems are relatively well-understood by now. For example, it is known how to simultaneously achieve $O(1)$-approximation ratio, $\tilde{O}(1)$ recourse and $\tilde{O}(k)$ update time for dynamic metric $k$-center~\cite{simple-kcenter}. They have also been studied under special classes of metrics, such as Euclidean spaces~\cite{alenex/GoranciHLSS21,BateniEFHJMW23,icml/BhattacharyaGJQ24}, or shortest-path metrics in graphs undergoing edge-updates~\cite{soda/CrucianiFGNS24}.

There is another line of work on dynamic metric $k$-center and metric $k$-median, which considers the incremental (insertion only) setting, and achieves total recourse guarantees that are {\em sublinear} in the total number of updates~\cite{icml/LattanziV17,soda/FichtenbergerLN21}. However, the update times of these algorithms are large polynomials in $n$. \Cref{sec:intro:incremental} contains a detailed discussion on the algorithm of~\cite{soda/FichtenbergerLN21}.

\section{Technical Overview}

In \Cref{sec:intro:fullydynamic,sec:intro:incremental}, we summarize the technical contributions of two relevant papers~\cite{FOCS24kmedian,soda/FichtenbergerLN21}. We obtain our algorithm via carefully synthesizing the techniques from both these papers, along with some key, new insights of our own. In \Cref{sec:intro:challenges}, we outline the major technical challenges we face while trying to prove \Cref{main-theorem-2}, and how we overcome them.

\subsection{The Fully Dynamic Algorithm of \cite{FOCS24kmedian}}
\label{sec:intro:fullydynamic}

There are two main technical contributions in \cite{FOCS24kmedian}; we briefly review each of them below. 

\subsubsection{A Hierarchical Approach to Dynamic $k$-Median}
\label{sec:intro:hierarchy}

This approach allows the authors to obtain $O(\epsilon^{-1})$-approximation ratio with $\tilde{O}(k^{\epsilon})$ recourse, and works as follows. We maintain a hierarchy of nested subsets of centers $S_0 \supseteq \dots \supseteq S_{\ell + 1}$, where $\ell = 1/\epsilon$ and $|S_i| = k + \lfloor k^{1- i\epsilon}\rfloor$ for each $i \in [0, \ell + 1]$. We refer to $s_i := \lfloor k^{1- i\epsilon} \rfloor$ as the {\em slack} at {\em layer $i$} of the hierarchy. The set $S_{\ell + 1}$ has size exactly $k$ and is the  $k$-median solution maintained by the algorithm. We always maintain the following invariant.
\begin{invariant}
\label{inv:intro}
$\cost{S_0, \calP} = O(1) \cdot \OPT_k(\calP)$,  and $\cost{S_i, \calP} \leq \cost{S_{i-1}, \calP} + O(1) \cdot \OPT_k(\calP)$ for each $i \in [1, \ell]$. Here, $\OPT_k(\calP)$ is the optimal $k$-median objective w.r.t.~the input point-set $\calP$.
\end{invariant}

Given this invariant, we infer that $\cost{S_{\ell + 1}} = O(\ell) \cdot \OPT_k(\calP) = O(1/\epsilon) \cdot \OPT_k(\calP)$. Thus, the approximation guarantee is  proportional to the number of layers in this hierarchy. 

The hierarchy is maintained by periodically reconstructing each of the sets $S_i$.
To be more specific, each set $S_i$ (along with the sets $S_{i+1},\dots,S_{\ell + 1}$) is reconstructed from scratch every $s_i$ many updates, without modifying the sets $S_0,\dots,S_{i-1}$. In between these updates, the subset $S_i$ is maintained \emph{lazily}: Whenever a point $p$ is inserted into $\calP$, we set $S_i \leftarrow S_i \cup \{p\}$, and whenever a point is deleted from $\calP$, no changes are made to $S_i$.\footnote{We can afford to handle the {\em deletions} in this lazy manner if we consider the improper $k$-median problem, where we are allowed to open a center at a point that got deleted. See the discussion in the beginning of \Cref{sec:prelim}.} 

To analyze the recourse,
consider the solution $S_{\ell + 1}$ before and after an update during which $S_i$ is reconstructed. Let $S_{\ell + 1}'$ and $S_{\ell + 1}''$ denote the status of the solution before and after the update, respectively.
Observe that the total recourse incurred in the solution $S_{\ell + 1}$ during this update, i.e.~the value $|S_{\ell + 1}' \oplus S_{\ell + 1}''|$,
is $O(s_{i-1})$. This follows immediately from the fact that $S_{\ell + 1}'$ and $S_{\ell + 1}''$ are both subsets of size $k$ of the set $S_{i-1}$, which has size $\leq k + 2 s_{i-1}$.
We can amortize this recourse over the $s_i$ many lazy updates performed since the last time that $S_{i}$ was reconstructed. This implies that the amortized recourse of $S_{\ell + 1}$ which is caused by reconstructing $S_i$ is $O(s_{i-1}/s_i) = O(k^\epsilon)$. Summing over all $i \in [0, \ell + 1]$, we get an overall  amortized recourse bound of $O((\ell+2) \cdot k^\epsilon) = O(k^\epsilon/\epsilon)$.

\medskip
\noindent {\bf Barrier towards achieving $O(1)$-approximation with $\tilde{O}(1)$ recourse.} If we want to use this hierarchy to obtain a recourse of $\tilde O(1)$, then we need to ensure that $s_i = \tilde \Omega(s_{i-1})$ (i.e., we need the slacks at the layers to decrease by at most a $\tilde O(1)$ factor between the layers). If this is not the case, then the amortized recourse in $S_{\ell + 1}$ caused by reconstructing $S_i$ will be $\Omega (s_{i-1}/s_i) = \tilde \omega(1)$. Unfortunately, to have such a guarantee, the number of layers needs to be $\Omega(\log k / \log \log k)$; and since the approximation ratio of the algorithm 
is proportional to the number of layers (see \Cref{inv:intro}), this leads to an approximation ratio of $\omega(1)$. Thus, it is {\em not} at all clear if this approach can be used to obtain $O(1)$-approximation and $\tilde O(1)$ recourse simultaneously.

\subsubsection{Achieving $\tilde{O}(k^{1+\epsilon})$  Update Time via Randomized Local Search}\label{sec:introduce-rand-local-search}

 The second technical contribution in \cite{FOCS24kmedian} is to show that the hierarchy from \Cref{sec:intro:hierarchy} can be maintained in $\tilde{O}(k^{1+\epsilon})$ update time, using a specific type of {\em randomized local search}. 
 
 To be more specific, consider  a set of $n$ points $\calP$, a set of $k$ centers $\calU$, and an integer $s \in [1, k-1]$.
Suppose that we want to compute a subset $\calU' \subseteq \calU$ of $(k-s)$ centers, so as to minimize $\cost{\calU', \calP}$. In
\cite{FOCS24kmedian}, the authors present a $O(1)$-approximation algorithm for this problem using randomized local search, that runs in only $\tilde{O}(ns)$ time, assuming the algorithm has access to some ``auxiliary data structures'' to begin with (see \Cref{lm:localsearch}). 

Morally, the important message here is that the runtime of randomized local search (when given access to some auxiliary data structures) is proportional to the slack $s$, and independent of $k$. In \cite{FOCS24kmedian}, the authors call this procedure as a subroutine while reconstructing a set $S_i$ (along with the sets $S_{i+1},\dots,S_{\ell + 1}$) in the hierarchy from scratch, after every $s_i$ many updates. Although our algorithm does not use a hierarchical approach while bounding the approximation ratio and recourse, we use randomized local search to achieve fast update time (see \Cref{sec:intro:challenges}).

\subsection{The Incremental Algorithm of \cite{soda/FichtenbergerLN21}}
\label{sec:intro:incremental}

In \cite{soda/FichtenbergerLN21}, the authors  obtain $\tilde{O}(k)$ total recourse, while maintaining a $O(1)$-approximate $k$-median solution over a sequence of $n$ point insertions (starting from an empty input). Note that in this {\em incremental setting}, the total recourse is {\em sublinear} in the number of updates; this is achieved by using a technique known as {\em Myerson sketch}~\cite{Meyerson01}. Since it is not possible to achieve such a sublinear total recourse bound in the fully dynamic setting (the focus of our paper), in \Cref{thm:intro:soda} we summarize the main result of \cite{soda/FichtenbergerLN21} without invoking Myerson sketch. We emphasize that the update time  in \cite{soda/FichtenbergerLN21} is already prohibitively high (some large polynomial in $n$) for our purpose. Accordingly, to highlight the main ideas in the rest of this section, we will outline a variant of the algorithm in~\cite{soda/FichtenbergerLN21} with {\em exponential} update time.

\begin{theorem}[\cite{soda/FichtenbergerLN21}]
\label{thm:intro:soda}
 Suppose that the input $\calP$ undergoes a sequence of point-insertions. Then, we can maintain a $O(1)$-approximate $k$-median solution to $\calP$ with $\tilde{O}(1)$ amortized recourse.
\end{theorem}

A major technical insight in \cite{soda/FichtenbergerLN21} was to introduce the notion of {\em robust centers}. Informally, a set of centers $\calU$ is {\em robust} w.r.t.~a point-set $\calP$, iff each $u \in \calU$ is a good approximate $1$-median solution at every ``distance-scale'' w.r.t.~the points in $\calP$ that are sufficiently close to $u$. See \Cref{sec:robustcenters} for a formal definition. Below, we present our interpretation of the incremental algorithm in \cite{soda/FichtenbergerLN21}. We start with a key lemma summarizing an important property of robust centers. 

Consider any integer $0 \leq \ell \leq k$. We say that a set of $k$ centers $\calU$ is {\em maximally $\ell$-stable} w.r.t.~a point-set $\calP$ iff $\OPT_{k-\ell}^{\calU}(\calP) \leq c \cdot \cost{\calU, \calP}$ and $\OPT_{k-(\ell+1)}^{\calU}(\calP) > c \cdot \cost{\calU, \calP}$, where $c = 456000$ is an absolute constant, and $\OPT_{t}^{\calU}(\calP) := \min_{Z \subseteq \calU : |Z| \leq t} \cost{Z, \calP}$ is the objective  of the optimal $t$-median solution subject to the restriction that all the centers must be from the set $\calU$. This means that we can afford to remove $\ell$ centers from $\calU$ without increasing the objective value by more than a $O(1)$ factor, but {\em not} more than $\ell$ centers. We defer the proof of \Cref{lm:intro} to \Cref{sec:proof:lm:intro}.

\begin{lemma}[\cite{soda/FichtenbergerLN21}]
\label{lm:intro}
Consider any two  point-sets $\calP_{\init}$ and $\calP_{\final}$ with $\left|\calP_{\init} \oplus \calP_{\final}\right| \leq \ell + 1$, for  $\ell \in [0, k]$.\footnote{Here, the notation $\oplus$ denotes the symmetric difference between two sets.} W.r.t.~$\calP_{\init}$, let $\calU_{\init}$ be any set of $k$ centers that is robust and maximally $\ell$-stable.
Let $\calV$ be any set of $k$ centers such that $\cost{\calV, \calP_{\init}} \leq 18 \cdot \cost{\calU_{\init}, \calP_{\init}}$. Then, there is a set of $k$ centers $\calW^{\star}$, such that $\left| \calW^{\star} \oplus \calU_{\init}\right|  \leq 5\ell + 5$ and $\cost{\calW^{\star}, \calP_{\final}} \leq 3 \cdot \cost{\calV, \calP_{\final}}$.
\end{lemma}

 The algorithm works in $\tilde{O}(1)$ many {\em phases}, where each phase consists of a sequence of consecutive updates (only insertions) such that the optimal objective value, given by $\OPT_k(\calP)$, does not increase by more than a factor of $18$ within any given phase.
 The algorithm restarts whenever one phase terminates and the next phase begins, and computes a new $k$-median solution to the current input from scratch.
 Within each phase, the algorithm  incurs $\tilde{O}(n)$ total recourse, and this  implies the amortized recourse guarantee of $\tilde{O}(1)$. Each phase is further partitioned into {\em epochs}, as follows.

By induction hypothesis, we start an epoch with an $100$-approximate $k$-median solution $\calU_{\init}$ that is robust and also maximally $\ell$-stable, for some $\ell \in [0, k]$, w.r.t.~the current input $\calP_{\init}$. Let $\lambda_{\init} := \OPT_k(\calP_{\init})$ denote the optimal objective value at this point in time.
We then compute a subset $\calU \subseteq \calU_{\init}$ of $(k-\ell)$ centers, such that $\cost{\calU, \calP_{\init}} \leq c \cdot \cost{\calU_{\init}, \calP_{\init}} \leq  100c \cdot \OPT_{k}(\calP_{\init}) =  100c \cdot \lambda_{\init}$. The epoch  lasts for the next $(\ell+1)$ updates.

The algorithm {\em lazily} handles the first $\ell$ updates in the epoch, incurring a worst-case recourse of one per update: Whenever a point $p$ gets inserted into $\calP$, it sets $\calU \leftarrow \calU \cup \{ p\}$. Since initially $|\calU| = k-\ell$, the set $\calU$ never grows large enough to contain more than $k$ centers. Further, the objective  of the maintained solution $\calU$ does not increase due to these $\ell$ updates, and remains
\begin{equation}
\label{eq:intro:lazy}
\cost{\calU, \calP} \leq 100c \cdot \lambda_{\init} \leq  200c \cdot \OPT_{k}(\calP),
\end{equation}
 where the last inequality holds because $\OPT_{k}(\calP)$ is {\em almost monotone} as $\calP$ undergoes point-insertions (more precisely, it can decrease by at most a factor of $2$).

While handling the last (i.e., $(\ell+1)^{th}$) update in the epoch, our goal is to come up with a $k$-median solution $\calU_{\final}$ such that: (i) the induction hypothesis holds w.r.t.~$\calU_{\final}$ and $\calP_{\final}$ (where $\calP_{\final}$ is the state of the input at the end of the epoch), and (ii) the recourse remains small, i.e., $\left| \calU_{\final} \oplus \calU_{\init}\right| = O(\ell+1)$.
We can assume that $\OPT_{k}(\calP_{\final}) \leq 18 \cdot \lambda_{\init}$, for otherwise we would initiate a new phase at this point in time.

We find the set $\calU_{\final}$ as follows. Let $\calV$ be an optimal $k$-median solution w.r.t.~$\calP_{\final}$, so that:
\begin{equation}
\label{eq:intro:1}
\cost{\calV, \calP_\init} \leq \cost{\calV, \calP_{\final}} = \OPT_k(\calP_{\final}) \leq 18 \cdot \OPT_k(\calP_{\init}) \leq 18 \cdot \cost{\calU_{\init}, \calP_{\init}}.
\end{equation}

Applying \Cref{lm:intro}, we find a set of $k$ centers $\calW^\star$ such that
\begin{equation}
\label{eq:intro:3}
\left| \calW^\star \oplus \calU_{\init}\right| \leq 5\ell + 5 \text{ and }
\cost{\calW^\star, \calP_{\final}} \leq 3 \cdot \cost{\calV, \calP_{\final}} = 3 \cdot \OPT_k(\calP_{\final}).
\end{equation}
Next, we call a subroutine \Robustify$(\calW^\star)$ which returns a set of $k$ robust centers $\calU_{\final}$ such that $\cost{\calU_{\final}, \calP_{\final}} \leq (3/2) \cdot \cost{\calW^{\star}, \calP_{\final}} = (9/2) \cdot \OPT_k(\calP_{\final}) \leq 100 \cdot \OPT_k(\calP_{\final})$ (see \Cref{lem:cost-robustify-part-1}). This restores the induction hypothesis for the next epoch, w.r.t.~$\calU_{\final}$ and $\calP_{\final}$. We can show that the subroutine \Robustify \ works in such a manner that the step where we transition from $\calW^{\star}$ to $\calU_{\final}$ incurs at most $\tilde{O}(1)$ recourse, amortized over the entire sequence of updates within a phase (spanning across multiple epochs).  This implies \Cref{thm:intro:soda}.

\subsection{Our Approach}
\label{sec:intro:challenges}

At a high level, we achieve our result in two parts. First, we generalize the framework of \cite{soda/FichtenbergerLN21} to achieve $O(1)$-approximation and $\tilde{O}(1)$ recourse in the fully dynamic setting. 
Second, we use the randomized local search procedure
to implement our algorithm in $\tilde{O}(k)$ update time. In addition, both these parts require us to come up with important and new technical insights of our own. Below, we explain three significant challenges and outline how we overcome them. Challenge I and Challenge II refers to the first part (approximation and recourse guarantees), whereas Challenge III refers to the third part (update time guarantee).

\subsubsection{Challenge I: Double-sided Stability} 
In \Cref{sec:intro:incremental}, we crucially relied on the observation that the optimal $k$-median objective is (almost) monotonically increasing as more and more points get inserted into $\calP$. This allowed us to derive \Cref{eq:intro:lazy}, which guarantees that the maintained solution $\calU$ remains $O(1)$-approximate while we lazily handle the first $\ell$ updates within the epoch. This guarantee, however, breaks down in the fully dynamic setting: If points can get deleted from $\calP$, then within an epoch we might end up in a situation where $\OPT_k(\calP) \ll \lambda_{\init}$. To address this issue, we derive a new {\em double-sided stability} property in the fully dynamic setting (see \Cref{lem:double-sided-stability-part-1}).
Informally, this implies that if $\OPT_{k-\ell}(\calP_{\init}) = \Theta(1) \cdot \OPT_{k}(\calP_{\init})$ (which follows from the hypothesis  at the start of an epoch), then for some $\Theta(\ell) = r \leq \ell$ we have $\OPT_{k+r}(\calP_{\init}) = \Theta(1) \cdot \OPT_{k}(\calP_{\init})$. Furthermore, we have $\left| \calP \oplus \calP_{\init}\right| \leq r$ throughout the first $r$ updates in the epoch, which gives us: $\OPT_k(\calP) \geq  \OPT_{k+r}(\calP_{\init})$ (see \Cref{lem:lazy-updates-part-1}).
It follows that for the first $r$ updates in the epoch, we have $\OPT_k(\mathcal{P}) \geq \OPT_{k+r}(\calP_{\init}) = \Theta(1) \cdot \OPT_k(\mathcal{P}_{\init}) = \Theta(1) \cdot \lambda_{\init}$.
Accordingly, we truncate the epoch to last for only $r+1$ updates, and now we can rule out the scenario where $\OPT_k(\calP)$ drops significantly below $\lambda_{\init}$ during the epoch. But since $r = \Omega(\ell)$, the epoch remains sufficiently long, so that we can still manage to generalize the recourse analysis from \Cref{sec:intro:incremental}.

\subsubsection{Challenge II: Getting Rid of the Phases}
\label{sec:challenge2}
To derive \Cref{eq:intro:1}, we need to have $\OPT_k(\calP_{\final}) \leq 18 \cdot \OPT_k(\calP_{\init})$. This is precisely the reason why the algorithm in \cite{soda/FichtenbergerLN21} works in {\em phases}, so that $\OPT_k(\mathcal{P})$ increases by at most a $O(1)$ factor within each phase. Further, the analysis in \cite{soda/FichtenbergerLN21} crucially relies on showing that we incur $\tilde{O}(n)$ total recourse within a phase. This, combined with the fact that there are $\tilde{O}(1)$ phases overall in the incremental setting, implies the amortized recourse guarantee. 

From the preceding discussion, it becomes apparent that we cannot hope to extend such an argument in the fully dynamic setting, because it is not possible to argue that  we have at most $\tilde{O}(1)$ phases when the value of $\OPT_k(\calP)$ can fluctuate in either direction (go up or down)  over a sequence of fully dynamic updates.
To circumvent this obstacle, we make the following subtle but important change to the framework of~\cite{soda/FichtenbergerLN21}. Recall \Cref{eq:intro:3}.
Note that we can find the set $\calW^\star$ by solving the following computational task:\footnote{For now, we ignore any consideration about keeping the update time of our algorithm low, or even polynomial.} Compute the set $\calW^\star$ of $k$ centers, which minimizes the $k$-median objective {\bf w.r.t.~$\calP_{\final}$}, subject to the constraint that it can be obtained by adding/removing $\Theta(\ell+1)$ points in $\calU_{\init}$. Informally, in our algorithm, we replace this task by three separate (and new) tasks, which we perform one after another (see Step 4 in \Cref{sec:alg:describe}). 
\begin{itemize}
\item {\bf Task (i).} Find a set of $k+\Theta(\ell+1)$ centers $\calU^{\star}$, which minimizes the $(k+\ell+1)$-median objective {\bf w.r.t.~$\calP_{\init}$}, such that $\calU^{\star}$ is obtained by adding $\Theta(\ell+1)$ many centers to $\calU_{\init}$.
\item {\bf Task (ii).} Set $\calV^{\star} \leftarrow \calU^{\star} \cup (\calP_{\final} \setminus \calP_{\init})$.
\item {\bf Task (iii).} Find a set of $k$ centers $\calW^{\star}$, which minimizes the $k$-median objective {\bf w.r.t.~$\calP_{\final}$}, such that $\calW^\star$ is obtained by removing $\Theta(\ell+1)$ many centers from $\calV^{\star}$.
\end{itemize}
Thus, while adding centers in Task (i), we optimize the $(k+\ell+1)$-median objective w.r.t.~$\calP_{\init}$.
On the other hand, while removing centers in Task (iii), we optimize the $k$-median objective w.r.t.~$\calP_{\final}$.
This is in sharp contrast to the approach in \cite{soda/FichtenbergerLN21}, where we need to optimize the $k$-median objective w.r.t.~$\calP_{\final}$, {\em both} while adding centers and while removing centers. {\em Strikingly, we show that this modification allows us to get rid of the concept of phases altogether}. In particular, our algorithm can be cast in the classical {\em periodic recomputation} framework: We work in epochs. Within an epoch we handle the updates lazily, and at the end of the epoch we reinitialize our maintained solution so that we get ready to handle the next epoch. See \Cref{sec:recourse:part1} for details.

\subsubsection{Challenge III: Achieving Fast Update Time} As mentioned previously, the update time of the algorithm in~\cite{soda/FichtenbergerLN21} is prohibitively large. This occurs because of the following computationally expensive {\em steps} at the start and at the end of an epoch.\footnote{Note that it is straightforward to lazily handle the  updates within the epoch.} (1)
At the start of an epoch, \cite{soda/FichtenbergerLN21} computes the value of $\ell$, by solving an LP for the $(k-s)$-median problem with potential centers $\calU_\init$, for each $s \in [0,k-1]$.
(2) Next, to initialize the subset $\calU \subseteq \calU_\init$ of $(k-\ell)$ centers at the start of the epoch, \cite{soda/FichtenbergerLN21} again invokes an algorithm for the $(k-\ell)$-median problem  from scratch. (3) At the end of the epoch, \cite{soda/FichtenbergerLN21} solves another LP and applies a rounding procedure, to get an approximation of the desired set $\calW^\star$  (see \Cref{eq:intro:3}).
(4) Finally, at the end of the epoch, the call to the \Robustify \ subroutine also takes a prohibitively long time for our purpose.

In contrast, we take alternative approaches while performing the above steps. At the start of an epoch, we implement Steps (1) and (2) via  randomized local search (see 
\Cref{sec:introduce-rand-local-search}). For Step (3), we compute (an approximation) of the set $\calW^{\star}$ by solving the three tasks outlined in \Cref{sec:challenge2}.  One of our  contributions is to design a new algorithm for Task (i) that runs  in $\Tilde{O}(n\cdot(\ell + 1))$ time, assuming it has access to some auxiliary data structures (see \Cref{lm:developcenters}).
It is trivial to perform Task (ii).
For Task (iii), we again invoke the randomized local search procedure (see 
\Cref{sec:introduce-rand-local-search}).

Finally, for step (4),
 we need to efficiently implement the calls to $\Robustify(\cdot)$.  See \Cref{sec:robustify-implementation} for a more detailed discussion on this challenge, and how we overcome it (we defer the discussion to \Cref{sec:robustify-implementation} because it requires an understanding of the inner workings of the $\Robustify(\cdot)$ subroutine, which we have not described until now).

\section{Preliminaries}
\label{sec:prelim}

We now define some basic notations, and recall some relevant results from the existing literature. For the sake of completeness, we provide self-contained proofs for most of the lemmas stated in this section, but defer those proofs (since we do {\em not} take any credit for them) to \Cref{sec:missing:proof:prelim}.

 Consider a set of points $\ground$ and a distance function $d : \ground \times \ground \rightarrow \mathbb{R}^+$ that together form a metric space. The input to our dynamic algorithm is a subset $\calP \subseteq \ground$, which changes by means of updates. Let $n$ be an upper bound on the maximum size of $\calP$ throughout these updates. Each update either inserts a point $p \in \ground \setminus \calP$ into $\calP$, or deletes a point $p \in \calP$ from $\calP$. At all times, we have to maintain a set $\calU \subseteq \ground$ of at most $k$ ``centers'', so as to minimize the objective function
$$\cost{\calU, \calP} := \sum_{p \in \calP} d(p, \calU), \text{ where } d(p, \calU) := \min_{q \in \calU} d(p, q) \text{ is the distance from } p \text{ to the set } \calU.$$

We will refer to this  as the dynamic {\bf improper $k$-median} problem. Our goal is to design an algorithm for this problem that has: (1) good approximation ratio, (2) small update time, which is the time it takes to process an update in $\calP$, and (3) small recourse, which is the number of changes (point insertions/deletions) in the maintained solution $\calU$ per update. What makes this setting distinct from the standard  $k$-median problem is this: Here, we are allowed to open centers at locations that are not part of the current input, i.e., we can have $\calU \cap (\ground \setminus \calP) \neq \emptyset$. Nevertheless, in a black-box manner we can convert any dynamic algorithm for improper $k$-median into a dynamic algorithm for $k$-median, with essentially the same guarantees (see Lemma~\ref{lm:improper}). Accordingly, {\bf for the rest of this paper, we focus on designing a dynamic algorithm for improper $k$-median}.

\begin{lemma}[\cite{FOCS24kmedian}]
\label{lm:improper}
    Given an $\alpha$-approximation algorithm for  dynamic {\em improper} $k$-median, we can get a $2\alpha$-approximation algorithm for  dynamic $k$-median,  with an extra $O(1)$ multiplicative factor overhead in the recourse, and an extra $\tilde{O}(n)$ additive factor overhead in the update time.
\end{lemma}

\medskip
\noindent {\bf Remark.} At this point, the reader might get alarmed by the fact that \Cref{lm:improper} incurs an additive overhead of $\tilde{O}(n)$ update time. To assuage this concern, in \Cref{improve:updatetime:k}, we explain how to bring down the update time from $\tilde{O}(n)$ to $\tilde{O}(k)$, using standard sparsification techniques.

\subsection{Basic Notations}
\label{sec:notations}

By a simple scaling, we can assume that all of the distances in the metric space lie in the range $[1, \Delta]$, where $\Delta$ is the {\em aspect ratio}. Throughout the paper, we use the symbol $\ground$ to denote the underlying metric space with distance function $d : \ground \times \ground \rightarrow \mathbb{R}^{\geq 0}$, and $\calP \subseteq \ground$ to denote the current input.
For simplicity, for each set $S$ and element $p$, we denote $S\cup \{ p \}$ and $S \setminus \{p\}$ by $S+p$ and $S-p$ respectively.
For two sets of points $S$ and $S'$, we use $S \oplus S'$ to denote their symmetric difference.

For each $S \subseteq \ground$, we define $\pi_S: \ground \rightarrow S$ to be the projection function onto $S$, i.e., $\pi_S(x) := \arg\min_{s \in S} d(x,s)$,
breaking the ties arbitrarily. 
For each $\calU \subseteq \ground$ and $S \subseteq \calP$, we also define
$$\avcost{\calU, S} := \frac{\cost{\calU, S}}{|S|} = \frac{\sum_{p \in S} d(p, \calU)}{|S|} = \frac{\sum_{p \in S} \min_{q \in \calU} d(p, q)}{|S|}.$$
Consider any  subset of points $\calC \subseteq \ground$.
For every $k \geq 1$, we let $\OPT_{k}^{\calC}(\calP)$ denote the cost of the optimum $k$-median solution for $\calP \subseteq \ground$, where we can only open centers from $\calC$. Thus, we have $$\OPT^{\calC}_k(\calP) = \min\limits_{\substack{\calU \subseteq \calC, |\calU|\leq k}} \cost{\calU, \calP}.$$ 
When $\calC = \ground$, we slightly abuse the notation and write  $\OPT_k(\calP)$ instead of $\OPT_k^{\calC}(\calP)$.
Next, for each $\calU \subseteq \ground$ and $u \in \calU$, we define $C_u(\calU, \calP) := \{ p\in\calP \mid \pi_\calU(p) = u \}$ to be the set of points in $\calP$ that are ``assigned to'' the center $u$ in the solution $\calU$ (breaking ties arbitrarily). For each point $p \in \ground$ and value $r \geq 0$, let $\text{Ball}^\calP_r(p) := \{ q \in \calP \mid d(p,q) \leq r \}$ denote the ball of radius $r$ around $p$.
Note that if $p \in \ground \setminus \calP$, then $p$ itself is not part of the ball  $\text{Ball}^\calP_r(p)$.
Finally, throughout the paper we  use a sufficiently large constant parameter $\gamma = 4000$.

\subsection{Robust Centers}
\label{sec:robustcenters}

We will use the notion of {\em robust centers}~\cite{soda/FichtenbergerLN21}. Morally,  a point $p \in \ground$ is {\em $t$-robust} for an integer $t \geq 1$ iff it satisfies the following condition for all $i \in [1, t]$: Let $B_i = \text{Ball}_{10^i}^\calP(p)$, and consider any point $q \in \calP$ with $d(p, q) \ll 10^i$, i.e., $q$ is sufficiently close to $p$ compared to the radius of $B_i$. Then $\cost{p, S} \leq O(1) \cdot \cost{q, S}$  for all $B_i \subseteq S \subseteq \calP$. In words, the point $p$ is a good approximate $1$-median solution, compared to any other nearby point,  at every ``distance scale'' up to  $10^t$. 

The above definition, however, is too strong, in the sense that there might not exist any $t$-robust point under this definition.\footnote{For instance, it might be the case that there are (say) $\lambda$ many points in $B_i$, and all of them (except $p$) are in the exact same location as $q$, i.e., $d(p', q) = 0$ for all $p' \in B_i - p$. Then, the condition $\cost{p, B_i} \leq O(1) \cdot \cost{q, B_i}$ clearly does not hold, since the LHS is $(\lambda-1) \cdot d(p, q)$, whereas the RHS is only $d(p, q)$.}
Instead, the actual definition that we will use is stated below (see Definition~\ref{def:robust}), along with the relevant  properties that follow from it (see Lemma~\ref{lem:robustness-property-1-part-1} and Lemma~\ref{lem:robustness-property-2-part-1}). Conceptually, here the key difference from the idealized definition is that the balls $\{B_i\}_i$ are centered around different points $\{p_i\}_i$, with $p_0 = p$ and $d(p, p_i) \ll 10^i$ for all $i \in [1, t]$.

\begin{definition}
\label{def:robust}
    Let $(p_0,p_1,\ldots, p_t)$ be a sequence of $t+1$ points in $\ground$,  and let $B_i = \text{Ball}_{10^i}^\calP(p_i)$ for each $i \in [0,t]$.
    We refer to $(p_0,p_1,\ldots, p_t)$ as a {\em $t$-robust sequence} iff for every $i \in [1, t]$:
    \begin{eqnarray*}
        p_{i-1} =
        \begin{cases}
            p_i \quad & \text{if} \ 
            \avcost{p_i, B_i} \geq 10^i / 5; \\
            q_i \quad & \text{otherwise, where } q_i = \arg\min\limits_{q \in B_i + p_i} \cost{q, B_i}.
        \end{cases}
    \end{eqnarray*}    
     We say that a  point $p \in \ground$ is {\em $t$-robust} iff there exists a $t$-robust sequence $(p_0,p_1,\ldots, p_t)$ with $p_0 = p$.
\end{definition}

\begin{lemma}[\cite{soda/FichtenbergerLN21}]
\label{lem:robustness-property-1-part-1}
    Let $(p_0,p_1,\ldots , p_t)$ be a $t$-robust sequence and let $B_i = \text{Ball}^\calP_{10^i}(p_i)$ for all $i \in [0,t]$. Then, for all $i \in [1,t]$, we have
    $d(p_{i-1}, p_i) \leq  10^i/2$, $B_{i-1} \subseteq B_i$ and $d(p_0, p_i) \leq 10^i/2$.
\end{lemma}

\begin{lemma}[\cite{soda/FichtenbergerLN21}]
\label{lem:robustness-property-2-part-1}
    Let $(p_0,p_1,\ldots , p_t)$ be a $t$-robust sequence and let $B_i = \text{Ball}^\calP_{10^i}(p_i)$ for all $i \in [0,t]$. Then, for every $i \in[0,t]$ and every  $B_i \subseteq S \subseteq \calP $, we have
    $\cost{p_0, S} \leq (3/2) \cdot  \cost{p_i,S}$.
\end{lemma}

We next define the concept of a {\em robust collection of centers}.

\begin{definition}
\label{def:robuscenters}
A set of centers $\calW \subseteq \ground$ is {\em robust} iff the following holds for every $w \in \calW$:
    \begin{equation}\label{condition-robust-solution-part-1}
        w\  \text{is}\  t\text{-robust, where} \ t\ \text{is the smallest integer satisfying} \ 
        10^t \geq d(w,\calW - w )/ 200.
    \end{equation}
\end{definition}

Suppose that we have a set of centers $\calW \subseteq \ground$ that is {\em not} robust. A natural way to convert them into a robust set of centers is to call the subroutine \Robustify$(\calW)$, as described below.

\begin{algorithm}[H]
  \DontPrintSemicolon
  \While{there exist a $w \in \calW$ violating~(\ref{condition-robust-solution-part-1}) \label{line:whileloop:robustify}}{
    $t \gets $ Smallest integer satisfying $10^t \geq d(w,\calW-w)/100$. \label{line:constant:101} \label{line:constant}

    $w_0 \gets $ \MakeRbst$(w,t)$. \label{line:call:makerobust}

    $\calW \gets \calW - w + w_0$. \label{line:adjust}
    
    }
\caption{\Robustify$(\calW)$}
\label{alg:robustify}
\end{algorithm}

\begin{algorithm}[H]
  \DontPrintSemicolon
  $p_t \gets p$.

  \For{$i=t$ down to $1$}{
    $B_i \gets \text{Ball}_{10^i}^\calP(p_i)$. 

    \If{$\avcost{p_i, B_i} \geq 10^i/5$}{
        $p_{i-1} \gets p_i$.
    }\Else{
        $p_{i-1} \gets \arg\min\limits_{q \in B_i + p_i} \cost{q, B_i}$.
    }
  }
  \Return $p_0$.
\caption{\MakeRbst$(p, t)$}
\label{alg:makerobust}
\end{algorithm}

During a call to \MakeRbst$(p, t)$, we simply apply the rule from Definition~\ref{def:robust} to obtain a $t$-robust sequence $(p_0, p_1, \dots, p_t)$ with $p_t = p$, and then return the point $p_0$. Further, Line~\ref{line:constant} of the subroutine \Robustify$(\calW)$ considers the inequality $10^t \geq d(w,\calW-w)/100$, whereas~(\ref{condition-robust-solution-part-1}) refers to the inequality $10^t \geq d(w,\calW-w)/200$. This discrepancy in the constants on the right hand side ($100$ vs $200$) of these two inequalities is intentional, and plays a crucial role in deriving Lemma~\ref{lem:robustify-calls-once-part-1}.

\begin{lemma}[\cite{soda/FichtenbergerLN21}]\label{lem:robustify-calls-once-part-1}
    Consider any call to \Robustify$(\calW)$, and suppose that it sets $w_0 \leftarrow \MakeRbst(w, t)$ during some iteration of the {\bf while} loop. Then in subsequent iterations of the {\bf while} loop in the same call to \Robustify$(\calW)$, we will {\em not} make any call to  \MakeRbst$(w_0, \cdot)$.
\end{lemma}

\begin{lemma}[\cite{soda/FichtenbergerLN21}]\label{lem:cost-robustify-part-1}
    If $\calU$ is the output of \Robustify$(\calW)$, then
    $\cost{\calU, \calP} \leq \frac{3}{2} \cdot \cost{\calW, \calP}$.
\end{lemma}

\subsection{Well-Separated Pairs}
\label{sec:wellseparated}

We will also use the notion of a well-separated pair of points~\cite{soda/FichtenbergerLN21}, defined  as follows.

\begin{definition}
   Consider any $\calU, \calV \subseteq \ground$.
    A pair $(u,v) \in \calU \times \calV$ is {\em well-separated} w.r.t. $(\calU, \calV)$ iff 
    $$
        d(u, \calU - u) \geq \gamma \cdot  d(u, v) \text{ and }
        d(v, \calV - v) \geq \gamma \cdot   d(u, v).$$
\end{definition}

Using triangle inequality, it  is easy to verify that each point $u \in \calU$ either forms a well-separated pair with a unique $v \in \calV$, or it does not form a well-separated pair with any $v \in \calV$. The next lemma implies that if $\calU$ is robust, then we can replace every center $v \in \calV$ that is well-separated by its counterpart in $\calU$, and this will increase the cost of the solution $\calV$ by at most a constant factor.

\begin{lemma}[\cite{soda/FichtenbergerLN21}]\label{lem:cost-well-sep-part-1}
    Consider any two sets of centers $\calU, \calV \subseteq \ground$ such that $\calU$ is robust.  Then, for every well-separated pair $(u,v) \in \calU \times \calV$, we have 
    $\cost{u, C_v(\calV, \calP) } \leq 3 \cdot  \cost{v, C_v(\calV, \calP) }$.
\end{lemma}

\subsection{Projection Lemma and Lazy-Updates Lemma}
\label{sec:keylemmas}

We conclude by recalling two lemmas that are folklore in the literature on clustering~\cite{FOCS24kmedian}. 

Intuitively, the projection lemma says that if we have a set $\calU$ of more than $k$ centers, then the cost of the best possible $k$-median solution, subject to the constraint that all of the $k$ centers must be picked from $\calU$, is not too large compared to $\cost{\calU, \calP}$.

\begin{lemma}[Projection Lemma~\cite{FOCS24kmedian}]\label{lem:projection-lemma-part-1}
    Consider any set of centers $\calU \subseteq \ground$ of size $|\calU| \geq k$, where $k$ is a positive integer. Then we have $\OPT_{k}^{\calU}(\calP) \leq \cost{\calU, \calP} + 2 \cdot \OPT_{k}(\calP)$.
\end{lemma}

The lazy updates lemma, stated below, is derived from the following observation. Suppose that whenever a new point gets inserted into $\calP$, we create a  center at the position of the newly inserted point; and whenever a point gets deleted from $\calP$, we do not make any changes to the set of centers. Then this {\em lazy rule} for handling  updates ensures that the cost of the solution we maintain does not increase over time (although the solution might  consist of more than $k$ centers).

\begin{lemma}[Lazy-Updates Lemma~\cite{FOCS24kmedian}]\label{lem:lazy-updates-part-1}
    Consider any two sets of input points $\calP,\calP' \subseteq \ground$ such that $|\calP \oplus \calP'| \leq s$. Then for every $k \geq 1$, we have $\OPT_{k+s}(\calP') \leq \OPT_{k}(\calP)$.
\end{lemma}

\section{Two Key Lemmas}
\label{sec:building:block}

We now state two key lemmas that will be used in the design and analysis of our dynamic algorithm. We defer the formal proofs of these two lemmas to Section~\ref{sec:proof:building:block}.

\begin{lemma}[Double-Sided Stability Lemma]\label{lem:double-sided-stability-part-1} Consider any $r \in [0, k-1]$ and any $\eta \geq 1$. If 
    $\OPT_{k-r}(\calP) \leq \eta \cdot  \OPT_k(\calP)$, then we must have 
    $\OPT_{k}(\calP) \leq 4 \cdot \OPT_{k + \lfloor r / (12\eta) \rfloor }(\calP)$. 
\end{lemma}

To interpret Lemma~\ref{lem:double-sided-stability-part-1}, first note that $\OPT_k(\calP)$ is a monotonically non-increasing function of $k$, since the objective value can only drop if we open extra centers. Now, suppose there is a sufficiently large integer $r \in [0, k-1]$ such that $\OPT_{k-r}(\calP) \leq \Theta(1) \cdot \OPT_k(\calP)$. Then,  Lemma~\ref{lem:double-sided-stability-part-1} guarantees that $\OPT_k(\calP) \leq  \Theta(1) \cdot \OPT_{k+r'}(\calP)$ for some  integer $r' = \Theta(r)$. In other words, if the optimal objective remains stable as we {\em decrease} the number of centers by some additive $r$, then it also remains stable as we {\em increase} the number of centers by roughly the same amount. Conceptually, this holds because of two reasons: (i) the optimal objective of  the fractional version of the $k$-median problem (encoded by its standard LP-relaxation) is convex as a function of $k$, and (ii) the concerned LP-relaxation has $\Theta(1)$-integrality gap.

\begin{lemma}[Generalization of Lemma 7.3 in the arXiv version of~\cite{soda/FichtenbergerLN21}]\label{lem:generalize-lemma-7.3-in-FLNS21-part-1}
      Let $r \geq 0$ and $m \in [0, k]$. Consider any two sets of centers $\calU, \calV \subseteq \ground$ such that $|\calU| = k$ and $|\calV| = k+r$. 
    If the number of well-separated pairs w.r.t.~$(\calU, \calV)$ is $k - m$, then there exists a subset $\tilde{\calU} \subseteq \calU$ of size at most $k - \lfloor (m- r) / 4 \rfloor$ such that
    $\cost{\tilde{\calU}, \calP} \leq 6\gamma \cdot \left( \cost{\calU, \calP} + \cost{\calV, \calP} \right)$.
\end{lemma}

Intuitively, think of $\calU$ as the $k$-median solution maintained by our algorithm, and let $\calV$ be another set of $k+r$ centers such that $\cost{\calV, \calP} \leq \Theta(1) \cdot \OPT_k(\calP)$. 
The above lemma implies that if  $m \gg r$ (i.e., the number of well-separated pairs w.r.t.~$(\calU, \calV)$ is sufficiently small), then we can delete $\lfloor (m- r) / 4 \rfloor = \Omega(r)$ centers from $\calU$ without significantly increasing the objective  $\cost{\calU, \calP}$.

\section{Achieving $O(1)$ Approximation Ratio and $O(\log^2 \Delta)$ Recourse}
\label{sec:recourse:part}

In this section, we focus only on achieving good approximation ratio and recourse bounds. We prove the following theorem, {\em without  any concern for the update time of the algorithm. In particular, to keep the exposition as simple as possible, we present an algorithm with exponential update time}.

\begin{theorem}
\label{th:recourse:part1}
There is a deterministic $O(1)$-approximation algorithm for dynamic metric $k$-median with $O(\log^2 \Delta)$ recourse.
\end{theorem}

\subsection{Description of the Algorithm}
\label{sec:alg:describe}

Our algorithm works in {\bf epochs};  each epoch lasts for some consecutive updates in $\calP$. Let $\calU \subseteq \ground$ denote the maintained solution (set of $k$ centers). We satisfy the following invariant.

\begin{invariant}
\label{inv:start}
At the start of an epoch, the set $\calU$ is robust and $\cost{\calU, \calP} \leq 8 \cdot \OPT_k(\calP)$.
\end{invariant}

We now describe how our dynamic algorithm works in a given epoch, in four steps.

\medskip
\noindent {\bf Step 1: Determining the length of the epoch.} At the start of an epoch, we compute the maximum $\ell^{\star} \geq 0$ such that $\OPT_{k-\ell^{\star}}(\calP) \leq 54\gamma \cdot \OPT_{k}(\calP)$,\footnote{Recall $\gamma$ from \Cref{sec:notations}}
and set $\ell \leftarrow \lfloor \ell^{\star}/(12 \cdot 54 \gamma) \rfloor$. Since $\ell \leq \ell^{\star}$, it follows that $\OPT_{k-\ell}(\calP) \leq \OPT_{k-\ell^{\star}}(\calP)$. Thus, by setting $\eta = 54 \gamma$ and $r = \ell^{\star}$ in Lemma~\ref{lem:double-sided-stability-part-1}, at the start of the epoch we have:
\begin{equation}
\label{eq:length}
\frac{\OPT_{k-\ell}(\calP)}{54 \gamma} \leq \OPT_k(\calP) \leq 4 \cdot \OPT_{k+\ell}(\calP).
\end{equation}
The epoch will last for the next $\ell+1$ updates.\footnote{Note that we might very well have $\ell = 0$.} From now on, we will use the superscript $t \in [0, \ell+1]$ to denote the status of some object after our algorithm has finished processing the $t^{th}$ update in the epoch. For example, at the start of the epoch we have $\calP = \calP^{(0)}$.

\medskip 
\noindent {\bf Step 2: Preprocessing at the start of the epoch.} Let $\calU_{\init} \leftarrow \calU$ be the solution maintained by the algorithm after it finished processing the last update in the previous epoch. Before handling the  very first update in the current epoch, we initialize the maintained solution by setting
\begin{equation}
\label{eq:init:epoch}
\calU^{(0)} \leftarrow \arg \min_{\calU' \subseteq \calU_{\init} \, : \, |\calU'| = k-\ell} \cost{\calU', \calP^{(0)}}.
\end{equation}
Thus, at this point in time,  we have $\cost{\calU^{(0)}, \calP^{(0)}} = \OPT_{k-\ell}^{\calU_{\init}}\left(\calP^{(0)}\right) \leq \cost{\calU_{\init}, \calP^{(0)}} + 2 \cdot \OPT_{k-\ell}\left(\calP^{(0)}\right) \leq 8 \cdot \OPT_{k}\left( \calP^{(0)}\right) + 2 \cdot \OPT_{k-\ell}\left(\calP^{(0)}\right) \leq (32+432\gamma) \cdot \OPT_{k+\ell}\left(\calP^{(0)}\right)$, where the first inequality follows from Lemma~\ref{lem:projection-lemma-part-1}, the second inequality follows from Invariant~\ref{inv:start}, and the last inequality follows from~(\ref{eq:length}). To summarize, we get:
\begin{equation}
\label{eq:init:epoch:1}
\cost{\calU^{(0)}, \calP^{(0)}} \leq (32+432\gamma) \cdot \OPT_{k+\ell}\left(\calP^{(0)}\right) \text{ and } \left| \calU^{(0)} \right| \leq k-\ell.
\end{equation}
In words, before we deal with the very first update in the epoch, the maintained solution $\calU^{(0)}$ is a $(32+432\gamma) = \Theta(1)$-approximation of $\OPT_{k+\ell}(\po)$, and consists of at most $(k - \ell)$ centers. Both these properties will be crucially exploited while handling the updates within the epoch.

\medskip 
\noindent {\bf Step 3: Handling the updates within the epoch.} Consider the $t^{th}$ update in the epoch, for $t \in [1, \ell+1]$. We handle this update in a lazy manner, as follows. If the update involves the deletion of a point from $\calP$, then we do not change our maintained solution, and set $\calU^{(t)} \leftarrow \calU^{(t-1)}$. (The maintained solution remains valid, since we are considering the improper $k$-median problem). In contrast, if the update involves the insertion of a point $p$ into $\calP$, then we set $\calU^{(t)} \leftarrow \calU^{(t-1)} + p$.

It is easy to verify that this lazy way of dealing with an update does not increase the objective, and increases the number of centers in the maintained solution by at most one. Thus, we have $\cost{\calU^{(t)}, \calP^{(t)}} \leq \cost{\calU^{(t-1)}, \calP^{(t-1)}}$ and $\left| \calU^{(t)} \right| \leq \left| \calU^{(t-1)}\right| + 1$. From~(\ref{eq:init:epoch:1}), we now derive that
\begin{equation}
\label{eq:lazy:epoch}
\cost{\calU^{(t)}, \calP^{(t)}} \leq \cost{\calU^{(0)}, \calP^{(0)}} \text{ and } \left| \calU^{(t)} \right| \leq k-\ell+t, \text{ for all } t \in [1, \ell+1].
\end{equation}

\medskip
\noindent {\bf Step 4: Post-processing at the end of the epoch.} By~(\ref{eq:lazy:epoch}), the  set $\calU^{(t)}$ remains a valid solution for the improper $k$-median problem, for all $t \in [1, \ell]$.
After the very last update in the epoch, however, the set $\calU^{(\ell+1)}$ might have more than $k$ centers.
At this point in time, we do some post-processing, and compute another set $\calU_{\final} \subseteq \ground$ of at most $k$ centers (i.e., $|\calU_{\final}| \leq k$) that satisfies Invariant~\ref{inv:start}.
We then initiate the next epoch, with $\calU \leftarrow \calU_{\final}$ being the current solution.
The post-processing is done as follows.

We first add $O(\ell + 1)$ extra centers to the set $\calU_{\init}$, while minimizing the cost of the resulting solution w.r.t.~$\calP^{(0)}$. This gives us the set of centers $\calU^{\star}$. Note that $\left|\calU^{\star} \right| = k + O(\ell+1)$.
\begin{equation}
\label{eq:augment}
\calF^{\star} \leftarrow  \arg\min\limits_{\substack{\calF \subseteq \ground  : |\calF| \leq 2700\gamma \cdot (\ell+1)}} \cost{\calU_{\init} + \calF, \po}, \text{ and } \calU^{\star} \leftarrow \calU_{\init} + \calF^{\star}.
\end{equation}

We next add the newly inserted points within the epoch to the set of centers, so as to obtain the set $\calV^{\star}$. Since the epoch lasts for $\ell+1$ updates, we have $|\calV^{\star}| = k + O(\ell+1)$. Next, we identify the  subset $\calW^{\star} \subseteq \calV^{\star}$ of  $k$ centers that minimizes the $k$-median objective w.r.t.~$\calP^{(\ell+1)}$.
\begin{equation}
\label{eq:augment:1}
\calV^{\star} \leftarrow \calU^{\star} + \left( \calP^{(\ell+1)} - \calP^{(0)}\right), \text{ and }
\calW^\star \leftarrow \arg\min\limits_{\substack{\calW \subseteq \calV^{\star} \, : \, |\calW| = k}} \cost{\calW, \calP^{(\ell+1)}}.
\end{equation}

Finally, we call \Robustify$(\calW^{\star})$ and let $\calU_{\final}$ be the set of $k$ centers returned by this subroutine. Before starting the next epoch, we set $\calU \leftarrow \calU_{\final}$. 
\begin{equation}
\label{eq:augment:2}
\calU_{\final} \leftarrow \Robustify(\calW^{\star}).
\end{equation}

It is easy to verify that we always maintain a set $\calU \subseteq \ground$ of $k$ centers. In Section~\ref{sec:approx:part1}, we show that $\calU = \calU_{\final}$ satisfies Invariant~\ref{inv:start} at the end of Step 4, and analyze the approximation ratio of the overall algorithm. Finally, Section~\ref{sec:recourse:part1} bounds the recourse of the algorithm. We conclude this section with a corollary that will play an important role in our recourse analysis.

\begin{corollary}
\label{cor:recourse:easy}
We have $\left| \calW^{\star} \oplus \calU_{\init}\right| = O(\ell+1)$.
\end{corollary}

\begin{proof}
From~(\ref{eq:augment}) and~(\ref{eq:augment:1}), we infer that 
$\left| \calV^{\star} \oplus \calU_{\init} \right| \leq  \left| \calV^{\star} \oplus \calU^{\star} \right| + \left| \calU^{\star} \oplus \calU_{\init} \right| \leq \left| \calP^{(\ell+1)} - \calP^{(0)} \right| + \left| \calF^{\star} \right| \leq (2700\gamma \cdot (\ell+1)) + (\ell+1) = O(\ell+1)$.
Next, recall that $\left| \calV^{\star} \right| = k + O(\ell+1)$, and $\calW^{\star}$ is a subset of $\calV^{\star}$ of size $k$.
Thus, we get: $\left| \calW^{\star} \oplus \calU_{\init} \right| \leq 
\left| \calW^{\star} \oplus \calV^{\star} \right| + \left| \calV^{\star} \oplus \calU_{\init} \right| = O(\ell+1) + O(\ell+1) = O(\ell+1)$.
This concludes the proof.
\end{proof}

\subsection{Analyzing the Approximation Ratio}
\label{sec:approx:part1}

Consider any $t \in [0, \ell]$, and note that $\left| \calP^{(t)} \oplus \calP^{(0)} \right| \leq t \leq \ell$. Thus, by Lemma~\ref{lem:lazy-updates-part-1}, we have:
\begin{equation}
\label{eq:part1:approx:1}
\OPT_{k+\ell}\left( \calP^{(0)} \right) \leq \OPT_k\left( \calP^{(t)} \right).
\end{equation}
From~(\ref{eq:init:epoch:1}),~(\ref{eq:lazy:epoch}) and~(\ref{eq:part1:approx:1}), we now infer that:
\begin{equation}
\label{eq:part1:approx:2}
\cost{\calU^{(t)}, \calP^{(t)}} \leq (32+432\gamma) \cdot \OPT_k\left( \calP^{(t)} \right) \text{ and } \left| \calU^{(t)} \right| \leq k, \text{ for all } t \in [0, \ell].
\end{equation}
In other words, at all times within an epoch, the set $\calU^{(t)}$ maintained by our algorithm remains a valid $\Theta(1)$-approximate solution to the improper $k$-median problem on the current input $\calP^{(t)}$. It remains to show that the algorithm successfully restores Invariant~\ref{inv:start} when the epoch terminates after the $(\ell+1)^{th}$ update. Accordingly, we devote the rest of this section to the proof of Lemma~\ref{lm:restore:inv}.

\begin{lemma}
\label{lm:restore:inv}
At the end of Step 4 in Section~\ref{sec:alg:describe}, the set $\calU = \calU_{\final}$ satisfies Invariant~\ref{inv:start}.
\end{lemma}

The claim below bounds the cost of the solution $\calU^\star$ w.r.t.~the point-set $\calP^{(0)}$.

\begin{claim}
\label{cl:approx:key}
We have $\cost{\calU^{\star}, \calP^{(0)}} \leq 3 \cdot \OPT_{k+\ell+1}\left(\calP^{(0)}\right)$.
\end{claim}

Before proving Claim~\ref{cl:approx:key}, we explain how it implies Lemma~\ref{lm:restore:inv}. Towards this end, note that:
\begin{eqnarray}
\cost{\calU_{\final}, \calP^{(\ell+1)}} & \leq & \frac{3}{2} \cdot \cost{\calW^{\star}, \calP^{(\ell+1)}}  =  \frac{3}{2} \cdot \OPT_{k}^{\calV^{\star}}\left( \calP^{(\ell+1)}\right) \label{eq:derive:1} \\
& \leq & \frac{3}{2} \cdot \cost{\calV^{\star}, \calP^{(\ell+1)}} + 3 \cdot \OPT_k\left( \calP^{(\ell+1)} \right) \label{eq:derive:2} \\
& \leq & \frac{3}{2} \cdot \cost{\calU^{\star}, \calP^{(0)}} + 3 \cdot \OPT_k\left( \calP^{(\ell+1)} \right) \label{eq:derive:3} \\
& \leq & \frac{9}{2} \cdot \OPT_{k+\ell+1}\left( \calP^{(0)} \right) + 3 \cdot \OPT_k\left( \calP^{(\ell+1)} \right) \label{eq:derive:4} \\
& \leq & \frac{9}{2} \cdot \OPT_{k}\left( \calP^{(\ell+1)} \right) + 3 \cdot \OPT_k\left( \calP^{(\ell+1)} \right) \label{eq:derive:5}\\
& \leq & 8 \cdot \OPT_k\left( \calP^{(\ell+1)}\right). \label{eq:derive:6}
\end{eqnarray}
In the above derivation, the first step~(\ref{eq:derive:1}) follows from~(\ref{eq:augment:1}),~(\ref{eq:augment:2}) and  Lemma~\ref{lem:cost-robustify-part-1}. The second step~(\ref{eq:derive:2}) follows from Lemma~\ref{lem:projection-lemma-part-1}. The third step~(\ref{eq:derive:3}) follows from~(\ref{eq:augment:1}). The fourth step~(\ref{eq:derive:4}) follows from Claim~\ref{cl:approx:key}. The fifth step~(\ref{eq:derive:5}) follows from Lemma~\ref{lem:lazy-updates-part-1} and the observation that $\left| \calP^{(\ell+1)} \oplus \calP^{(0)} \right| \leq \ell+1$. From~(\ref{eq:derive:6}), we infer that at the start of the next epoch $\cost{\calU, \calP} \leq 8 \cdot \OPT_k(\calP)$, and the set $\calU$ is robust because of~(\ref{eq:augment:2}). This implies Lemma~\ref{lm:restore:inv}.

\subsubsection{Proof of Claim~\ref{cl:approx:key}}

Let $\calV \subseteq \ground$ be an optimal improper $(k+\ell+1)$-median solution for the point-set $\calP^{(0)}$, i.e., $\left| \calV \right| = k+\ell+1$ and $\cost{\calV, \calP^{(0)}} = \OPT_{k+\ell+1}\left( \calP^{(0)} \right)$. Let $m \in [0, k]$ be the unique integer such that there are $(k-m)$  well-separated pairs w.r.t.~$\left(\calU_{\init}, \calV\right)$. Let $\{ (u_1, v_1), (u_2, v_2), \ldots, (u_{k-m}, v_{k-m})\} \subseteq  \calU_{\init} \times \calV$ denote the collection of $(k-m)$ well-separated pairs w.r.t.~$\left( \calU_{\init}, \calV\right)$. Define the set $\tilde{\calF} := \calV \setminus \{v_1, \ldots, v_{k-m}\}$. It is easy to verify that: 
\begin{equation}
\label{eq:size}
\left| \tilde{\calF} \right| = \left| \calV \right| - (k-m) = m+\ell+1.
\end{equation}

\begin{claim}
\label{cl:approx:key:10}
We have $m \leq 2600\gamma \cdot (\ell+1)$.
\end{claim}

\begin{claim}
\label{cl:approx:key:11}
We have $\cost{\calU_{\init}+\tilde{\calF}, \calP^{(0)}} \leq 3 \cdot \OPT_{k+\ell+1}\left(\calP^{(0)}\right)$.
\end{claim}

By~(\ref{eq:size}),  Claim~\ref{cl:approx:key:10} and Claim~\ref{cl:approx:key:11}, there exists a set $\tilde{\calF}  \subseteq \ground$ of  $m+\ell+1 \leq 2700 \gamma \cdot (\ell+1)$ centers such that $\cost{\calU_{\init}+\tilde{\calF}, \calP^{(0)}} \leq 3 \cdot \OPT_{k+\ell+1}\left(\calP^{(0)}\right)$. Accordingly, from~(\ref{eq:augment}), we get $\cost{\calU_{\init}+\calF^{\star}, \calP^{(0)}} \leq  \cost{\calU_{\init}+\tilde{\calF}, \calP^{(0)}} \leq 3 \cdot \OPT_{k+\ell+1}\left(\calP^{(0)}\right)$, which implies Claim~\ref{cl:approx:key}.

It now remains to prove Claim~\ref{cl:approx:key:10} and Claim~\ref{cl:approx:key:11}.

\subsubsection*{Proof of Claim~\ref{cl:approx:key:10}}
We apply Lemma~\ref{lem:generalize-lemma-7.3-in-FLNS21-part-1}, by setting $r = \ell+1$, $\calU = \calU_{\init}$ and $\calP = \calP^{(0)}$. This implies the existence of a set $\tilde{\calU} \subseteq \calU_{\init}$ of at most $(k - b)$ centers, with $b = \lfloor (m- \ell-1)/4 \rfloor$, such that 
\begin{eqnarray*}
\cost{\tilde{\calU}, \calP^{(0)}} & \leq &  6 \gamma \cdot \left( \cost{\calU_{\init}, \calP^{(0)}} + \cost{\calV, \calP^{(0)}} \right) \\
& \leq & 6 \gamma \cdot \left( 8 \cdot \OPT_k\left( \calP^{(0)} \right) + \OPT_{k+\ell+1}\left(\calP^{(0)} \right) \right)  \leq  54 \gamma \cdot  \OPT_k\left( \calP^{(0)}\right).
\end{eqnarray*}
In the above derivation, the second inequality follows from Invariant~\ref{inv:start}, and the last inequality holds because $\OPT_{k+\ell+1}\left(\calP^{(0)}\right) \leq \OPT_k(\calP^{(0)})$. Since $\OPT_{k-b}\left( \calP^{(0)}\right) \leq \cost{\tilde{\calU}, \calP^{(0)}}$, we get
\begin{equation}
\label{eq:approx:key:1}
\OPT_{k-b}\left( \calP^{(0)} \right) \leq 54 \gamma \cdot  \OPT_k\left(\calP^{(0)} \right).
\end{equation}
Recall the way we defined  $\ell, \ell^{\star}$ at Step 1 in Section~\ref{sec:alg:describe}. From~(\ref{eq:approx:key:1}), it follows that $b \leq \ell^{\star}$. Since $b  \geq \frac{m- \ell-1}{4} - 1 = \frac{m-\ell-5}{4}$ and $\ell \geq \frac{\ell^{\star}}{12 \cdot 54 \gamma} - 1$, we get
$\frac{m- \ell-5}{4} \leq 12 \cdot 54\gamma \cdot (\ell + 1), \text{ and hence } m \leq (2592\gamma+1)(\ell+1)+4 \leq 2600\gamma \cdot (\ell+1)$.
This concludes the proof of the claim.

\subsubsection*{Proof of Claim~\ref{cl:approx:key:11}}
We define assignment $\sigma : \calP^{(0)} \rightarrow \calU_{\init} + \tilde{\calF}$, as follows.\footnote{Recall the notations $C_u(\calU, \calP)$ and $\pi_{\calU}(p)$ from Section~\ref{sec:notations}.} Consider any point $p \in \calP^{(0)}$. 
\begin{itemize}
\item If $p \in C_{v_i}\left( \calV, \calP^{(0)}\right)$ for some $i \in [1, k-m]$, then $\sigma(p) := u_i$. 
\item Otherwise,  $\sigma(p) := \pi_{\calV}(p)$.
\end{itemize}
In words, for every well-separated pair $(u_i, v_i) \in \calU_{\init} \times \calV$ all the points in the cluster of $v_i$ get reassigned to the center $u_i$, and the assignment of all other points remain unchanged (note that their assigned centers are present in $ \calU_{\init} + \tilde{\calF}$ as well as $\calV$).
Now, recall that the set of centers $\calU_{\init}$ is robust (see Invariant~\ref{inv:start}).
Hence, by applying Lemma~\ref{lem:cost-well-sep-part-1}, we infer that 
$$\cost{\calU_{\init}+\tilde{\calF}, \calP^{(0)}} \leq \sum_{p \in \calP^{(0)}} d(p, \sigma(p)) \leq 3 \cdot \cost{\calV, \calP^{(0)}}.$$
The claim follows since $\cost{\calV, \calP^{(0)}} = \OPT_{k+\ell+1}\left(\calP^{(0)}\right)$ by definition.

\subsection{Analyzing the Recourse}
\label{sec:recourse:part1}

Recall the description of our algorithm from Section~\ref{sec:alg:describe}, and consider a given epoch (say) $\mathcal{E}$ of length $(\ell+1)$. The total recourse incurred by the algorithm during this epoch is
\begin{equation}
\label{eq:recourse:1}
R_{\mathcal{E}} \leq \left| \calU_{\init} \oplus \calU^{(0)} \right| + \left( \sum_{t=1}^{\ell+1}  \left| \calU^{(t)} \oplus \calU^{(t-1)} \right| \right) + \left| \calU^{(\ell+1)} \oplus \calU_{\final} \right|.
\end{equation}
We will now bound each term on the right hand side of~(\ref{eq:recourse:1}). Towards this end, recall that $\calU^{(0)}$ is obtained by deleting $\ell$ centers from $\calU_{\init}$, and hence we have $\left| \calU_{\init} \oplus \calU^{(0)} \right| = \ell$. Next, it is easy to verify that in Step 3 (see Section~\ref{sec:alg:describe}) we incur a worst-case recourse of at most one per update.  Specifically, we have $\left| \calU^{(t)} \oplus \calU^{(t-1)} \right| \leq 1$ for all $t \in [1, \ell+1]$, and hence $\left|\calU^{(\ell+1)} \oplus \calU_{\init} \right| \leq \left|\calU^{(\ell+1)} \oplus \calU^{(0)} \right| + \left|\calU^{(0)} \oplus \calU_{\init} \right| \leq (\ell+1) + \ell = 2 \ell+1$. Thus, from~(\ref{eq:recourse:1}) we get:
\begin{eqnarray}
R_{\mathcal{E}} & \leq & \ell + (\ell+1) + \left| \calU^{(\ell+1)} \oplus \calU_{\final} \right| \nonumber  \\
& \leq & (2\ell+1) + \left|\calU^{(\ell+1)} \oplus \calU_{\init} \right| + \left|\calU_{\init} \oplus \calU_{\final} \right| \nonumber \\
& \leq & (4\ell+2) + \left|\calU_{\init} \oplus \calW^{\star} \right| + \left|\calW^{\star} \oplus \calU_{\final} \right| \nonumber \\
& = & O(\ell+1) + \left|\calW^{\star} \oplus \calU_{\final} \right|. \label{eq:recourse:3}
\end{eqnarray}
In the above derivation, the last step follows from Corollary~\ref{cor:recourse:easy}.

Since the epoch lasts for $\ell+1$ updates, the $O(\ell+1)$ term in the right hand side of~(\ref{eq:recourse:3}) contributes an amortized recourse of $O(1)$. Moreover, the term $\left|\calW^{\star} \oplus \calU_{\final} \right|$ is proportional to the number of  $\MakeRbst(\cdot \, , \cdot)$ calls made while computing $\calU_{\final} \leftarrow \Robustify\left( \calW^{\star} \right)$. So, the recourse of our algorithm is dominated by the number of calls made to the $\MakeRbst(\cdot \, , \cdot)$ subroutine, and Lemma~\ref{lm:makerobustcalls}  implies that our algorithm has an amortized recourse of $O(\log^2 \Delta)$.

\begin{lemma}
\label{lm:makerobustcalls}
The dynamic algorithm from Section~\ref{sec:alg:describe} makes at most $O(\log^2 \Delta)$ many calls to $\MakeRbst(\cdot \, , \cdot)$, amortized over the entire sequence of updates (spanning multiple epochs).
\end{lemma}

We devote the rest of this section towards proving Lemma~\ref{lm:makerobustcalls}.

\medskip
\noindent {\bf Contaminated vs Clean Centers.} Focus on the scenario at the start of a given epoch (see Section~\ref{sec:alg:describe}). By Invariant~\ref{inv:start}, the set $\calU_{\init}$ is robust w.r.t.~$\calP^{(0)}$. For each center $u \in \calU_{\init}$, we maintain an integer $t[u]$ such that: (i) $u$ is $t[u]$ robust w.r.t.~$\calP^{(0)}$, and (ii) $10^{t[u]} \geq d(u,\calU_{\init} - u )/ 200$. The existence of $t[u]$ is guaranteed by Definition~\ref{def:robuscenters}.\footnote{We use  $t[u]$ only for the sake of analyzing recourse. Here, the actual algorithm remains the same as in Section~\ref{sec:alg:describe}. But, we make use of these integers in \Cref{sec:updatetime:part1} to get fast update time.} Let $(p_0(u), p_1(u), \ldots, p_{t[u]}(u))$ be the $t[u]$-robust sequence w.r.t.~$\calP^{(0)}$ corresponding to $u$ (i.e., $u = p_0(u)$), and for each $i \in [1, t[u]]$, let $B_i(u) = \text{Ball}_{10^i}^{\calP^{(0)}}(p_i(u))$. Recall that by Lemma~\ref{lem:robustness-property-1-part-1}, we have $B_1(u) \subseteq B_2(u) \subseteq \cdots \subseteq B_{t[u]}(u)$.

Now, consider any $j \in [1, \ell+1]$, and let $q_j \in \ground$ denote the point being inserted/deleted in $\calP$ during the $j^{th}$ update in the epoch, i.e., $\calP^{(j)} \oplus \calP^{(j-1)} = \{q_j\}$.
We say that this $j^{th}$ update {\bf contaminates} a center $u \in \calU_{\init}$ iff
$d(q_j, p_{t[u]}(u)) \leq 10^{t[u]}$.
Intuitively, this means that if this $j^{th}$ update were taken into account while defining the balls $\{B_i(u)\}_i$ at the start of the epoch, then it might have impacted our decision to classify $u$ as being $t[u]$-robust at that time.
Furthermore,  we say that the center $u \in \calU_{\init}$ is {\bf clean} at the end of the epoch if {\em no} update $j \in [1, \ell+1]$  contaminated it (i.e.~$B_{t[u]}(u)$ and accordingly all of the balls $B_{i}(u)$ remain intact during the updates in this epoch);
otherwise we say that the center $u$ is {\bf contaminated} at the end of the epoch.

In the two claims below, we summarize a few key properties of clean and contaminated centers. We defer the proof of Claim~\ref{cl:cont} to Section~\ref{sec:cl:cont}.

\begin{claim}
    \label{cl:clean}
    If  $u \in \calU_{\init}$ is clean at the end of the epoch, then $u$ is $t[u]$-robust w.r.t.~$\calP^{(\ell+1)}$.
\end{claim}

\begin{proof}
Let $u \in \calU_{\init}$ be a clean center at the end of the epoch. So, during the epoch no point was inserted into/deleted from $B_{t[u]}$. As $B_1(u) \subseteq B_2(u) \subseteq \cdots \subseteq B_{t[u]}(u)$ by Lemma~\ref{lem:robustness-property-1-part-1}, this implies that during the epoch no point was inserted into/deleted from {\em any} of the balls  $B_1(u), B_2(u), \ldots, B_{t[u]}(u)$. Thus,  we have $\text{Ball}_{10^i}^{\calP^{(0)}}(p_i(u)) = \text{Ball}_{10^i}^{\calP^{(\ell+1)}}(p_i(u))$  for all  $i \in [0, \ell+1]$. 

The claim now follows from Definition~\ref{def:robust}.
\end{proof}

\begin{claim}
    \label{cl:cont}
    Each update during the epoch contaminates at most $O(\log \Delta)$ centers in $\calU_{\init}$.
\end{claim}

At the end of the epoch,  we set $\calU_{\final} \leftarrow \Robustify\left(\calW^{\star}\right)$. By Lemma~\ref{lem:robustify-calls-once-part-1},  the subroutine $\Robustify\left(\calW^{\star}\right)$ makes at most one call to $\MakeRbst(w, \cdot)$ for each point $w \in \calW^{\star}$, and {\em zero} call to $\MakeRbst(w, \cdot)$ for each point $w \in \ground \setminus \calW^{\star}$. 
Accordingly, we can partition the calls to $\MakeRbst(\cdot \, , \cdot)$ that are made by $\Robustify\left(\calW^{\star}\right)$ into the following three types.

\medskip
\noindent {\bf Type I.} A call to $\MakeRbst(w , \cdot)$ for some $w \in \calW^{\star} \setminus \calU_{\init}$. The total number of such  calls is at most $\left| \calW^{\star} \setminus \calU_{\init}\right| \leq \left| \calW^{\star} \oplus \calU_{\init}\right| = O(\ell+1)$ (see Corollary~\ref{cor:recourse:easy}). Since the epoch lasts for $(\ell+1)$ updates, the amortized number of Type I calls to $\MakeRbst(\cdot \, , \cdot)$, per update, is $O(1)$.

\medskip
\noindent {\bf Type II.} A call to $\MakeRbst(w , \cdot)$ for some $w \in \calW^{\star} \cap \calU_{\init}$ that is contaminated at the end of the epoch. By Claim~\ref{cl:cont}, the amortized number of such Type II $\MakeRbst(\cdot \, , \cdot)$ calls, per update, is $O(\log \Delta)$.

\medskip
\noindent {\bf Type III.} A call to $\MakeRbst(w , \cdot)$ for some $w \in \calW^{\star} \cap \calU_{\init}$ that is clean at the end of the epoch. Recall that the center $w$ was $t[w]$-robust w.r.t.~$\calP^{(0)}$ at the start of the epoch, and by Claim~\ref{cl:clean} it remains $t[w]$-robust w.r.t.~$\calP^{(\ell+1)}$ at the end of the epoch. Furthermore, note that if a center is $t'$-robust, then it is also $t''$-robust for all $t'' \leq t'$. Thus, at the end of the epoch, the call to $\MakeRbst(w , t)$ could have been made for only one reason: The subroutine $\Robustify\left(\calW^{\star}\right)$ wanted to ensure that $w$ was $t$-robust w.r.t.~$\calP^{(\ell+1)}$ for some $t > t[w]$, but it was not the case. Suppose that $w' \leftarrow \MakeRbst(w , t)$ was the center returned by this call to $\MakeRbst(\cdot \, , \cdot )$. Then at the end of this call, we set $t[w'] \leftarrow t$. Clearly, we have $t[w'] > t[w]$.

To bound the amortized number of Type III calls, we need to invoke a more ``global'' argument, that spans across multiple epochs. Consider a maximal ``chain'' of $j$ many Type III calls (possibly spanning across multiple different epochs),  in increasing order of time: 
\begin{eqnarray*}
&& w_1 \leftarrow \MakeRbst(w_0 , t[w_1]),  w_2 \leftarrow \MakeRbst(w_1 , t[w_2]), \\ 
&& \ldots, w_j \leftarrow \MakeRbst(w_{j-1} , t[w_j]).
\end{eqnarray*}
Note that the calls in the above chain can be interspersed with other Type I, Type II or Type III calls that are {\em not} part of the chain. Still, from the above discussion, we get $0\leq t[w_0] < t[w_1]  < \cdots < t[w_j] \leq \lceil \log \Delta \rceil$. So, the chain has length at most $O(\log \Delta)$. Also, for the chain to start in the first place, we must have had a Type I or Type II call to $\MakeRbst(\cdot \, , \cdot)$ which returned the center $w_0$. We can thus  ``charge'' the length (total number of Type III calls) in this chain to the Type I or Type II call that triggered it (by returning the center $w_0$). In other words, the total number of Type III calls ever made is at most $O(\log \Delta)$ times the total number of Type I plus Type II calls.  Since  the amortized number of Type I and Type II calls per update is $O(\log \Delta)$, the amortized number of Type III calls per update is $O(\log^2 \Delta)$. This concludes the proof of Lemma~\ref{lm:makerobustcalls}.

\subsubsection{Proof of Claim~\ref{cl:cont}}
\label{sec:cl:cont}

Assume $p \in \po \oplus \calP^{(\ell+1)}$ is an updated point during the epoch.
Let $u \in \calU_{\init}$ be contaminated by $p$, i.e.,~$d(p, p_{t[u]}(u)) \leq 10^{t[u]}$.
According to \Cref{lem:robustness-property-1-part-1}, $d(p_{t[u]}(u), u) = d(p_{t[u]}(u), p_0(u)) \leq 10^{t[u]}/2$, which concludes
\begin{equation}\label{eq:distance-p-to-contaminated-center}
   d(p, u) \leq d(p, p_{t[u]}(u)) + d(p_{t[u]}(u), u) \leq 10^{t[u]} + 10^{t[u]}/2 \leq 2 \cdot 10^{t[u]}. 
\end{equation}
Let $\{u_1,u_2,\ldots,u_\mu \} \subseteq \calU_{\init}$ be all centers contaminated by $p$ ordered in decreasing order by the time they were added to the main solution via a call to \MakeRbst$(\cdot,\cdot)$.
So, when $u_i$ is added to the main solution $\calU$, $u_{i+1}$ was already present in $\calU$, which concludes
$ 10^{t[u_i]} \leq d(u_i,u_{i+1})/10, $
for every $i \in [1, \mu-1]$ (by the choice of $t[u_i]$ in \MakeRbst \ at that time).
Hence,
$$ 10^{t[u_{i+1}]} \geq d(p, u_{i+1})/2 \geq (d(u_i, u_{i+1}) - d(p, u_i))/2 \geq 5 \cdot 10^{t[u_i]} - 10^{t[u_i]} = 4 \cdot 10^{t[u_i]}. $$
The first and last inequalities hold by (\ref{eq:distance-p-to-contaminated-center}) for $u = u_{i+1}, u_{i}$.
Finally, this concludes 
$t[u_\mu]> \cdots > t[u_{2}] > t[u_{1}]$. Since distances between any two point in the space is between $1$ and $\Delta$, we know $ 0 \leq t[u_{i}] \leq \lceil \log \Delta \rceil$ for each $i \in [1, \mu]$, which concludes $\mu = O(\log \Delta)$.

\section{Achieving $\tilde{O}(k)$ Update Time}
\label{sec:updatetime:part1}

We first outline how to implement the algorithm from Section~\ref{sec:alg:describe} in $\tilde{O}(n)$ update time, by incurring only a $O(1)$ multiplicative overhead in approximation ratio and recourse ($n$ is an upper bound on the size of the input $\calP \subseteq \ground$, throughout the sequence of updates).
In Section~\ref{improve:updatetime:k}, we show how to further improve the update time from $\tilde{O}(n)$ to $\tilde{O}(k)$ using standard sparsification techniques.

\medskip
\noindent {\bf Disclaimer.}
In this section, we often use asymptotic notations and informal arguments without proofs.
The reader can find the correct values of the parameters together with complete formal proofs for the statements of this section in \Cref{part:full} (full version).

\subsection{Auxiliary Data Structure and Randomized Local Search}
\label{sec:auxiliary}
Recall that $\calU \subseteq \ground$ is the solution (set of $k$ centers) maintained by our algorithm, and $\calP \subseteq \ground$ denotes the current input. For each point $p \in \calP \cup \calU$, we maintain a BST (balanced search tree) $\calT_p$ that stores the centers $\calU$ in increasing order of their distances to $p$. Note that after every change (insertion/deletion of a center)  in the set $\calU$, we can update all these BSTs in $\tilde{O}(\left| \calP \right| + \left| \calU \right|) = \tilde{O}(n+k) = \tilde{O}(n)$ time. Similarly, after the insertion/deletion of a point $p \in \calP$, we can construct/destroy the relevant BST $\calT_p$ in $\tilde{O}(\left| \calU \right|) = \tilde{O}(k) = \tilde{O}(n)$ time. In other words, if the algorithm from Section~\ref{sec:alg:describe} incurs a total recourse of $\tau$ while handling a sequence of $\mu$ updates, then we spend $\tilde{O}(n \cdot (\tau+\mu))$ total time  on maintaining this auxiliary data structure (collection of BSTs) over the same update-sequence. Since $\tau = O(\mu \cdot \log^2 \Delta)$ (see Section~\ref{sec:recourse:part1}), this incurs an amortized update time of $\tilde{O}(n)$, which is within our budget.

We will use the randomized local search algorithm, developed in~\cite{FOCS24kmedian} and summarized in the lemma below, as a crucial subroutine.

\begin{lemma}[Randomized Local Search~\cite{FOCS24kmedian}]
\label{lm:localsearch}
Suppose that we have access to the auxiliary data structure described above. Then, given any integer $s \in [0, k-1]$, in $\tilde{O}(ns)$ time we can find a subset $\calU^{\star} \subseteq \calU$ of $(k-s)$ centers such that $\cost{\calU^{\star}, \calP} \leq  O(1) \cdot \OPT_{k-s}^{\calU}\left(\calP\right)$.
\end{lemma}

\subsection{Implementing Our Dynamic Algorithm}

Henceforth, we focus on a given epoch of our dynamic algorithm that lasts for $(\ell+1)$ updates (see \Cref{sec:alg:describe}), and outline how to implement the algorithm in such a manner that it spends $\tilde{O}(n \cdot (\ell+1))$ total time during the whole epoch, except the call to \Robustify \ (see \Cref{eq:augment:2}).
For \Robustify, we provide an implementation that takes an amortized time of $\tilde{O}(n)$, over the entire sequence of updates (spanning multiple epochs).
This implies an overall amortized update time of $\tilde{O}(n)$.
Below, we first show how to implement each of Steps 1 - 4, as described in \Cref{sec:alg:describe}, one after another.
Then, we provide the implementation of \Robustify \ in \Cref{sec:robustify-implementation}.

\subsubsection*{Implementing Step 1.} 

Our task here is to compute an estimate of the value of $\ell^{\star}$. For each $i \in [0, \log_2 k]$ define $s_i := 2^i$, and let $s_{-1} := 0$. We now run a {\bf for} loop, as described below.

\medskip

\begin{algorithm}[H]
  \DontPrintSemicolon
  \For{$i=0$ to $\log_2 k$}{
    \label{line:estimate:1} Using \Cref{lm:localsearch}, in $\tilde{O}(ns_i)$ time compute a subset $\calU^{\star}_i \subseteq \calU_{\init}$ of $(k-s_i)$ centers,

    \label{line:estimate:2} that is a $O(1)$-approximation to $\OPT_{k-s_i}^{\calU_{\init}}\left( \calP^{(0)}\right)$. 

    \If{$\cost{\calU^{\star}_i, \calP^{(0)}}  \geq \Theta(\gamma) \cdot \cost{\calU_{\init}, \calP^{(0)}}$}{
        \Return $\hat{\ell} := s_{i-1}$.
    }
    }
\caption{Computing an estimate $\hat{\ell}$ of the value of $\ell^{\star}$.}
\end{algorithm}

After finding $\hat{\ell}$, we set  the length of the epoch to be $\ell + 1$ where $\ell \gets \left\lfloor \frac{\hat{\ell}}{12 \cdot \Theta(\gamma)} \right\rfloor$.
With some extra calculations, we can show that $\ell+1 = \Omega(\hat{\ell}+1) = \Omega(\ell^\star+1)$ and 
$$\frac{\OPT_{k-{\ell}}(\po)}{\Theta(\gamma)} \leq \OPT_{k}(\po) \leq 4 \cdot \OPT_{k+{\ell}}(\po). $$
The running time of this {\bf for} loop  is $\tilde{O}\left(n \sum_{j=0}^{i^\star}  s_j\right) = \tilde{O}(n s_{i^\star})$, where $i^{\star}$ is the index s.t.~$s_{i^{\star}-1} = \hat{\ell}$. Thus, we have $s_{i^\star} = O(\hat{\ell}) = O(\ell + 1)$, and hence we can implement Step 1  in $\tilde{O}(n\cdot (\ell+1))$ time.

\subsubsection*{Implementing Step 2.}
Instead of finding the optimum set of $(k-\ell)$ centers within $\calU_{\init}$, we approximate it using \Cref{lm:localsearch}:
We compute a set of $(k-\ell)$ centers $\uo \subseteq \calU_{\init}$ such that $\cost{\uo, \po} \leq O(1) \cdot \OPT_{k-\ell}^{\calU_{\init}}(\po)$. 
With the same arguments as before, we get $\cost{\uo, \po} \leq O(1) \cdot \OPT_{k + \ell}(\po)$.
Note that the running time for this step is also $\tilde{O}(n \cdot (\ell+1))$.

\subsubsection*{Implementing Step 3.}
Trivially, we can implement each of these updates in constant time.

\subsubsection*{Implementing Step 4.}
We first need to add $O(\ell+1)$ centers to $\calU_{\init}$, while minimizing the cost of the solution w.r.t.~$\po$.
We compute an approximation of $\calU^{\star}$, by setting $\calU = \calU_{\init}$, $\calP = \calP^{(0)}$ and $s = \Theta(\gamma \cdot (\ell+1))$ in \Cref{lm:developcenters} below (see \Cref{eq:augment}). This also takes $\tilde{O}(n \cdot (\ell+1))$ time. We defer the proof sketch of \Cref{lm:developcenters} to \Cref{sec:explain-developcenters}. 

\begin{lemma}
\label{lm:developcenters}
Suppose that we have access to the auxiliary data structure described above (see \Cref{sec:auxiliary}).
Then, given any integer $s \geq 1$, in $\tilde{O}(ns)$ time we can find a superset $\calU^{\star} \supseteq \calU$ of $(k+s)$ centers such that $\cost{\calU^{\star}, \calP} \leq O(1) \cdot \min\limits_{\calF \subseteq \ground: |\calF| \leq s} \cost{\calU + \calF, \calP} $.
\end{lemma}

At this stage, we compute $\calV^\star := \calU^\star + \left( \calP^{(\ell+1)} - \po \right)$ as in \Cref{eq:augment:1}, in only $\tilde{O}(\ell+1)$ time.
Next, we compute an approximation of  $\calW^\star$ (see \Cref{eq:augment:1}) using \Cref{lm:localsearch}, which again takes $\tilde{O}(n \cdot (\ell+1))$ time. It follows that $\cost{\calW^\star, \calP^{(\ell+1)}} \leq O(1) \cdot \OPT_{k}(\calP^{(\ell+1)})$.
Finally, we explain below how we implement the call to $\Robustify(\calW^\star)$ (see \Cref{eq:augment:2}).

\subsubsection{Implementing the calls to $\Robustify(\cdot)$ subroutine}\label{sec:robustify-implementation}

Recall \Cref{alg:robustify} and \Cref{alg:makerobust} from \Cref{sec:robustcenters}. In the static setting, there are known $O(1)$-approximation algorithms for $1$-median with $\tilde{O}(n)$ runtime~\cite{MettuP02}. Using any such $1$-median algorithm, it is relatively straightforward to (approximately) implement a call to $\MakeRbst(\cdot \, , \cdot)$ in $\tilde{O}(n)$ time (see \Cref{line:call:makerobust} in \Cref{alg:robustify}).  Using the auxiliary data structure (see \Cref{sec:auxiliary}), it is easy to implement each invocation of \Cref{line:constant:101} in \Cref{alg:robustify} in $\tilde{O}(1)$ time, whereas \Cref{line:adjust} in \Cref{alg:robustify} trivially takes $\tilde{O}(1)$ time. By \Cref{lm:makerobustcalls}, our dynamic algorithm makes $\tilde{O}(1)$ many amortized calls to $\MakeRbst(\cdot \, , \cdot)$, per update.  Thus, it follows that we spend  $\tilde{O}(n)$ amortized time per update on implementing \Cref{line:constant:101,line:call:makerobust,line:adjust} in \Cref{alg:robustify}.

There remains a significant challenge: We might have to iterate over $\left| \calW \right| = k$ centers  in \Cref{line:whileloop:robustify} in \Cref{alg:robustify}, before we find a center $w \in \calW$ that violates \Cref{condition-robust-solution-part-1}. Let us refer to this operation as ``testing a center $w$''; this occurs when we check whether $w$ violates \Cref{condition-robust-solution-part-1}. In other words, we  need to perform $O(k)$ many such tests in \Cref{line:whileloop:robustify}, before we  execute a $\MakeRbst(\cdot \, , \cdot)$ call in \Cref{line:call:makerobust} in \Cref{alg:robustify}. Note that because of \Cref{lm:makerobustcalls}, we perform $\tilde{O}(k)$ many tests, on average, per update. Thus,  if we could hypothetically perform each test in $\tilde{O}(1)$ time, then we would incur an additive overhead of $\tilde{O}(k)$ in our amortized update time, and everything would be fine.

The issue, however, is that testing a center can be prohibitively expensive. This is because \Cref{condition-robust-solution-part-1} consists of two conditions.
The second condition (which finds the value of $t$) requires us to know the value of $d(w, \calW - w)$, and this can indeed be implemented in $\tilde{O}(1)$ time using our auxiliary data structure (see \Cref{sec:auxiliary}).
The first condition asks us to check whether $w$ is $t$-robust, and there does not seem to be {\em any} efficient way in which we can implement this check (see \Cref{def:robust}).
To address this significant challenge, we  modify the execution of a call to $\Robustify(\calW^{\star})$ at the end of an epoch, as described below.

\subsubsection*{Modified version of the call to $\Robustify(\calW^{\star})$.} 

At the end of an epoch, the call to $\Robustify(\calW^{\star})$ is supposed to return  the set $\calU_{\final}$ (see \Cref{eq:augment:2}). We replace this call to $\Robustify(\calW^{\star})$ by the procedure  in \Cref{alg:modify:robustify} below.

To appreciate what \Cref{alg:modify:robustify} does, recall the recourse analysis in \Cref{sec:recourse:part1}; in particular, the distinction between {\em contaminated} vs {\em clean} centers in $\calU_{\init}$, and the three types of calls to the $\MakeRbst(\cdot \, , \cdot)$ subroutine. Note that in \Cref{line:type1,line:type2,line:type3} in \Cref{alg:modify:robustify}, the sets $\calW_1, \calW_2$ and $\calW_3$ respectively correspond to those centers $w$ that might {\em potentially} be the sources of Type I, Type II and Type III calls to $\MakeRbst(w, \cdot)$. We refer to the centers in $\calW_1, \calW_2$ and $\calW_3$ respectively as Type I, Type II and Type III centers.

The key difference between \Cref{alg:modify:robustify} and the previous version of  $\Robustify(\calW^{\star})$ is this:  In \Cref{alg:modify:robustify}, we proactively make calls to $\MakeRbst(w \, , \cdot)$  without even checking the first condition in \Cref{condition-robust-solution-part-1}, which was the main bottleneck in achieving efficient update time. To be more specific, we proactively call $\MakeRbst(w, \cdot)$ for every Type I and Type II center $w$ (see \Cref{line:mandatorycall}).
In contrast, for a Type III center $w$, we call $\MakeRbst(w, \cdot)$ whenever we observe that $t > t[w]$, where $t$ is the smallest integer satisfying $10^t \geq d(w, \calW - w)/200$ (see \Cref{line:prep,line:optionalcall}). Note that if $t \leq t[w]$ in \Cref{line:if}, then by \Cref{cl:clean} the center $w$ does {\em not} violate \Cref{condition-robust-solution-part-1}.

\medskip

\begin{algorithm}[H]
  \DontPrintSemicolon
   $\calW_1 \gets \calW^{\star} \setminus \calU_{\init}$ \label{line:type1} \;
   $\calW_2 \gets \{ w \in \calW^{\star} \cap \calU_{\init} : w \text{ is contaminated}\}$ \label{line:type2} \;
   $\calW_3 \gets \{ w \in \calW^{\star} \cap \calU_{\init} : w \text{ is clean} \}$ \label{line:type3} \;
   \tcp{The set $\calW^{\star}$ is partitioned into the subsets $\calW_1, \calW_2, \calW_3$}
  $\calW \leftarrow \calW^{\star}$ \label{line:type4}\;
  \For{each center $w \in \calW_1 \cup \calW_2$}{\label{for-loop-Type-I-and-II}
  $t \gets $ Smallest integer satisfying $10^t \geq d(w,\calW-w)/100$ \label{line:findt:1}\;
    $w_0 \gets $ \MakeRbst$(w,t)$ \label{line:mandatorycall}\; 
   $\calW \gets \calW - w + w_0$ \; 
    Save $t[w_0] \gets t$ and $p_{t[w_0]}(w_0) \gets w$ (see \Cref{sec:recourse:part1}) together with $w_0$ \;\label{end-for-loop-Type-I-and-II}
  }
  % $C \gets 0$ \;
  \While{true}{\label{line:while-loop}
   \For{each center $w \in \calW_3$}{\label{line-for-loop-W3}
  $t \gets $ Smallest integer satisfying $10^t \geq d(w,\calW-w)/200$ \label{line:prep} \;
  \If{$t > t[w]$}{\label{line:if}
    $t' \gets $ Smallest integer satisfying $10^t \geq d(w,\calW-w)/100$  \label{line:findt:2} \;
      $w_0 \gets $ \MakeRbst$(w,t')$ \label{line:optionalcall}\;
   $\calW \gets \calW - w + w_0$ \;
      Save $t[w_0] \gets t'$ and $p_{t[w_0]}(w_0) \gets w$ (see \Cref{sec:recourse:part1}) together with $w_0$ \;
   go back to Line~\ref{line:while-loop}.
  }
  }
  return $\calU_{\final} \gets \calW$ \label{line:end}.
  }
\caption{Modified version of the call to $\Robustify(\calW^{\star})$ at the end of an epoch.}
\label{alg:modify:robustify}
\end{algorithm}

\begin{lemma}[Informal]
\label{lm:informal:makerobuscalls}
\Cref{lm:makerobustcalls} continues to hold even after  the call to $\Robustify(\calW^{\star})$ at the end of every epoch is replaced by the procedure in \Cref{alg:modify:robustify}.
\end{lemma}

\begin{proof}(Sketch)
The lemma holds because the procedure in \Cref{alg:modify:robustify} is perfectly aligned with the recourse analysis in \Cref{sec:recourse:part1}. In other words, the recourse analysis  accounts for the scenario where we make proactive calls to $\MakeRbst(\cdot \, , \cdot)$, precisely in the manner specified by \Cref{alg:modify:robustify}. For example, the recourse analysis bounds the number of Type I and Type II calls to $\MakeRbst(\cdot \, , \cdot)$, by pretending that every Type I or Type II center makes such a call, regardless of whether or not it violates \Cref{condition-robust-solution-part-1}.
\end{proof}

\medskip
\noindent{\bf Bounding the Update Time.} First, note that we can implement each invocation of \Cref{line:type1,line:type2,line:type3,line:type4} in \Cref{alg:modify:robustify} in $\tilde{O}(k \cdot (\ell+1))$ time, which gets amortized over the length of the epoch, as  an epoch lasts for $(\ell+1)$ updates. Specifically, to compute the sets  $\calW_2$ and $\calW_3$,  we  iterate over all $w \in \calW^\star \cap \calU_{\init}$ and $q \in \po \oplus \calP^{(\ell+1)}$, and check whether the update involving $q$ contaminates $w$,  using the value $t[w]$ and the point $p_{t[w]}(w)$.
Moreover, \Cref{line:type1,line:type4} trivially take  $O(k)$ time. 

Using the auxiliary data structure (see \Cref{sec:auxiliary}), we can implement each invocation of \Cref{line:findt:1,line:prep,line:findt:2} in $\tilde{O}(1)$ time. Next, say that an iteration of the {\bf for} loop in \Cref{line-for-loop-W3} is {\em uninterrupted} if we find that $t \leq t[w]$ in \Cref{line:if}
% for all $w \in \calW_3$ 
(and accordingly do not execute any line within the {\bf if} block) and {\em interrupted} otherwise.
Each uninterrupted iteration of the {\bf for} loop takes $\tilde{O}(1)$ time. Furthermore, there can be at most 
 $|\calW_3| = O(k)$ many {\em consecutive} uninterrupted iterations of the {\bf for} loop in \Cref{line-for-loop-W3}: Any such chain of uninterrupted iterations is broken (i) either by an interrupted iteration, which involves a call to $\MakeRbst(w, t')$ in \Cref{line:optionalcall}, (ii) or by the termination of the procedure in \Cref{line:end} (this can happen only once in an epoch).

From the preceding discussion, it  follows that total time spent on the remaining lines in \Cref{alg:modify:robustify} is dominated by the time spent on the calls to  $\MakeRbst(\cdot \, , \cdot)$. Recall that at the start of \Cref{sec:robustify-implementation}, we have already explained that we spend $\tilde{O}(n)$ time to implement each call to $\MakeRbst(\cdot \, , \cdot)$. Finally, \Cref{lm:informal:makerobuscalls,lm:makerobustcalls} imply that we make $\tilde{O}(1)$ amortized calls to $\MakeRbst(\cdot \, , \cdot)$, per update. This gives us an overall amortized update time of $\tilde{O}(n)$.

\subsubsection{Algorithm for  \Cref{lm:developcenters}}\label{sec:explain-developcenters}

Since we have a set of fixed $k$ centers $\calU$ that must be contained in $\calU^\star$,
if we can treat these fixed centers as a single center instead of a set of $k$ centers, we might be able to reduce the problem to a $(s+1)$-median problem.
So, we contract all of the points $\calU$ to a single point $u^\star$ and define a new space $\calP' = (\calP - \calU) + u^\star$ with a new metric $d'$ as follows.
\begin{equation}\label{def:new-metric:part1}
    \begin{cases}
    d'(x, u^\star) := d(x, \calU) &\forall x \in \calP \\
    d'(x,y) := \min\{ d'(x,u^\star) + d'(y,u^\star), d(x,y)  \} &\forall x,y \in \calP
\end{cases}    
\end{equation}
This function $d'$ defines a metric on $\calP'$.
A simple way to verify this is that $d'$ is the metric derived from the shortest path in a complete graph where weight of edges between any two nodes in $\calU$ is zero and the weight of the other edges is the $d$ distance of their endpoints (you can find a complete proof in \Cref{dprime-is-metric}).
We also define weights for each $x \in \calP'$.
We define the weight of all $x \in \calP' - u^\star$ to be $1$ and the weight of $u^\star$ to be very large denoted by $\infty$ to enforce any constant approximate solution for $(s+1)$-median in $\calP'$ to open center $u^\star$.

In order to have access to metric $d'$, we can simply construct an oracle $D'$ to compute it.
This is because we have access to sorted distances of any $p \in \calP$ to $\calU$ through $\calT_p$ in our auxiliary data structure.
Combining with the definition of $d'$, it is easy to see that we can compute $d'(x,y)$ in $O(1)$ time for each $x,y \in \calP'$.

Next, we run the algorithm of \cite{MettuP02} for $(s+1)$-median problem on $\calP'$ w.r.t.~metric $d'$ to find a $\calF$ of size at most $(s+1)$ which is a constant approximation for $\OPT_{s+1}(\calP')$ in a total of $\tilde{O}(|\calP'|\cdot (s+1)) = \tilde{O}(n \cdot (\ell + 1))$ time.
Finally, we let $\calU^\star \gets \calU + (\calF - u^\star)$.
Note that $\calF - u^\star \subseteq \calP$ and its size is at most $s$ since $u^\star \in \calF$.
This set $\calU + (\calF - u^\star)$ is going to be a good solution w.r.t.~metric $d$ as well as $d'$ (see \Cref{lem:guarantee-develop-centers}).

\begin{algorithm}[H]
  \DontPrintSemicolon
  $\calP' \gets (\calP - \calU) + u^\star$.
 
  $w(u^\star) \gets \infty$, $w(x) \gets 1 \ \forall x \in \calP'-u^\star$.

   Consider $D'$ as an oracle to the metric $d'$ defined in \Cref{def:new-metric:part1}.

   Compute any $O(1)$ approximate solution $ \calF \subseteq \calP' $ for $(s+1)$-median problem on weighted metric space $\calP'$ in $\tilde{O}(|\calP'| \cdot (s+1))$ time using \cite{MettuP02} with access to distance oracle $D'$.

   \Return $ \gets \calU + (\calF - u^\star)$.

  \caption{Computing an approximation of the optimum $\calU^\star \supseteq \calU$ of size $(k+s)$.}
\end{algorithm}

\subsection{Improving the Update Time to $\tilde{O}(k)$}
\label{improve:updatetime:k}

\medskip
\noindent {\bf Extension to Weighted Case.}
First, we argue that our algorithm can be extended to the weighted case defined as follows.
We have a metric space $\calP$ with positive weights $w(p)$ for each $p \in \calP$ and distance function $d$. Denote this weighted space by $(\calP, w, d)$.
The cost of a collection $\calU$ of $k$ centers is defined as $\cost{\calU, \calP} = \sum_{p \in \calP} w(p) \cdot d(p, \calU)$ and subsequently $\avcost{\calU, S} = \cost{\calU, S}/(\sum_{p \in S} w(p))$ for all $S \subseteq \calP$.
We can extend our algorithm and all of the arguments for weighted case.

\medskip
\noindent {\bf Sparsification.}
Note that parameter $n$ in our algorithm is the maximum size of the space at any time during the total sequence of updates (it is \textbf{not} the size of the underlying ground metric space $\ground$).
As a result, if we make sure that the size of the space $\calP$ is at most $\tilde{O}(k)$ at any point in time, then the amortized running time of the algorithm would be $\tilde{O}(k)$ as desired.

We use the result of \cite{ourneurips2023} to sparsify the input.
A simple generalization of this result is presented in Section 10 of \cite{FOCS24kmedian} which extends this sparsifier to weighted metric spaces.
The authors provided an algorithm to sparsify the space to $\tilde{O}(k)$ weighted points. 
More precisely, given a dynamic metric space $(\calP, w, d)$ and parameter $k \in \mathbb{N}$, there is an algorithm that maintains a dynamic metric space $(\calP', w', d)$ in $\tilde{O}(k)$ amortized update time such that the following holds.
\begin{itemize}
    \item $\calP' \subseteq \calP$  and the size of $\calP'$ at any time is $\tilde{O}(k)$.
    
    \item A sequence of $T$ updates in $(\calP, w, d)$, leads to a sequence of $O(T)$ updates in $(\calP', w', d)$.\footnote{This guarantee follows from a slightly more refined analysis of the recourse of this sparsifier which is presented in \cite{simple-kcenter} (see Lemma 3.4 of \cite{simple-kcenter}).}

    \item Every $\alpha$ approximate solution to the $k$-median problem in the metric space $(\calP', w', d)$ is also a $O(\alpha)$ approximate solution to the $k$-median problem in the metric space
    $(\calP,w,d)$ with probability at least $1- \tilde{O}(1/n^c)$.
\end{itemize}
Now, suppose that we are given a sequence of updates $\sigma_1, \sigma_2, \ldots, \sigma_T$ in a dynamic metric space $(\calP,w,d)$.
Instead of feeding the metric space $(\calP,w,d)$ directly to our algorithm, we can
perform this dynamic sparsification to obtain a sequence of updates $\sigma'_1, \sigma'_2, \ldots, \sigma'_{T'}$ for a metric space $(\calP',w',d)$, where $T' = O(T)$, and feed the metric space $(\calP',w',d)$ to our dynamic algorithm instead.
Since our algorithm maintains a $O(1)$ approximate solution $\calU$ to the $k$-median problem in $(\calP',w',d)$, then $\calU$ is also a $O(1)$ approximate solution for the $k$-median problem in $(\calP,w,d)$ with high probability.

Since the length of the stream is multiplied by $O(1)$, we would have a multiplicative overhead of $O(1)$ in the amortized update time and recourse (amortized w.r.t.~the original input stream).

    \section{Missing Proofs From Section~\ref{sec:building:block}}
\label{sec:proof:building:block}

\subsection{Proof of Lemma~\ref{lem:double-sided-stability-part-1}: Double-Sided Stability Lemma}

Consider the LP relaxation for improper $k$-median problem on $\calP$ for each $k$ as follows.
\begin{align*}
    \min & \sum_{p \in \calP} \sum_{c \in \ground} d(c,p) \cdot x_{cp} & \\
    \text{s.t.}  & \quad x_{cp} \leq y_{c}  & \forall c \in \ground, p \in \calP 
    \\
    &\sum_{c \in \ground} x_{cp} \geq 1  & \forall p \in \calP \label{eq:all-open0}\\
    &\sum_{c \in \ground} y_c \leq k & \\
    &x_{cp}, y_c \geq 0 & \forall c \in \ground, p \in \calP
\end{align*}

\noindent
Denote the cost of the optimal fractional solution for this LP by $\FOPT_k$. Since space $\calP$ is fixed here, we denote $\OPT_k(\calP)$ by $\OPT_k$.
It is known that the integrality gap of this relaxation is at most $3$ \cite{CS11}. So, for every $k$ we have
\begin{equation}\label{eq:int-gap-part-1}
    \FOPT_k \leq \OPT_k \leq 3\cdot \FOPT_k.   
\end{equation}

\begin{claim}
    For every $k_1$, $k_2$ and $0 \leq \alpha, \beta \leq 1$ such that $\alpha + \beta = 1$, we have
    $$\FOPT_{\alpha k_1 + \beta k_2} \leq \alpha \cdot \FOPT_{k_1} + \beta \cdot \FOPT_{k_2}. $$
\end{claim}

\begin{proof}
    Assume optimal fractional solutions $(x^*_1,y^*_1)$ and $(x^*_2,y^*_2)$ to be an optimal solution for above LP for fractional $k_1$ and $k_2$-median problems respectively. It is easy to verify that their convex combination $(\alpha x^*_1 + \beta x^*_2, \alpha y^*_1 + \beta y^*_2)$ is a feasible solution for LP relaxation of $(\alpha k_1 + \beta k_2)$-median problem whose cost is $\alpha \ \FOPT_{k_1} + \beta \ \FOPT_{k_2}$ which concludes the claim.
\end{proof}

\noindent
Now, plug $k_1 = k-r$, $k_2 = k + r/(12\eta)$, $\alpha = 1/(12\eta)$ and $\beta = 1 - \alpha$ in the claim. We have
$$\alpha k_1 + \beta k_2 = \frac{1}{12\eta}(k-r) + \left(1-\frac{1}{12\eta}\right)\left(k+\frac{r}{12\eta}\right) = k - \frac{r}{(12\eta)^2} \leq k. $$
As a result,
$ \FOPT_k \leq \FOPT_{\alpha k_1 + \beta k_2} \leq \alpha \ \FOPT_{k_1} + \beta \ \FOPT_{k_2} $.
Together with \Cref{eq:int-gap-part-1}, we have
$ \OPT_k \leq 3\alpha\cdot \OPT_{k_1} + 3\beta \cdot \OPT_{k_2} $.
We also have the assumption that $ \OPT_{k_1} = \OPT_{k-r} \leq \eta \cdot \OPT_k$, which implies
$ \OPT_k \leq 3\alpha\eta\cdot \OPT_{k} + 3\beta\cdot \OPT_{k_2} $.
Finally,
$$ \OPT_k \leq \left( \frac{3\beta}{1-3\alpha\eta} \right) \cdot \OPT_{k_2} \leq 4 \cdot \OPT_{k_2} \leq 4\cdot  \OPT_{k +\lfloor r/(12\eta) \rfloor}. $$
The second inequality holds since
$$ \frac{3\beta}{1 - 3\alpha\eta} \leq \frac{3}{1 - 3\alpha\eta} = \frac{3}{1 - 1/4} = 4. $$

\subsection{Proof of Lemma~\ref{lem:generalize-lemma-7.3-in-FLNS21-part-1}}

Consider the standard LP relaxation for the weighted $k$-median problem as follows.
\begin{align*}
    \min & \sum_{p \in \calP} \sum_{u \in \calU}  d(u,p) \cdot x_{up} & \\
    \text{s.t.}  & \quad x_{up} \leq y_{u}  & \forall u \in \calU, p \in \calP 
    \\
    &\sum_{u \in \calU} x_{up} \geq 1  & \forall p \in \calP \\
    &\sum_{u\in\calU} y_u \leq k - (m- r) / 4 & \\
    &x_{up}, y_u \geq 0 & \forall u\in \calU, p \in \calP
\end{align*}
We consider the set of potential centers to open is $\calU$, and we want to open at most $k - (m- r) / 4 $ many centers.
Since the integrality gap of this LP is known to be at most $3$ \cite{CS11}, it suffices to show the existence of a fractional solution whose cost is at most $2\gamma\cdot\left( \cost{\calU,\calP} + \cost{\calV,\calP} \right)$
and this solution opens at most 
$ k - (m-r)/4 $ centers.
Now, we explain how to construct a fractional solution for this LP.

\noindent
\textbf{Fractional Opening of Centers.}
Consider the projection $\pi_\calU: \calV \rightarrow \calU$ function.
Assume $\calU = \calU_I + \calU_F$ is a partition of $\calU$ where $\calU_I$ contains those centers $u \in \mathcal{U}$ satisfying at least one of the following conditions:
\begin{itemize}
    \item $u$ forms a well-separated pair with one center in $\mathcal{V}$.
    \item $|\pi^{-1}_\calU(u)| \geq 2$.
\end{itemize}
For every $u\in \mathcal{U}_I$, set $y_u=1$ and for every $u\in \mathcal{U}_F$ set $y_u = 1/2$. 
% First, we show that
% $$\sum_{u\in \mathcal{U}} y_u \leq  k - (m- r) / 4.$$
Each center $u \in \calU$ that forms a well-separated pair with a center $v\in \calV$ has $|\pi^{-1}_\calU(u)| \geq 1$ since $u$ must be the closest center to $v$ in $\calU$.
Since the number of well-separated pairs is $k-m$, we have
\begin{equation*}
    k+r = |\calV| = \sum_{u\in \mathcal{U}} |\pi^{-1}_\calU(u)| \geq \sum_{u\in \mathcal{U}_I} | \pi^{-1}_\calU(u)| \geq 1 \cdot (k- m)   + 2\cdot\left(|\mathcal{U}_I| - (k-m)\right).
\end{equation*}
Hence,
$|\calU_I| \leq \frac{k+r + (k-m)}{2} = k - \frac{m-r}{2} $.
Finally, we conclude
\begin{eqnarray*}
   \sum_{u\in \calU} y_u = \sum_{u\in \calU_I} y_u + \sum_{u\in \calU_F} y_u 
   \leq 1 \cdot \left( k - \frac{m-r}{2} \right) + \frac{1}{2} \cdot \frac{m-r}{2}
   =
   k - \frac{m-r}{4}.
\end{eqnarray*}

\noindent
\textbf{Fractional Assignment of Points.}     
For every $p\in \calP$, assume $ v_p = \pi_{\calV}(p) $ is the closest center to $p$ in $\mathcal{V}$ and $u_p = \pi_\calU(v_p)$ is the closest center in $\calU$ to $v_p$.
We have three cases:
\begin{itemize}

\item If $y_{u_p} = 1$, then set $x_{u_pp} = 1$. The cost of this assignment would be $d(u_p, p)$.

\item If $y_{u_p} = 1/2$ and there is a center $u'_p\in \calU  - u_p$ such that $d(u_p, u'_p) \leq \gamma\cdot  d(u_p, v_p)$, then set $x_{u_pp} = x_{u_p' p} = 1/2$. Note that $u_p' \neq u_p$ which means this solution is feasible to the LP.
The cost of this assignment would be
\begin{eqnarray*}
    \frac{1}{2} \cdot \left(d(p,u_p) + d(p,u_p')\right) &\leq&  \frac{1}{2} \cdot \left(d(p,u_p) + d(p,u_p) + d(u'_p,u_p)\right)  \\
    &\leq& d(p,u_p) + \frac{\gamma}{2}\cdot d(u_p,v_p).
\end{eqnarray*}

\item If $y_{u_p} = 1/2$ and the previous case does not hold, then since $(u_p, v_p)$ is not a well-separated pair, there is a center $v'_p \in \calV - v_p$ such that $d(v_p, v'_p) \leq \gamma \cdot d(u_p, v_p)$. Let $u_p' = \pi_\calU(v'_p)$ and  set $x_{u_p p} = x_{u'_p p} = 1/2$. First, we show that $u'_p \neq u_p$. Since $y_{u_p} = 1/2$, we have $u_p \in \calU_F$ which concludes $|\pi^{-1}_\calU(u_p)| \leq 1$.
We also know that $\pi_\calU(v_p) = u_p$.
So, $v_p$ is the only center in $\calV$ mapped to $u_p$ which implies
$\pi_\calU(v'_p) \neq u_p$ or $u_p' \neq u_p$ (note that $v_p' \neq v_p$). So, point $p$ is assigned one unit to centers.
The cost of this assignment would be
\begin{eqnarray*}
    \frac{1}{2} \cdot \left(d(p, u_p) + d(p, u'_p)\right) 
    &\leq&
    \frac{1}{2} \cdot \left( d(p, u_p) + d(p, u_p) + d(u_p, v'_p)  + d(v'_p, u'_p)\right) \\
    &\leq&
    \frac{1}{2} \cdot \left( d(p, u_p) + d(p, u_p) + d(u_p, v'_p)  + d(v'_p, u_p)\right) \\
    &=&
     d(p, u_p)  + d(u_p,v_p')\\
    &\leq&
    d(p, u_p)  + d(u_p, v_p) + d(v_p,v_p') \\
    &\leq& d(p,u_p) + (\gamma+1)\cdot d(u_p,v_p).
\end{eqnarray*}
The second inequality is because of the choice of $u'_p = \pi_{\calU}(v_p')$ and the last inequality is because $d(v_p,v_p') \leq \gamma \cdot d(u_p,v_p)$.
\end{itemize} 

\noindent
\textbf{Bounding the Cost.}
Assume $u_p^* = \pi_\calU(p)$ for each $p \in \calP$.
We have
$$ d(u_p,v_p) \leq d(u_p^*,v_p) \leq d(u_p^*,p) + d(p,v_p). $$
The first inequality is by the choice of $u_p = \pi_\calU(v_p)$.
As a result,
\begin{equation}\label{eq:bound-on-d-up-vp-part-111}
    \sum_{p \in \calP} d(u_p,v_p)
    \leq
    \sum_{p \in \calP} \left( d(u^*_p,p) + d(p,v_p) \right)
    =
    \cost{\calU,\calP} + \cost{\calV,\calP}.
\end{equation}
In each of the cases of fractional assignments of points to centers, the cost of assigning a point $p \in C_v(\calV, \calP)$ is at most
$ d(p, u_p) + (\gamma+1) \cdot d(u_p, v_p)$.
As a result, the total cost of this assignment is upper bounded by
\begin{eqnarray*}
    \sum_{p\in P} \left( d(p,u_p) + (\gamma+1) \cdot d(u_p, v_p) \right)
    &\leq& 
     \sum_{p\in \calP} \left( d(p,v_p) + (\gamma+2) \cdot d(u_p, v_p) \right) \\
    &\leq& 
    \cost{\calV,\calP} +(\gamma+2)\cdot \left( \cost{\calU,\calP} + \cost{\calV,\calP} \right) \\
    &\leq& 
    2\gamma\cdot \left( \cost{\calU,\calP} + \cost{\calV,\calP} \right).
\end{eqnarray*}
The second inequality follows by \Cref{eq:bound-on-d-up-vp-part-111}.

    \newpage
    \part{Proofs That Follow From Prior Work}

\section{Missing Proofs From Section~\ref{sec:prelim}}
\label{sec:missing:proof:prelim}

\subsection{Proof of \Cref{lem:robustness-property-1-part-1}}

The lemma is obvious if $\avcost{p_j, B_j} \geq 10^j / 5$, because $p_j=p_{j-1}$ which means $d(p_{j-1},p_j)=0$.
Now, assume $\avcost{p_j, B_j} < 10^j / 5$. By the definition of a $t$-robust sequence, we know 
$\cost{p_{j-1}, B_j} \leq \cost{p_{j}, B_j}$. Hence,
\begin{eqnarray*}
    d(p_j,p_{j-1}) 
    &=& \frac{1}{|B_j|}\cdot \sum_{p \in B_j}  d(p_j,p_{j-1})
    \leq
    \frac{1}{|B_j|}\cdot  \sum_{p \in B_j} (d(p_j,p) + d(p_{j-1},p)) \\
    &=&
    \avcost{p_j, B_j} + \avcost{p_{j-1}, B_j} \\
    &\leq&
    2 \cdot \avcost{p_j, B_j}
    \leq
    \frac{2}{5}\cdot 10^j.
\end{eqnarray*}
For every $p \in B_{j-1}$ we have
$d(p,p_{j}) \leq d(p,p_{j-1}) + d(p_{j-1},p_j) \leq 10^{j-1} + \frac{2}{5}\cdot 10^j \leq 10^j,$
which implies $B_{j-1} \subseteq B_j$.
Finally,
$$ d(p_0,p_j) \leq \sum_{i=1}^j d(p_{i-1},p_i) \leq \frac{2}{5} \cdot \sum_{i=1}^j 10^i = \frac{2}{5}\cdot \frac{10^{j+1} - 2}{9} \leq 10^j / 2. $$

\subsection{Proof of \Cref{lem:robustness-property-2-part-1}}

Since $(p_0,p_1,\ldots , p_t)$ is $t$-robust, for every $1 \leq j \leq t$, we know
\begin{equation}\label{eq:pj-1-better-than-pj-part-111}
    \cost{p_{j}, B_j} \geq \cost{p_{j-1}, B_j}.
\end{equation}
Assume $q \in \calP \setminus B_j$. Then, by \Cref{lem:robustness-property-1-part-1},
$ d(q,p_j) \geq 10^j \geq 2\cdot  d(p_0,p_j) $.
Hence,
$$ \frac{3}{2}\cdot  d(q,p_j)
= d(q,p_j) + \frac{1}{2}\cdot d(q,p_j)
\geq d(q,p_j) + d(p_0,p_j)
\geq d(q,p_0),$$
which means
\begin{equation}\label{eq:dqpj-compare-to-dqp0-part-111}
    d(q,p_j) \geq \frac{2}{3}\cdot d(q,p_0).
\end{equation}
Finally, Since $ B_1 \subseteq \cdots \subseteq B_i \subseteq S$ (by \Cref{lem:robustness-property-1-part-1}) we can apply \Cref{eq:pj-1-better-than-pj-part-111} and \Cref{eq:dqpj-compare-to-dqp0-part-111} repeatedly to get
\begin{eqnarray*}
    \cost{p_i,S} &=& \cost{p_i, B_i} + \cost{p_i, S \setminus B_i}
    \geq
    \cost{p_{i-1}, B_i} + \frac{2}{3}\cdot  \cost{p_0, S \setminus B_i} \\
    &=&
    \cost{p_{i-1}, B_{i-1}} + \cost{p_{i-1}, B_i \setminus B_{i-1}} + \frac{2}{3} \cdot \cost{p_0, S \setminus B_i} \\
    &\geq&
    \cost{p_{i-2}, B_{i-1}} + \frac{2}{3}\  \cost{p_{0}, B_i \setminus B_{i-1}} + \frac{2}{3}\cdot  \cost{p_0, S \setminus B_i} \\
    &=&
    \cost{p_{i-2}, B_{i-1}} + \frac{2}{3}\cdot  \cost{p_0, S \setminus B_{i-1}} \\
    &\vdots& \\
    &\geq&
    \cost{p_{0}, B_{1}} + \frac{2}{3} \cdot \cost{p_0, S \setminus B_{1}}
    \geq
    \frac{2}{3}\cdot  \cost{p_0, S} .
\end{eqnarray*}

\subsection{Proof of \Cref{lem:robustify-calls-once-part-1}}

Assume a \MakeRbst \ call to a center $w \in \calW$ is made and it is replaced by $w_0$. Let $\calW_1$ be the set of centers just before this call to \MakeRbst \ is made on $w$.
At this point of time, $t[w_0]$ is the smallest integer satisfying $10^{t[w_0]} \geq d(w, \calW_1 - w)/100$ (note that $t[w_0] \geq 0$, otherwise the algorithm should not have called \MakeRbst \ duo to line 10 in the \Robustify).
For the sake of contradiction, assume that another call to \MakeRbst \ is made on $w_0$.
We assume that this pair $w$, $w_0$ is the first pair for which this occurs.
Let $w' = \pi_{\calW_1 - w}(w)$ which means $d(w,\calW_1 - w)=d(w,w')$.
Then,
$10^{t[w_0]} \geq d(w,w')/100$ and $10^{t[w_0]} < d(w,w')/10$ by definition of $t[w_0]$. So, by \Cref{lem:robustness-property-1}, we have
\begin{equation}\label{eq:dw0w-leq-dwwprime-over-20-part-1}
   d(w_0,w) \leq 10^{t[w_0]}/2 \leq d(w,w')/20. 
\end{equation}
Assume $\calW_2$ is the set of centers just after the call to \MakeRbst \ is made on $w_0$.
Since we assumed $w,w_0$ is the first pair that another call to \MakeRbst \ on $w_0$ is made, we have one of the two following cases.
\begin{itemize}
    \item $w' \in \calW_2$.
    In this case, we have
    \begin{equation*}
        d(w_0,\calW_2 - w_0) \leq d(w_0,w') \leq d(w_0,w) + d(w,w') \leq \frac{21}{20}\cdot d(w,w') \leq 2 \cdot d(w,w').
    \end{equation*}
    The last inequality holds by \Cref{eq:dw0w-leq-dwwprime-over-20-part-1}.
    This is in contradiction with the assumption that a call to \MakeRbst \ on $w_0$ is made, because $w_0$ is $t[w_0]$-robust where
    $10^{t[w_0]} \geq d(w,w')/100 \geq d(w_0,\calW_2-w_0)/200$. So, $w_0$ must not be selected to be added to $\calS$ in \Robustify.
    
    \item There is $ w'_0 \in \calW_2$ that replaced $w'$ by a single call to \MakeRbst.
    When \MakeRbst \ is called on $w'$, the algorithm picks integer $t[w_0']$ that satisfies
    $ 10^{t[w_0']} \leq d(w',w_0) / 10 $ (otherwise, it should have picked $t[w_0']-1$ instead of $t[w_0']$).
    Hence,
    \begin{equation}\label{eq:dw0primewprime-leq-dwwprime-over-10-part1}
       d(w_0',w') \leq 10^{t[w_0']}/2 \leq d(w',w_0) / 20 \leq \left( d(w_0,w) + d(w,w') \right)/20 \leq d(w,w')/10.  
    \end{equation}
    The first inequality holds by \Cref{lem:robustness-property-1} and the last inequality is \Cref{eq:dw0w-leq-dwwprime-over-20-part-1}.
    As a result,
    \begin{eqnarray*}
        d(w_0, \calW_2-w_0) &\leq& d(w_0,w_0')
        \leq d(w_0,w) + d(w,w') + d(w',w_0') \\
        &\leq& d(w,w')/20 + d(w,w') + d(w,w')/10 \\
        &=& \frac{23}{20}\cdot d(w,w')
        \leq 2 \cdot d(w,w').
    \end{eqnarray*}
    The last inequality holds by \Cref{eq:dw0w-leq-dwwprime-over-20-part-1} and \Cref{eq:dw0primewprime-leq-dwwprime-over-10-part1}.
    This leads to the same contradiction.
\end{itemize}

\subsection{Proof of \Cref{lem:cost-robustify-part-1}}

By \Cref{lem:robustify-calls-once-part-1}, \Robustify \ calls \MakeRbst \ on each center at most once. 
Let $r$ be the number of centers in $\calW$ for which a call to \MakeRbst \ has happened.
Assume
$\calW = \{w_1,\ldots, w_k\}$ is an ordering by the time that \MakeRbst \ is called on the centers $w_1,w_2,\ldots,w_r$ and the last elements are ordered arbitrarily.
Assume also that $\calU =  \{w'_1,\ldots, w'_k\}$ where $w_i'$ is obtained by the call \MakeRbst \ on $w_i$ for each $1 \leq i\leq r$ and $w_i' = w_i$ for each $r+1 \leq i \leq k$.

\noindent
For every $1 \leq j \leq r$, integer $t[w_j']$ is the smallest integer satisfying
$$ 10^{t[w_j']} \geq d(w_j,\{ w_1',\ldots,w_{j-1}',w_{j+1},\ldots,w_k \})/100. $$
Hence,
\begin{equation}\label{eq:proof-cost-robustify-dwjwiprime-part-1}
    10^{t[w_j']} \leq d(w_j,w_i')/10 \quad \forall i < j
\end{equation}
and
\begin{equation}\label{eq:proof-cost-robustify-dwjwi-part-1}
    10^{t[w_j']} \leq d(w_j,w_i)/10 \quad \forall i > j.
\end{equation}
Now, for every $1 \leq i < j \leq r$, we conclude
\begin{equation*}
    10\cdot 10^{t[w_j']} \leq d(w_i',w_j) \leq  d(w_i,w_j) + d(w_i,w_i') \leq d(w_i,w_j) + 10^{t[w_i']}/2 \leq \frac{11}{10} \cdot d(w_i,w_j).
\end{equation*}
The first inequality holds by \Cref{eq:proof-cost-robustify-dwjwiprime-part-1}, the third one holds by \Cref{lem:robustness-property-1-part-1}  (note that $w_i' = $ \MakeRbst$(w_i)$) and the last one holds by \Cref{eq:proof-cost-robustify-dwjwi-part-1} (by exchanging $i$ and $j$ in this equation since $i < j$).
This concludes for each $i < j$, we have $d(w_i,w_j) > 2 \cdot 10^{t[w_j']}$.
The same inequality holds for each $i>j$ by \Cref{eq:proof-cost-robustify-dwjwi-part-1}.
As a result,
$d(w_i,w_j) > 2 \cdot 10^{t[w_j']}$ for all $i \neq j$.
This concludes
$$\text{Ball}_{10^{t[w_j']}}^\calP(w_j) \subseteq C_{w_j}(\calW, \calP), $$
for every $1 \leq j \leq r$.
Then, we can apply \Cref{lem:robustness-property-2-part-1}, we have
\begin{equation}\label{eq:proof-cost-robustify-cost-wjprime-less-3-over-2-part-1}
   \cost{w_j',C_{w_j}(\calW, \calP)} \leq \frac{3}{2}\cdot  \cost{w_j, C_{w_j}(\calW, \calP)} \quad \forall 1 \leq j \leq r. 
\end{equation}
Finally, since $C_{w_j}(\calW, \calP)$ for $1 \leq j \leq k$ form a partition of $\calP$, we conclude
\begin{eqnarray*}
    \cost{\calU, \calP} &=& 
    \sum_{j=1}^k \cost{\calU, C_{w_j}(\calW, \calP)}
    \leq
    \sum_{j=1}^k \cost{w_j', C_{w_j}(\calW, \calP)} \\
    &=&
    \sum_{j=1}^r \cost{w_j', C_{w_j}(\calW, \calP)} + \sum_{j=r+1}^k \cost{w_j', C_{w_j}(\calW, \calP)} \\
    &\leq&
    \frac{3}{2}\cdot  \sum_{j=1}^r \cost{w_j, C_{w_j}(\calW, \calP)} + \sum_{j=r+1}^k \cost{w_j, C_{w_j}(\calW, \calP)} \\
    &\leq&
    \frac{3}{2}\cdot  \sum_{j=1}^k \cost{w_j, C_{w_j}(\calW, \calP)}
    =\frac{3}{2}\cdot  \cost{\calW, \calP}.
\end{eqnarray*}
The second inequality holds by \Cref{eq:proof-cost-robustify-cost-wjprime-less-3-over-2-part-1}.

\subsection{Proof of \Cref{lem:cost-well-sep-part-1}}

Since, the point-set $\calP$ is fixed, denote $C_v(\calV,\calP)$ by $C_v$ for simplicity.
If $d(u,v)=0$, the statement is trivial.
Now assume $d(u,v) \geq 1$.
Since $\calU$ is robust, then $u$ is $t$-robust where $t$ is the smallest integer such that
$ 10^t \geq d(u,\calU-u)/200 \geq (\gamma/200)\cdot d(u,v) = 20\cdot d(u,v) $.
The second inequality is because $(u,v)$ is well-separated pair and the equality is because $\gamma = 4000$.
Assume $t^*$ is the integer satisfying
\begin{equation}\label{eq:definition-tstar-part-1}
   20\cdot d(u,v) \leq 10^{t^*} < 10\cdot 20\cdot d(u,v). 
\end{equation}
Note that $t^* \geq 1$. Since $t^* \leq t$ and $u$ is $t$-robust, then it is also $t^*$-robust which means there is a $t^*$-robust sequence 
$(p_0,p_1,\ldots,p_{t^*})$ such that $p_0=u$.
For Simplicity, assume $B_i = \text{Ball}_{10^i}^\calP(p_i)$ for $0 \leq i \leq t^*$.

\begin{claim}\label{claim:Bi-subseteq-Cv-part-1}
    For each $0 \leq i \leq t^*$, we have $B_i \subseteq C_v$.
\end{claim}

\begin{proof}
    Assume $q \in B_i$. Then,
    \begin{eqnarray*}
        d(q,v) &\leq& d(q,p_i) + d(p_i,u) + d(u,v)
        = d(q,p_i) + d(p_i,p_0) + d(u,v) \\
        &\leq& 10^i + 10^i/2 + 10^{t^*}/20
        \leq 2 \cdot 10^{t^*}.
    \end{eqnarray*}
    For the second inequality, we used $q \in B_i$, \Cref{lem:robustness-property-1-part-1} and \Cref{eq:definition-tstar-part-1}.
    Now, for every $v' \in \calV - v$ we have
    $ d(v,v') \geq d(v,\calV-v) \geq 4\cdot10^3\cdot  d(u,v) > 10^{t^*+1} \geq 2\cdot d(q,v) $.
    The second inequality is by $(u,v)$ being well-separated (and $\gamma=4000$), the third one is by \Cref{eq:definition-tstar-part-1}.
    This concludes $d(q,v') \geq d(v,v') - d(q,v) > d(q,v)$ which means $v$ is the closest center to $q$ in $\calV$. So, $q \in B_i$ and finally $B_i \subseteq C_v$.
\end{proof}

\noindent
Now, according to the definition of $t$-robust sequence we have two cases.

\noindent
\textbf{Case 1:} $\avcost{p_{t^*}, B_{t^*}} \geq 10^{t^*}/5$. In this case, we have $p_{t^*-1} = p_{t^*}$ and
\begin{eqnarray*}
    d(v,p_{t^*}) &=& d(v,p_{t^*-1})
    \leq d(v,p_0) + d(p_0,p_{t^*-1})
    = d(v,u) + d(p_0,p_{t^*-1}) \\
    &\leq& 10^{t^*}/20 + 10^{t^*-1}/2
    = 10^{t^*}/10.
\end{eqnarray*}
The second inequality holds by \Cref{lem:robustness-property-1-part-1} and \Cref{eq:definition-tstar-part-1}.
So,
\begin{equation}\label{eq:d-v-ptstar-part-1}
    d(v,p_{t^*}) \leq 10^{t^*}/10.
\end{equation}
For every $q \in C_v - B_{t^*}$, we have
$ d(q,v) \geq d(p_{t^*},q) - d(v,p_{t^*}) \geq 10^{t^*} - d(v,p_{t^*}) \geq d(v, p_{t^*}) $.
The second inequality follows by $q \in C_v-B_{t^*}$ and the last inequality follows by
\Cref{eq:d-v-ptstar-part-1}.
So, 
$ d(q,v) \geq (d(q,v) + d(v,p_{t^*}))/2 \geq d(q,p_{t^*})/2 $,
which implies
\begin{equation}\label{eq:cost-v-outside-part-1}
    \cost{v, C_v - B_{t^*}} \geq \frac{\cost{p_{t^*}, C_v - B_{t^*}}}{2}.
\end{equation}
We also have
\begin{equation}\label{eq:avregcost-geq-2d-part-1}
   \avcost{p_{t^*}, B_{t^*}} \geq 10^{t^*} / 5 \geq 2\cdot d(v,p_{t^*}), 
\end{equation}
by the assumption of this case and \Cref{eq:d-v-ptstar-part-1}.
Now, we conclude 
\begin{eqnarray*}
    \cost{v,C_v} &=& \cost{v, B_{t^*}} + \cost{v, C_v-B_{t^*}} \\
    &=& |B_{t^*}| \cdot \avcost{v, B_{t^*}} + \cost{v, C_v-B_{t^*}} \\
    &\geq& |B_{t^*}| \cdot (\avcost{p_{t^*}, B_{t^*}} - d(v,p_{t^*})) + \cost{v, C_v-B_{t^*}} \\
    &\geq& |B_{t^*}| \cdot \frac{\avcost{p_{t^*}, B_{t^*}}}{2} + \cost{v, C_v-B_{t^*}} \\
    &=& \frac{\cost{p_{t^*}, B_{t^*}}}{2} + \cost{v, C_v-B_{t^*}} \\
    &\geq& \frac{\cost{p_{t^*}, B_{t^*}}}{2} + \frac{\cost{p_{t^*}, C_v - B_{t^*}}}{2}
    = \frac{\cost{p_{t^*}, C_v}}{2}.
\end{eqnarray*}
The first inequality is triangle inequality, the second inequality holds by \Cref{eq:avregcost-geq-2d-part-1} and the last inequality holds by \Cref{eq:cost-v-outside-part-1}.
Finally, since $B_{t^*} \subseteq C_v$ (by \Cref{claim:Bi-subseteq-Cv-part-1}), we can apply \Cref{lem:robustness-property-2-part-1} to get
$$ \cost{u,C_v} = \cost{p_0,C_v} \leq \frac{3}{2} \cdot \cost{p_{t^*},C_v} \leq \frac{3}{2}\cdot 2 \cdot  \cost{v,C_v} = 3 \cdot \cost{v,C_v}. $$

\noindent
\textbf{Case 2:} $\avcost{p_{t^*}, B_{t^*}} < 10^{t^*}/5$. In this case, for every $q \in C_v - B_{t^*}$,
$$ d(q,p_{t^*-1}) \geq d(q,p_{t^*}) - d(p_{t^*},p_{t^*-1}) \geq 10^{t^*} - 10^{t^*}/2 = 10^{t^*}/2. $$
The second inequality holds by $q \in C_v - B_{t^*}$ and \Cref{lem:robustness-property-1-part-1}.
So,
\begin{equation}\label{eq:dqpt-geq-10t-part-1}
    d(q,p_{t^*-1}) \geq 10^{t^*}/2.
\end{equation}
We also have
\begin{eqnarray*}
    d(v,p_{t^*-1}) \leq d(v,u) + d(u,p_{t^*-1})
    = d(v,u) + d(p_0,p_{t^*-1})
    \leq 10^{t^*}/20 + 10^{t^*-1}/2
    = 10^{t^*}/10.
\end{eqnarray*}
The second inequality holds by \Cref{eq:definition-tstar-part-1} and \Cref{lem:robustness-property-1-part-1}.
So,
\begin{equation}\label{eq:dvpt-leq-10t-part-1}
    d(v,p_{t^*-1}) \leq 10^{t^*}/10.
\end{equation}
Combining \Cref{eq:dqpt-geq-10t-part-1} and \Cref{eq:dvpt-leq-10t-part-1}, we have
$$ d(v,p_{t^*-1}) \leq 10^{t^*}/10 = \frac{1}{5}\cdot 10^{t^*}/2 \leq \frac{1}{5}\cdot d(q,p_{t^*-1}), $$
which concludes
$$  \frac{d(q,p_{t^*-1}) }{d(q,v)} \leq \frac{d(q,p_{t^*-1}) }{d(q,p_{t^*-1}) - d(p_{t^*-1},v)} \leq \frac{d(q,p_{t^*-1}) }{d(q,p_{t^*-1}) - \frac{1}{5}\cdot d(q,p_{t^*-1})} = 5/4 \leq 2. $$
So, we have
$d(q,p_{t^*-1}) \leq 2 \cdot d(q,v) $.
Using this inequality for all  $q \in C_v - B_{t^*}$, we conclude
\begin{equation}\label{eq:cost-ptstarminusone-compare-to-v-part-1}
    \cost{p_{t^*-1}, C_v-B_{t^*}} \leq 2 \cdot \cost{v, C_v-B_{t^*}}.
\end{equation}
We also have
\begin{equation}\label{eq:cost-tstar-1-is-les-than-2-times-cost-of-v-part-1}
   \cost{p_{t^*-1}, B_{t^*}}
= \OPT^{B_{t^*}+p_{t^*}}_{1}(B_{t^*})
\leq 2 \cdot \OPT_{1}(B_{t^*})  
\leq 2 \cdot \cost{v, B_{t^*}}. 
\end{equation}
The equality holds by the definition of $p_{t^*-1}$ in a $t^*$-robust sequence, the first inequality holds by \Cref{lem:projection-lemma-part-1} and $\cost{B_{t^*} + p_{t^*}, B_{t^*}} = 0$ (we need this argument since we do not know if $v$ is inside the current space $\calP$).
Hence,
\begin{eqnarray*}
    \cost{p_{t^*-1}, C_v} &=& \cost{p_{t^*-1}, B_{t^*}} + \cost{p_{t^*-1}, C_v - B_{t^*}} \\
    &\leq& 2 \cdot \cost{v, B_{t^*}} + 2\cdot  \cost{v, C_v - B_{t^*}}
    = 2 \cdot \cost{v, C_v}.
\end{eqnarray*}
The equality holds since $B_{t^*} \subseteq C_v$ (by \Cref{claim:Bi-subseteq-Cv-part-1}) and 
the first inequality holds by \Cref{eq:cost-ptstarminusone-compare-to-v-part-1} and \Cref{eq:cost-tstar-1-is-les-than-2-times-cost-of-v-part-1}.
Finally, applying \Cref{lem:robustness-property-2-part-1} implies
$$ \cost{u, C_v} = \cost{p_0, C_v} \leq \frac{3}{2}\cdot \cost{p_{t^*-1}, C_v} \leq \frac{3}{2}\cdot 2 \cdot \cost{v, C_v} = 3 \cdot \cost{u, C_v}. $$
In both cases, we showed that  $\cost{u, C_v} \leq 3\cdot  \cost{v, C_v}$ which completes the proof.

\subsection{Proof of \Cref{lem:projection-lemma-part-1}}

Let $\calV^\star \subseteq \ground$ be of size $k$ such that $\cost{\calV^\star, \calP} = \OPT_k(\calP)$.
Assume $\calV = \pi_{\calU}(\calV^\star)$.
Obviously $\calV$ is of size at most $k$ and $\OPT_k^\calU(\calP) \leq \cost{\calV, \calP}$.
For each $p \in \calP$, let $v_p^\star$ and $u_p$ be the projection of $p$ onto $\calV^\star$ and $\calU$ respectively.
Also assume $v_p = \pi_\calU(v^\star_p)$.
Hence,
$d(p,v_p) \leq d(p,v_p^\star) + d(v_p^\star, v_p) \leq d(p,v_p^\star) + d(v_p^\star, u_p) \leq d(p,v_p^\star) + d(v_p^\star, p) + d(p, u_p) = d(p, \calU) + 2\cdot d(p, \calV^\star). $
Now, assign each $p$ to $v_p$ in solution $\calV$, which concludes
$$\cost{\calV, \calP} \leq \sum_{p \in \calP}  d(p, \calV) \leq \sum_{p \in \calP} (d(p, \calU) + 2 \cdot d(p, \calV^\star)) = \cost{\calU, \calP} + 2 \cdot \cost{\calV^\star, \calP}.$$
Hence, $\OPT_k^\calU(\calP) \leq \cost{\calV, \calP} \leq \cost{\calU, \calP} + 2 \cdot \OPT_k(\calP)$.

\subsection{Proof of \Cref{lem:lazy-updates-part-1}}

Assume $\calU \subseteq \ground$ of size at most $k$ is such that $\cost{\calU, \calP} = \OPT_{k}(\calP)$.
Define
$\calU' = \calU + (\calP \oplus \calP')$.
Obviously, $\calU'$ is a feasible solution for the $(k+s)$-median problem on $\calP'$ and since $\calU'$ contains $\calP \oplus \calP'$, we conclude
$$ \OPT_{k+s}(\calP') \leq  \cost{\calU',\calP'} = \cost{\calU + (\calP \oplus \calP'), \calP'} \leq \cost{\calU, \calP} = \OPT_{k}(\calP). $$

\section{Proof of \Cref{lm:intro}}
\label{sec:proof:lm:intro}

Assume the number of well-separated pairs w.r.t.~$(\calU_\init, \calV)$ is $k-m$ for some $m \in [0,k]$.
We call a $u \in \calU_\init$ \textit{good} if
\begin{itemize}
    \item $u$ forms a well-separated pair with a $v \in \calV$; and
    \item $v$ is not the closest center to any inserted point in the epoch, i.e., all of the points assigned to $v$ in $\calP_\final$ (denoted by $C_{v}(\calV, \calP_\final)$) are in $\calP_\init$.
\end{itemize}
Assume we have $g$ many good centers in $\calU_\init$ and consider orderings $\calU_\init = \{ u_1,u_2, \ldots , u_k\}$ and $\calV = \{ v_1,v_2, \ldots , v_k\}$, such that for all $i \in [1,g]$, the center $u_i$ is good and $v_i$ is such that $(u_i,v_i)$ is a well-separated pair.
Since $\left| \calP_\final \oplus \calP_\init \right| \leq \ell + 1 $, there are at most $\ell+1$ many centers $v$ such that $C_{v}(\calV, \calP_\final) \not\subseteq \calP_\init$.
So, we have at most $m + \ell + 1$ centers in $\calU_\init$ that are not good.

Now, assume $\calW^\star = \{ u_1, \ldots , u_g, v_{g+1},\ldots, v_k\}$ is derived by swapping at most $k - g \leq m + \ell + 1$ many centers in $\calU_\init$.
According to \Cref{lem:cost-well-sep-part-1} for $\calP = \calP_\init$ and $\calU = \calU_\init$, we have
\begin{eqnarray*}    
    \cost{\calW^\star, \calP_\final}
    &\leq&
    \sum_{i=1}^g \cost{u_i, C_{v}(\calV, \calP_\final)} + \sum_{i=g+1}^k \cost{v_i, C_{v}(\calV, \calP_\final)} \\
    &\leq &
    \sum_{i=1}^g 3 \cdot \cost{v_i, C_{v}(\calV, \calP_\final)} + \sum_{i=g+1}^k \cost{v_i, C_{v}(\calV, \calP_\final)} \\
    &\leq&
    3 \cdot \cost{\calV, \calP_\final}
\end{eqnarray*}
As a result, there exist a set of $k$ centers $\calW^\star$ such that $ \cost{\calW^\star, \calP_\final} \leq 3 \cdot \cost{\calV, \calP_\final}$ and $|\calW^\star \oplus \calU_\init| \leq m + \ell + 1$.
In order to complete the proof of lemma, it suffices to show $m \leq 4 \cdot (\ell + 1)$.

According to \Cref{lem:generalize-lemma-7.3-in-FLNS21-part-1} for $\calP=\calP_\init$, $\calU = \calU_\init$ and $r = 0$, there exist a $\tilde{\calU} \subseteq \calU_\init$ of size at most $k - \lfloor m/4 \rfloor$ such that
$\cost{\tilde{\calU}, \calP_\init} \leq 6\gamma \cdot \left( \cost{\calU_\init, \calP_\init} + \cost{\calV, \calP_\init} \right)$.
By hypothesis we have $\cost{\calV, \calP_\init} \leq 18 \cdot \cost{\calU_\init, \calP_\init}$. Hence,
\begin{eqnarray*}
  \cost{\calU', \calP_\init} 
  &\leq& 
  6\gamma \cdot \left( \cost{\calU_\init, \calP_\init} + \cost{\calV, \calP_\init} \right) \\
  &\leq& 6\gamma \cdot (1 + 18) \cdot \cost{\calU_\init, \calP_\init} \\  
  &=& 456000 \cdot \cost{\calU_\init, \calP_\init}.
\end{eqnarray*}
Finally, since we assumed that $\calU_\init$ is maximally $\ell$-stable, we conclude that $\lfloor m / 4 \rfloor \leq \ell$ which shows $m \leq 4\cdot (\ell + 1)$.

    \newpage
    \part{Full Version}
    \label{part:full}
    
    \section{Preliminaries}

\subsection{Problem Setting}

Assume we have a ground metric space $\ground$ with distance function $d$ together with a subset $\calP \subseteq \ground$ which contains the current present weighted points.
In the $k$-median problem on $\calP$, we want to find $\calU \subseteq \calP$ of size $k$ (called centers) so as to minimize the objective function $\sum_{p \in \calP} w(p) \cdot d(p, \calU)$, where $w(p) > 0$ is the weight of $p$ and $d(p,\calU) = \min_{q \in \calU} d(p,q)$ is the distance of $p$ to $\calU$.
In the dynamic setting, the current space $\calP$ is changing. 
Each update is either deleting a point of $\calP$, or inserting a points of $\ground - \calP$ into $\calP$ and assign a weight to it.
The aim is to maintain a set $\calU \subseteq \calP$ of $k$ centers that is a good approximation of the optimum $k$-median in the current space $\calP$ at each point in time.

We consider the \textbf{improper} dynamic $k$-median problem,
which means even if a point is not present in $\calP$, we can pick it as a center in our main solution (i.e.~open it as a center).
In order to have an algorithm for the \textbf{proper} $k$-median problem that opens centers only present in the current set of points $\calP$ at any time, we use the following result of \cite{FOCS24kmedian}.

\begin{lemma}[Lemma 9.1 in \cite{FOCS24kmedian}]
    Given any fully dynamic $(k, p)$-clustering algorithm that maintains an $\alpha$-approximate
    improper solution, we can maintain a $2\alpha$-approximate proper solution while incurring a $O(1)$ factor
    in the recourse and an additive $\tilde{O}(n)$ factor in the update time.
\end{lemma}

In \cite{FOCS24kmedian}, the authors achieved this result by providing a projection scheme that for every set of improper centers $\calU$ returns a proper subset of current set of points. We refer the reader to Section 9 of \cite{FOCS24kmedian} to see the complete procedure.

\subsection{Notations}

By a simple scaling, we can assume that all of the distances in the metric space are between $1$ and a parameter $\Delta$.
We use $\ground$ for the points of the ground metric space and $\calP$ for the set of current present points.
For simplicity, for each set $S$ and element $p$, we denote $S\cup \{ p \}$ and $S \setminus \{p\}$ by $S+p$ and $S-p$ respectively.
Each $p \in \calP$ has a weight denoted by $w(p)$ and for every $S \subseteq \calP $, we define
$$w(S) := \sum_{p \in S} w(p).$$
For two point sets $\calP$ and $\calP'$, we use $\calP \oplus \calP'$ for their symmetric difference.
Note that if a point $p$ is present in both $\calP$ and $\calP'$, but its weight is different in these two sets, we also consider $p \in \calP \oplus \calP'$.
We assume the same for $\calP - \calP'$.
For each $ S \subseteq \ground $, we define $\pi_S: \ground \rightarrow S$ to be the projection function onto $S$. For each
$p \in \ground $ and $ S \subseteq \ground$, we define $d(p, S) := d(p, \pi_S(p))$. For each $\calU \subseteq \ground$ and $S \subseteq \calP $, we use 
$$\cost{\calU, S} := \sum\limits_{p \in S} w(p)\  d(p, \calU). $$
We also define
$$\avcost{\calU, S} := \frac{\cost{\calU, S}}{w(S)} . $$
Assume $\calP,\calC \subseteq \ground$.
For any integer $m \geq 0$, we denote the cost of the optimum $m$-median solution for $\calP$ where we can only open centers from $\calC$, by $\OPT_{m}^{\calC}(\calP)$, i.e.
$$\OPT^{\calC}_m(\calP) = \min\limits_{\substack{\calU \subseteq \calC, \\ |\calU|\leq m}} \cost{\calU, \calP} .$$
Whenever we do not use the superscript $\calC$, we consider $\calC$ is the underlying ground metric space $\ground$, i.e.
$$\OPT_m(\calP) = \min\limits_{\substack{\calU \subseteq \ground, \\ |\calU|\leq m}} \cost{\calU, \calP} .$$
Note that $\calU$ is a subset of the ground set $\ground$, but its cost is computed w.r.t.~$\calP$.

For each $\calU \subseteq \ground$ and $u \in \calU$, we define 
$$C_u(\calU, \calP) := \{ p\in\calP \mid \pi_\calU(p) = u \} $$
to be the points of $\calP$ assigned to the center $u$ in the solution $\calU$ (breaking ties arbitrarily).
For each point $p \in \ground$ and value $r \geq 0$, we define $\text{Ball}^\calP_r(p)$ to be the points in $\calP$ whose distances from $p$ are at most $r$, i.e.
$$ \text{Ball}^\calP_r(p) := \{ q \in \calP \mid d(p,q) \leq r \}. $$
Note that $p$ itself might not be in $\text{Ball}^\calP_r(p)$ since we only consider points of $\calP$ in this ball.

\subsubsection{Constant Parameters}\label{parameters}
Throughout the paper we use constant parameters $\beta,\gamma$ and $C$ for convenience.
$\beta = O(1)$ is the constant approximation of the algorithm for static $k$-median on $\calP$ in \cite{MettuP02} that runs in $\tilde{O}(|\calP| \cdot k)$ time.
The final values of $\gamma$ and $C$ are as follows. 
\begin{equation*}
   \gamma = 4000, \ \text{and}\ C = 12 \cdot 3 \cdot 10^5\gamma\beta^2 
\end{equation*}

\subsection{Robustness}

In this section we describe the notion of robustness.
This notion is first defined in \cite{soda/FichtenbergerLN21}.
We change the definition in order to be able to get linear update time while the main good properties derived from the previous definition remain correct up to a constant overhead in approximation parts.

\begin{definition}[$t$-robust sequence]
    Assume $(p_0,p_1,\ldots, p_t)$ is a sequence of $t+1$ points.
    Let $\calP \subseteq \ground$ and $B_i = \text{Ball}_{10^i}^\calP(p_i)$ for each $0 \leq i \leq t$.
    We call this sequence, $t$-robust w.r.t.~$\calP$ if for every $1 \leq i \leq t$,
    \begin{eqnarray*}
        p_{i-1} =
        \begin{cases}
            p_i \quad & \text{if} \ 
            \avcost{p_i, B_i} \geq 10^i / 5 \\
            q_i \quad & \text{Otherwise}
        \end{cases}
    \end{eqnarray*}
    where $q_i \in B_i + p_i$ must satisfy
    $$\cost{q_i,B_i} \leq \min \{ 3 \ \OPT_1(B_i), \cost{p_i, B_i} \}.\footnote{$q_i$ might be equal to $p_i$ itself.}$$
\end{definition}

Note that in this definition, points $p_i$ need not necessarily be inside $\calP$.
But, the balls $B_i = \text{Ball}_{10^i}^\calP(p_i)$
are considered as a subset of current $\calP$.
Also note that $q_i$ is picked from $B_i + p_i$, but its cost is compared to the optimum improper solution for $1$-median problem in $B_i$, i.e.~$\OPT_{1}(B_i)$.

\begin{lemma}\label{lem:robustness-property-1}
    Let $(p_0,p_1,\ldots , p_t)$ be a $t$-robust sequence and assume $B_j = \text{Ball}^\calP_{10^j}(p_j)$. Then, for every $ 1 \leq j \leq t$, we have
    $$ d(p_{j-1}, p_j) \leq  10^j/2,  \quad B_{j-1} \subseteq B_j \quad \text{and} \quad d(p_0, p_j) \leq 10^j/2. $$
\end{lemma}

\begin{proof}
    The first part is trivial
    if $\avcost{p_j, B_j} \geq 10^j / 5$ since $p_j=p_{j-1}$.
    Now, assume $\avcost{p_j, B_j} < 10^j / 5$. By the definition of a $t$-robust sequence, we know 
    $\cost{p_{j-1}, B_j} \leq \cost{p_{j}, B_j}$. Hence,
    \begin{eqnarray*}
        d(p_j,p_{j-1}) 
        &=& \sum_{p \in B_j} \frac{w(p)}{w(B_j)} d(p_j,p_{j-1}) \\
        &\leq&
        \sum_{p \in B_j} \frac{w(p)}{w(B_j)} (d(p_j,p) + d(p_{j-1},p)) \\
        &=&
        \avcost{p_j, B_j} + \avcost{p_{j-1}, B_j} \\
        &\leq&
        2 \ \avcost{p_j, B_j} \\
        &\leq&
        \frac{2}{5}\cdot 10^j.
    \end{eqnarray*}
    For the second part, if $p \in B_{j-1}$, we have
    $$d(p,p_{j}) \leq d(p,p_{j-1}) + d(p_{j-1},p_j) \leq 10^{j-1} + \frac{2}{5}\cdot 10^j \leq 10^j,$$
    which implies $B_{j-1} \subseteq B_j$.
    For the last part,
    $$ d(p_0,p_j) \leq \sum_{i=1}^j d(p_{i-1},p_i) \leq \frac{2}{5} \sum_{i=1}^j 10^i = \frac{2}{5}\cdot \frac{10^{j+1} - 2}{9} \leq 10^j / 2. $$
\end{proof}

\begin{lemma}\label{lem:robustness-property-2}
    Let $(p_0,p_1,\ldots , p_t)$ be a $t$-robust sequence and assume $B_j = \text{Ball}^\calP_{10^j}(p_j)$. Then, for every $0 \leq i \leq t$ and every $S \subseteq \calP $ containing $B_i$ (i.e.~$B_i \subseteq S$), we have
    $$ \cost{p_0, S} \leq \frac{3}{2} \ \cost{p_i,S} $$
\end{lemma}

\begin{proof}
    Since $(p_0,p_1,\ldots , p_t)$ is $t$-robust, for every $1 \leq j \leq t$, we know
    \begin{equation}\label{eq:pj-1-better-than-pj}
        \cost{p_{j}, B_j} \geq \cost{p_{j-1}, B_j}.
    \end{equation}
    Assume $q \in \calP \setminus B_j$. Then, using \Cref{lem:robustness-property-1},
    $$ d(q,p_j) \geq 10^j \geq 2\  d(p_0,p_j). $$
    Hence,
    $$ \frac{3}{2}\  d(q,p_j)
    = d(q,p_j) + \frac{1}{2}\ d(q,p_j)
    \geq d(q,p_j) + d(p_0,p_j)
    \geq d(q,p_0),$$
    which means
    \begin{equation}\label{eq:dqpj-compare-to-dqp0}
        d(q,p_j) \geq \frac{2}{3}\ d(q,p_0).
    \end{equation}
    Finally, Since $B_1 \subseteq \cdots \subseteq B_i \subseteq S$ (by \Cref{lem:robustness-property-1}) we can apply \Cref{eq:pj-1-better-than-pj} and \Cref{eq:dqpj-compare-to-dqp0} repeatedly to get
    \begin{eqnarray*}
        \cost{p_i,S} &=& \cost{p_i, B_i} + \cost{p_i, S \setminus B_i} \\
        &\geq&
        \cost{p_{i-1}, B_i} + \frac{2}{3}\  \cost{p_0, S \setminus B_i} \\
        &=&
        \cost{p_{i-1}, B_{i-1}} + \cost{p_{i-1}, B_i \setminus B_{i-1}} + \frac{2}{3} \ \cost{p_0, S \setminus B_i} \\
        &\geq&
        \cost{p_{i-2}, B_{i-1}} + \frac{2}{3}\  \cost{p_{0}, B_i \setminus B_{i-1}} + \frac{2}{3}\  \cost{p_0, S \setminus B_i} \\
        &=&
        \cost{p_{i-2}, B_{i-1}} + \frac{2}{3}\  \cost{p_0, S \setminus B_{i-1}} \\
        &\vdots& \\
        &\geq&
        \cost{p_{0}, B_{1}} + \frac{2}{3} \ \cost{p_0, S \setminus B_{1}} \\
        &\geq&
        \frac{2}{3}\  \cost{p_0, S} .
    \end{eqnarray*}
\end{proof}

\begin{definition}[$t$-robust point]
    We say $p$ is $t$-robust w.r.t.~$\calP$ whenever there exist a $t$-robust sequence $(p_0,p_1,\ldots, p_t)$ w.r.t.~$\calP$ such that $p_0 = p$.
\end{definition}

\begin{definition}[robust solution]\label{def:bounded-robust}
    Assume $\calU$ is a set of centers. We call it robust w.r.t.~$\calP$, if for every $u \in \calU$, the following condition holds.
    \begin{itemize}
        \item $u$ is $t$-robust w.r.t.~$\calP$, where $t$ is the smallest integer satisfying 
        $10^t \geq d(u,\calU - u  )/ 200. $
    \end{itemize}
\end{definition}

\subsection{Well-Separated Pairs}

\begin{definition}[well-separated pair]
    Suppose $ \calU$ and $\calV $ are two sets of centers.
    For $ u \in \calU $ and $ v \in \calV $, we call $(u,v)$ a well-separated pair with respect to $(\calU, \calV)$, whenever the following inequalities hold.
    \begin{eqnarray*}
        d(u, \calU - u) &\geq& \gamma \ d(u, v) \\
        d(v, \calV - v) &\geq& \gamma \ d(u, v)
    \end{eqnarray*}
\end{definition}

It is easy to see that each point $u \in \calU$ either forms a well-separated pair with a unique $v \in \calV$, or it does not form a well-separated pair with any center in $\calV$.
This can be shown by a simple argument using triangle inequality (assuming the value of $\gamma$ is large enough).

\subsection{Relation Between Improper and Proper Optimum Values}

For some technical reasons, in some part of our algorithm, we look for proper solutions and for some parts we look for improper solutions. Here, we provide a relation between the cost of an improper and a proper solution, which we will use in our analyses.

\begin{lemma}\label{lem:proper-is-2-approx-of-improper}
    Assume $\ground$ is the underlying ground metric space and $\calP \subseteq \ground$.
    Then, we have
    $$ \OPT_k(\calP) \leq  \OPT_k^{\calP}(\calP) \leq 2 \ \OPT_k(\calP). $$
\end{lemma}

\begin{proof}
    The left inequality is obvious since every proper solution can be considered as an improper solution as well.
    For the right inequality, assume $\calV^*$ is the optimum improper solution for $k$-median on $\calP$. Consider the projection function $\pi_{\calP}: \calV^* \rightarrow \calP$.
    Let $\calU = \pi_{\calP}(\calV^*)$.
    We show that $\calU$ is a $2$ approximate solution for the proper $k$-median problem on $\calP$. 
    Note that $\calU$ is a proper feasible solution.
    Assume $v^* \in \calV^*$ is projected to $\pi_{\calP}(v^*) = v \in \calU$. For each $p \in C_{v^*}(\calV^*, \calP)$, we have
    $$d(p,v) \leq d(p,v^*) + d(v^*,v) \leq 2\ d(p,v^*). $$
    The last inequality is because $p \in \calP$ and $\pi_{\calP}(v^*) = v$.
    By summing up this inequality for each $p \in C_{v^*}(\calV^*, \calP)$ (considering weights $w(p)$), we conclude
    $$ \cost{v, C_{v^*}(\calV^*, \calP)} \leq 2 \ \cost{v^*, C_{v^*}(\calV^*, \calP)}. $$
    Since $v \in \calU$, this means
    $$ \cost{\calU, C_{v^*}(\calV^*, \calP)} \leq 2 \ \cost{v^*, C_{v^*}(\calV^*, \calP)}. $$
    Finally, by summing up these inequalities for each $v^* \in \calV^*$, we have
    $$ \cost{\calU, \calP} \leq 2 \ \cost{\calV^*, \calP} =2\ \OPT_{k}(\calP), $$
    which concludes the right inequality since
    $$ \OPT^{\calP}_k(\calP) \leq \cost{\calU, \calP} \leq 2\  \OPT_{k}(\calP). $$
\end{proof}

\subsection{Key Lemmas Used Throughout The Paper}

In this section, we provide the key lemmas that we use in the analysis of our algorithm.

\begin{lemma}[Lazy Updates Lemma, Lemma 3.3 in \cite{FOCS24kmedian}]\label{lem:lazy-updates}
    Assume $\calP$ and $\calP'$ are two sets of points such that $|\calP \oplus \calP'| \leq l$. Then for every $k \geq 0$ we have
    $$\OPT_{k+l}(\calP') \leq \OPT_k(\calP).$$
\end{lemma}

\begin{proof}
    Assume $\calU \subseteq \ground$ of size at most $k$ is such that $\cost{\calU, \calP} = \OPT_{k}(\calP)$.
    Define
    $\calU' = \calU + (\calP \oplus \calP')$.
    Obviously, $\calU'$ is a feasible solution for the $(k+l)$-median problem on $\calP'$ and since $\calU'$ contains $\calP \oplus \calP'$, we conclude
    $$ \OPT_{k+l}(\calP') \leq  \cost{\calU',\calP'} = \cost{\calU + (\calP \oplus \calP'), \calP'} \leq \cost{\calU, \calP} = \OPT_{k}(\calP). $$
\end{proof}

\begin{lemma}[Double-Sided Stability Lemma]\label{lem:double-sided-stability}
    Assume for a point set $\calP$ and values $k$, $\eta$ and $0 \leq r \leq k$ we have
    $$ \OPT_{k-r}(\calP) \leq \eta \ \OPT_k(\calP). $$
    Then, the following inequality holds.
    $$ \OPT_{k}(\calP) \leq 4 \ \OPT_{k + \lfloor r / (12\eta) \rfloor }(\calP) $$ 
\end{lemma}

\begin{proof}

Consider the LP relaxation for the $k$-median problem on $\calP$ for each $k$ as follows.
\begin{align*}
    \min & \sum_{p \in \calP} \sum_{c \in \calC}  (w(p) \cdot d(c,p))x_{cp} & \\
    \text{s.t.}  & \quad x_{cp} \leq y_{c}  & \forall c \in \calC, p \in \calP 
    \\
    &\sum_{c \in \calC} x_{cp} \geq 1  & \forall p \in \calP \label{eq:all-open0}\\
    &\sum_{c \in \calC} y_c \leq k & \\
    &x_{cp}, y_c \geq 0 & \forall c \in \calC, p \in \calP
\end{align*}

Note that $\calC$ is the set of potential centers to open. We assume $\calC$ is the underlying ground set of points $\ground$.
This is because we considered the definition of $\OPT_k(\calP)$ for the improper case.

Denote the cost of the optimal fractional solution for this LP by $\FOPT_k$.
Since the space $\calP$ is fixed here, we denote $\OPT_k(\calP)$ by $\OPT_k$.
It is known that the integrality gap of this relaxation is at most $3$ \cite{CS11}. So, for every $k$ we have
\begin{equation}\label{eq:int-gap}
    \FOPT_k \leq \OPT_k \leq 3\ \FOPT_k.   
\end{equation}

\begin{claim}
    For every $k_1$, $k_2$ and $0 \leq \alpha, \beta \leq 1$ such that $\alpha + \beta = 1$, we have
    $$\FOPT_{\alpha k_1 + \beta k_2} \leq \alpha \ \FOPT_{k_1} + \beta \ \FOPT_{k_2}. $$
\end{claim}

\begin{proof}
    Assume optimal fractional solutions $(x^*_1,y^*_1)$ and $(x^*_2,y^*_2)$ for above LP relaxation of $k_1$ and $k_2$-median problems respectively. It is easy to verify that $(\alpha x^*_1 + \beta x^*_2, \alpha y^*_1 + \beta y^*_2)$ is a feasible solution for fractional $(\alpha k_1 + \beta k_2)$-median problem whose cost is $\alpha \ \FOPT_{k_1} + \beta \ \FOPT_{k_2}$, which concludes the claim.
\end{proof}

Now, plug $k_1 = k-r$, $k_2 = k + r/(12\eta)$, $\alpha = 1/(12\eta)$ and $\beta = 1 - \alpha$ in the claim. We have
$$\alpha k_1 + \beta k_2 = \frac{1}{12\eta}(k-r) + \left(1-\frac{1}{12\eta}\right)\left(k+\frac{r}{12\eta}\right) = k - \frac{r}{(12\eta)^2} \leq k. $$
As a result,
$$ \FOPT_k \leq \FOPT_{\alpha k_1 + \beta k_2} \leq \alpha \ \FOPT_{k_1} + \beta \ \FOPT_{k_2}. $$
Together with \Cref{eq:int-gap}, we have
$$ \OPT_k \leq 3\alpha\ \OPT_{k_1} + 3\beta \ \OPT_{k_2}. $$
We also have the assumption that $ \OPT_{k_1} = \OPT_{k-r} \leq \eta \ \OPT_k$, which implies
$$ \OPT_k \leq 3\alpha\eta\ \OPT_{k} + 3\beta\ \OPT_{k_2}. $$
Finally
$$ \OPT_k \leq \left( \frac{3\beta}{1-3\alpha\eta} \right) \OPT_{k_2} \leq 4 \ \OPT_{k_2} \leq 4\  \OPT_{k +\lfloor r/(12\eta) \rfloor}. $$
The second inequality follows from 
$$ \frac{3\beta}{1 - 3\alpha\eta} \leq \frac{3}{1 - 3\alpha\eta} = \frac{3}{1 - 1/4} = 4. $$

\end{proof}

\begin{lemma}[variation of Lemma 7.4 in \cite{soda/FichtenbergerLN21}]\label{lem:cost-well-sep}
    If $\calU$ and $\calV$ are two set of centers and $\calP$ is a set of points such that $\calU$ is robust w.r.t.~$\calP$. Then, for every well-separated pair $(u,v)$ w.r.t.~$(\calU, \calV)$, we have 
    $$\cost{u, C_v(\calV, \calP) } \leq 5 \ \cost{v, C_v(\calV, \calP) }.$$
\end{lemma}

\begin{proof}

Since, the point set $\calP$ is fixed, denote $C_v(\calV,\calP)$ by $C_v$ for simplicity.
If $d(u,v)=0$, the statement is trivial.
Now assume $d(u,v) \geq 1$.
Since $\calU$ is robust, then $u$ is $t$-robust where $t$ is the smallest integer such that
\begin{equation}
    10^t \geq d(u,\calU-u)/200 \geq (\gamma/200)\ d(u,v) = 20\ d(u,v)
\end{equation}
The second inequality is because $(u,v)$ is well-separated pair and the equality is because $\gamma = 4000$.
Assume $t^*$ is the integer satisfying
\begin{equation}\label{eq:definition-tstar}
   20\ d(u,v) \leq 10^{t^*} < 10\cdot 20\ d(u,v). 
\end{equation}
Note that $t^* \geq 1$. Since $t^* \leq t$ and $u$ is $t$-robust, then it is also $t^*$-robust which means there is a $t^*$-robust sequence 
$(p_0,p_1,\ldots,p_{t^*})$ such that $p_0=u$.
For Simplicity, assume $B_i = \text{Ball}_{10^i}^\calP(p_i)$ for $0 \leq i \leq t^*$.

\begin{claim}\label{claim:Bi-subseteq-Cv}
    For each $0 \leq i \leq t^*$, we have $B_i \subseteq C_v$.
\end{claim}

\begin{proof}
    Assume $q \in B_i$. Then,
    \begin{eqnarray*}
        d(q,v) &\leq& d(q,p_i) + d(p_i,u) + d(u,v) \\
        &=& d(q,p_i) + d(p_i,p_0) + d(u,v) \\
        &\leq& 10^i + 10^i/2 + 10^{t^*}/20 \\
        &\leq& 2 \cdot 10^{t^*}.
    \end{eqnarray*}
    For the second inequality, we used $q \in B_i$, \Cref{lem:robustness-property-1} and \Cref{eq:definition-tstar}.
    Now, for every $v' \in \calV - v$ we have
    $$ d(v,v') \geq d(v,\calV-v) \geq 4\cdot10^3\  d(u,v) > 10^{t^*+1} \geq 2\ d(q,v). $$
    The second inequality is by $(u,v)$ being well-separated (and $\gamma = 4000$), the third one is by \Cref{eq:definition-tstar}.
    This concludes $d(q,v') \geq d(v,v') - d(q,v) > d(q,v)$ which means $v$ is the closest center to $q$ in $\calV$. So, $q \in B_i$ and finally $B_i \subseteq C_v$.
\end{proof}

Now, we according to the definition of $t$-robust sequence we have two cases.

\noindent
\textbf{Case 1:} $\avcost{p_{t^*}, B_{t^*}} \geq 10^{t^*}/5$. In this case we have $p_{t^*-1} = p_{t^*}$ and
\begin{eqnarray*}
    d(v,p_{t^*}) &=& d(v,p_{t^*-1}) \\
    &\leq& d(v,p_0) + d(p_0,p_{t^*-1}) \\
    &=& d(v,u) + d(p_0,p_{t^*-1}) \\
    &\leq& 10^{t^*}/20 + 10^{t^*-1}/2 \\
    &=& 10^{t^*}/10.
\end{eqnarray*}
The second inequality holds by \Cref{lem:robustness-property-1} and \Cref{eq:definition-tstar}.
So,
\begin{equation}\label{eq:d-v-ptstar}
    d(v,p_{t^*}) \leq 10^{t^*}/10.
\end{equation}
For every $q \in C_v - B_{t^*}$ we have
$$ d(q,v) \geq d(p_{t^*},q) - d(v,p_{t^*}) \geq 10^{t^*} - d(v,p_{t^*}) \geq d(v, p_{t^*}). $$
The second inequality follows by $q \in C_v-B_{t^*}$ and the last inequality follows by
\Cref{eq:d-v-ptstar}.
So, 
$$ d(q,v) \geq (d(q,v) + d(v,p_{t^*}))/2 \geq d(q,p_{t^*})/2, $$
which implies
\begin{equation}\label{eq:cost-v-outside}
    \cost{v, C_v - B_{t^*}} \geq \frac{\cost{p_{t^*}, C_v - B_{t^*}}}{2}.
\end{equation}
We also have
\begin{equation}\label{eq:avregcost-geq-2d}
   \avcost{p_{t^*}, B_{t^*}} \geq 10^{t^*} / 5 \geq 2\ d(v,p_{t^*}), 
\end{equation}
by the assumption of this case and \Cref{eq:d-v-ptstar}.
Now, we conclude 
\begin{eqnarray*}
    \cost{v,C_v} &=& \cost{v, B_{t^*}} + \cost{v, C_v-B_{t^*}} \\
    &=& w(B_{t^*}) \cdot \avcost{v, B_{t^*}} + \cost{v, C_v-B_{t^*}} \\
    &\geq& w(B_{t^*}) \cdot (\avcost{p_{t^*}, B_{t^*}} - d(v,p_{t^*})) + \cost{v, C_v-B_{t^*}} \\
    &\geq& w(B_{t^*}) \cdot \frac{\avcost{p_{t^*}, B_{t^*}}}{2} + \cost{v, C_v-B_{t^*}} \\
    &=& \frac{\cost{p_{t^*}, B_{t^*}}}{2} + \cost{v, C_v-B_{t^*}} \\
    &\geq& \frac{\cost{p_{t^*}, B_{t^*}}}{2} + \frac{\cost{p_{t^*}, C_v - B_{t^*}}}{2}\\
    &=& \frac{\cost{p_{t^*}, C_v}}{2}.
\end{eqnarray*}
The first inequality is triangle inequality, the second inequality holds by \Cref{eq:avregcost-geq-2d} and the last inequality holds by \Cref{eq:cost-v-outside}.
Finally, since $B_{t^*} \subseteq C_v$ (by \Cref{claim:Bi-subseteq-Cv}), we can apply \Cref{lem:robustness-property-2} to get
$$ \cost{u,C_v} = \cost{p_0,C_v} \leq \frac{3}{2} \cdot \cost{p_{t^*},C_v} \leq \frac{3}{2}\cdot 2 \  \cost{v,C_v} = 3 \ \cost{v,C_v}. $$

\noindent
\textbf{Case 2:} $\avcost{p_{t^*}, B_{t^*}} < 10^{t^*}/5$. In this case, for every $q \in C_v - B_{t^*}$,
$$ d(q,p_{t^*-1}) \geq d(q,p_{t^*}) - d(p_{t^*},p_{t^*-1}) \geq 10^{t^*} - 10^{t^*}/2 = 10^{t^*}/2. $$
The second inequality holds by $q \in C_v - B_{t^*}$ and \Cref{lem:robustness-property-1}.
So,
\begin{equation}\label{eq:dqpt-geq-10t}
    d(q,p_{t^*-1}) \geq 10^{t^*}/2.
\end{equation}
We also have
\begin{eqnarray*}
    d(v,p_{t^*-1}) &\leq& d(v,u) + d(u,p_{t^*-1}) \\
    &=& d(v,u) + d(p_0,p_{t^*-1}) \\
    &\leq& 10^{t^*}/20 + 10^{t^*-1}/2 \\
    &=& 10^{t^*}/10.
\end{eqnarray*}
The second inequality holds by \Cref{eq:definition-tstar} and \Cref{lem:robustness-property-1}.
So,
\begin{equation}\label{eq:dvpt-leq-10t}
    d(v,p_{t^*-1}) \leq 10^{t^*}/10.
\end{equation}
Combining \Cref{eq:dqpt-geq-10t} and \Cref{eq:dvpt-leq-10t}, we have
$$ d(v,p_{t^*-1}) \leq 10^{t^*}/10 = \frac{1}{5}\cdot 10^{t^*}/2 \leq \frac{1}{5}\cdot d(q,p_{t^*-1}), $$
which concludes
$$  \frac{d(q,p_{t^*-1}) }{d(q,v)} \leq \frac{d(q,p_{t^*-1}) }{d(q,p_{t^*-1}) - d(p_{t^*-1},v)} \leq \frac{d(q,p_{t^*-1}) }{d(q,p_{t^*-1}) - \frac{1}{5}\cdot d(q,p_{t^*-1})} = 5/4 \leq 3. $$
So, we have
$$d(q,p_{t^*-1}) \leq 3 \ d(q,v). $$
Using this inequality for all  $q \in C_v - B_{t^*}$, we conclude
\begin{equation}\label{eq:cost-ptstarminusone-compare-to-v}
    \cost{p_{t^*-1}, C_v-B_{t^*}} \leq 3 \ \cost{v, C_v-B_{t^*}}.
\end{equation}
Hence,
\begin{eqnarray*}
    \cost{p_{t^*-1}, C_v} &=& \cost{p_{t^*-1}, B_{t^*}} + \cost{p_{t^*-1}, C_v - B_{t^*}} \\
    &\leq& 3 \ \OPT_{1}(B_{t^*}) + \cost{p_{t^*-1}, C_v - B_{t^*}} \\
    &\leq& 3 \ \cost{v, B_{t^*}} + \cost{p_{t^*-1}, C_v - B_{t^*}}. \\
    &\leq& 3 \ \cost{v, B_{t^*}} + 3\  \cost{v, C_v - B_{t^*}}. \\
    &=& 3 \ \cost{v, C_v}.
\end{eqnarray*}
The equalities hold since $B_{t^*} \subseteq C_v$ (by \Cref{claim:Bi-subseteq-Cv}), the first inequality holds by the definition of $t^*$-robust sequence, the second inequality holds since $v \in B_{t^*}$ (by \Cref{eq:definition-tstar}) and the last inequality holds by \Cref{eq:cost-ptstarminusone-compare-to-v}.
Finally, applying \Cref{lem:robustness-property-2} implies
$$ \cost{u, C_v} = \cost{p_0, C_v} \leq \frac{3}{2}\cdot \cost{p_{t^*-1}, C_v} \leq \frac{3}{2}\cdot 3 \ \cost{v, C_v} \leq 5 \ \cost{u, C_v}. $$
In both cases, we showed that  $\cost{u, C_v} \leq 5\  \cost{v, C_v}$ which completes the proof.

\end{proof}

\begin{lemma}[generalization of Lemma 7.3 in \cite{soda/FichtenbergerLN21}]\label{lem:generalize-lemma-7.3-in-FLNS21}
    Suppose $\calP$ is a set of points, $\calU$ is a set of $k$ centers and $\calV$ is a set of at most $k+r$ centers.
    If the number of well-separated pairs with respect to $(\calU, \calV)$ is $k - m$, then there exist a $\bar{\calU} \subseteq \calU$ of size at most $k - \lfloor (m- r) / 4 \rfloor$ such that
    $$ \cost{\bar{\calU}, \calP} \leq 6\gamma\left( \cost{\calU, \calP} + \cost{\calV, \calP} \right). $$
\end{lemma}

\begin{proof}

Consider the standard LP relaxation for the weighted $k$-median problem. We consider the set of potential centers to open $\calU$, and we want to open at most $k - (m- r) / 4 $ many centers. So, we have the following LP.
\begin{align*}
    \min & \sum_{p \in \calP} \sum_{u \in \calU}  (w(p) \cdot d(u,p))x_{up} & \\
    \text{s.t.}  & \quad x_{up} \leq y_{u}  & \forall u \in \calU, p \in \calP 
    \\
    &\sum_{u \in \calU} x_{up} \geq 1  & \forall p \in \calP \\
    &\sum_{u\in\calU} y_u \leq k - (m- r) / 4 & \\
    &x_{up}, y_u \geq 0 & \forall u\in \calU, p \in \calP
\end{align*}
Now, we explain how to construct a fractional solution for this LP.

\noindent
\textbf{Fractional Opening of Centers.}
Consider the projection $\pi_\calU: \calV \rightarrow \calU$ function.
Assume $\calU = \calU_I + \calU_F$ is a partition of $\calU$ where $\calU_I$ contains those centers $u \in \mathcal{U}$ satisfying at least one of the following conditions:
\begin{itemize}
    \item $u$ forms a well-separated pair with one center in $\mathcal{V}$.
    \item $|\pi^{-1}_\calU(u)| \geq 2$.
\end{itemize}
For every $u\in \mathcal{U}_I$, set $y_u=1$ and for every $u\in \mathcal{U}_F$ set $y_u = 1/2$. 
First, we show that
$$\sum_{u\in \mathcal{U}} y_u \leq  k - (m- r) / 4.$$
 
Each center $u \in \calU$ that forms a well-separated pair with a center $v\in \calV$ has $|\pi^{-1}_\calU(u)| \geq 1$ since $u$ must be the closest center to $v$ in $\calU$.
Since the number of well-separated pairs is $k-m$, we have
\begin{equation*}
    k+r \geq |\calV| = \sum_{u\in \mathcal{U}} |\pi^{-1}(u)| \geq \sum_{u\in \mathcal{U}_I} | \pi^{-1}(u)| \geq k- m   + 2\cdot\left(|\mathcal{U}_I| - (k-m)\right).
\end{equation*}
Hence,
$$|\calU_I| \leq \frac{k+r + (k-m)}{2} = k - \frac{m-r}{2}. $$
Finally, we conclude
\begin{eqnarray*}
   \sum_{u\in \calU} y_u &=& \sum_{u\in \calU_I} y_u + \sum_{u\in \calU_F} y_u \\
   &\leq&  k - \frac{m-r}{2} + \frac{1}{2} \cdot \frac{m-r}{2} \\
   &=&
   k - \frac{m-r}{4}.
\end{eqnarray*}

\noindent
\textbf{Fractional Assignment of Points.}     
For every $p\in \calP$, assume $ v_p = \pi_{\calV}(p) $ is the closest center to $p$ in $\mathcal{V}$ and $u_p = \pi_\calU(v_p)$ is the closest center in $\calU$ to $v_p$.
We have three cases:
\begin{itemize}

\item If $y_{u_p} = 1$, then set $x_{u_pp} = 1$. The cost of this assignment would be $w(p) \cdot d(u_p, p)$.

\item If $y_{u_p} = 1/2$ and there is a center $u'_p\in \calU  - u_p$ such that $d(u_p, u'_p) \leq \gamma\cdot  d(u_p, v_p)$, then set $x_{u_pp} = x_{u_p' p} = 1/2$. Note that $u_p' \neq u_p$ which means point $p$ is assigned one unit to centers.
The cost of this assignment would be
\begin{eqnarray*}
    \frac{1}{2} w(p)\left(d(p,u_p) + d(p,u_p')\right) &\leq&  \frac{1}{2} w(p)\left(d(p,u_p) + d(p,u_p) + d(u'_p,u_p)\right)  \\
    &\leq& w(p)\left(d(p,u_p) + \frac{\gamma}{2}\cdot d(u_p,v_p)\right).
\end{eqnarray*}

\item If $y_{u_p} = 1/2$ and the previous case does not hold, then since $(u_p, v_p)$ is not a well-separated pair, there is a center $v'_p \in \calV - v_p$ such that $d(v_p, v'_p) \leq \gamma \cdot d(u_p, v_p)$. Let $u_p' = \pi_\calU(v'_p)$ and  set $x_{u_p p} = x_{u'_p p} = 1/2$. First, we show that $u'_p \neq u_p$. Since $y_{u_p} = 1/2$, we have $u_p \in \calU_F$ which concludes $|\pi^{-1}_\calU(u_p)| \leq 1$.
We also know that $\pi_\calU(v_p) = u_p$.
So, $v_p$ is the only center in $\calV$ mapped to $u_p$ which implies
$\pi_\calU(v'_p) \neq u_p$ or $u_p' \neq u_p$ (note that $v_p' \neq v_p$). So, point $p$ is assigned one unit to centers.
The cost of this assignment would be
\begin{eqnarray*}
    \frac{1}{2} w(p) \left(d(p, u_p) + d(p, u'_p)\right) 
    &\leq&
    \frac{1}{2} w(p) \left( d(p, u_p) + d(p, u_p) + d(u_p, v'_p)  + d(v'_p, u'_p)\right) \\
    &\leq&
    \frac{1}{2} w(p) \left( d(p, u_p) + d(p, u_p) + d(u_p, v'_p)  + d(v'_p, u_p)\right) \\
    &=&
    w(p) \left( d(p, u_p)  + d(u_p,v_p')\right) \\
    &\leq&
    w(p) \left( d(p, u_p)  + d(u_p, v_p) + d(v_p,v_p') \right) \\
    &\leq& w(p)\left(   d(p,u_p) + (\gamma+1)\ d(u_p,v_p)\right).
\end{eqnarray*}
The second inequality is because of the choice of $u'_p = \pi_{\calU}(v_p')$ and the last inequality is because $d(v_p,v_p') \leq \gamma \cdot d(u_p,v_p)$.
\end{itemize} 

\noindent
\textbf{Bounding the Cost.}
Assume $u_p^* = \pi_\calU(p)$ for each $p \in \calP$.
We have
$$ d(u_p,v_p) \leq d(u_p^*,v_p) \leq d(u_p^*,p) + d(p,v_p). $$
The first inequality is by the choice of $u_p = \pi_\calU(v_p)$.
As a result,
\begin{equation}\label{eq:bound-on-d-up-vp}
    \sum_{p \in \calP} w(p)\ d(u_p,v_p)
    \leq
    \sum_{p \in \calP} w(p)\left( d(u^*_p,p) + d(p,v_p) \right)
    =
    \cost{\calU,\calP} + \cost{\calV,\calP}.
\end{equation}
In each of the cases of fractional assignments of points to centers, the cost of assigning a point $p \in C_v(\calV, \calP)$ is at most
$ w(p) \left(d(p, u_p) + (\gamma+1) \ d(u_p, v_p) \right)$.
As a result, the total cost of this assignment is upper bounded by
\begin{eqnarray*}
    \sum_{p\in P} w(p) \left(d(p,u_p) + (\gamma+1) \ d(u_p, v_p)\right) 
    &\leq& 
     \sum_{p\in \calP} w(p) \left(d(p,v_p) + (\gamma+2) \ d(u_p, v_p)\right) \\
    &\leq& 
    \cost{\calV,\calP} +(\gamma+2) \left( \cost{\calU,\calP} + \cost{\calV,\calP} \right) \\
    &\leq& 
    2\gamma \left( \cost{\calU,\calP} + \cost{\calV,\calP} \right).
\end{eqnarray*}
The second inequality follows by \Cref{eq:bound-on-d-up-vp}.

Finally, note that the integrality gap of the LP relaxation is known to be at most $3$ \cite{CS11}. As a result, there exist an integral solution whose cost is at most $6\gamma\left( \cost{\calU,\calP} + \cost{\calV,\calP} \right)$
and this solution opens at most 
$ k - \lfloor \frac{m-r}{4} \rfloor $ centers which completes the proof.

\end{proof}

    \section{Useful Subroutines of Our Algorithm}

\
In this section, we provide useful static subroutines that we are going to use in our main algorithm.

\subsection{\RandLocalSearch}\label{sec:rand-local-search}
This algorithm is first introduced in \cite{FOCS24kmedian}.
Given a set of points $\calP$, a set of $k$ centers $\calU$ and a number $s \leq k$, the aim is to find a good subset of centers $\calU$ of size $k - s$.
The algorithm starts with an arbitrary $\calU^* \subseteq \calU$ of size $k-s$ and tries to improve this subset by random local swaps for $\tilde{\Theta}(s)$ many iterations.

\begin{algorithm}[H]
  \DontPrintSemicolon
  $\calU^* \gets$ arbitrary subset of $\calU$ of size $k - s$. 

  \For{$\tilde{\Theta}(s)$ iterations}{
    Sample $v \in \calU-\calU^*$ independently and uniformly at random.

    $z^* \gets \arg\min_{z \in \calU^* + v} \{ \cost{\calU^*+v-z,\calP} - \cost{\calU^*, \calP} \} $.

    $\calU^* \gets \calU^* - z^* + v $.
  }
  \Return $\calU^*$.
  \caption{(Algorithm 3 in \cite{FOCS24kmedian}), \RandLocalSearch($\calP, \calU, s$)}
\end{algorithm}

We have the following lemma which shows the subset derived from \RandLocalSearch \ is actually a good subset w.h.p.

\begin{lemma}[Lemma 3.18 in \cite{FOCS24kmedian}]\label{lem:cost-local-search}
    If $\calU^*$ is returned by \RandLocalSearch$(\calP,\calU,s)$, then with high probability,
    $$\cost{\calU^*, \calP} \leq 2 \ \cost{\calU, \calP} + 12\ \OPT_{k - s}(\calP).$$
\end{lemma}

\subsection{\DevelopCenters}\label{sec:develop-centers}
Given a set of points $\calP$, a set of centers $\calU$ and an integer $s \geq 1$,
the aim of this subroutine is to extend $\calU$ to $\tu$ by adding $s$ centers which is a good approximation to the best solution that adds $s$ centers to $\calU$, i.e.~the algorithm finds $\tu$ such that
$$ \cost{\tu,\calP} \leq \beta \ \min_{\substack{\calF \subseteq \calP, \\ |\calF| \leq s}} \cost{\calU+\calF, \calP}.\footnote{Recall $\beta$ from \Cref{parameters}} $$

The algorithm defines a new metric space and runs a standard $(s+1)$-median static algorithm on this space.
The high level idea behind this subroutine is that we have a set of fixed $k$ centers $\calU$ that must be contained in $\tu$.
So, if we can treat these fixed centers as a single center instead of a set of $k$ centers, we might be able to reduce the problem to $(s+1)$-median.
So, we contract all of the points $\calU$ to a single point $u^*$ and define a new space $\calP' = (\calP - \calU) + u^*$ with a new metric $d'$ as follows.
\begin{equation}\label{def:new-metric}
    \begin{cases}
    d'(x, u^*) := d(x, \calU) &\forall x \in \calP \\
    d'(x,y) := \min\{ d'(x,u^*) + d'(y,u^*), d(x,y)  \} &\forall x,y \in \calP
\end{cases}    
\end{equation}
We define the weight of a point $x \in \calP' - u^*$ the same weight of $x$ in $\calP$.
Finally, we set $w(u^*)$ to be very large say $w(u^*) = \beta nW\Delta$, where $W$ is the maximum weight of points in $\calP$.
This makes any $\beta$ approximate solution for $(s+1)$-median problem on $\calP'$ to contain $u^*$ by force.

\begin{claim}\label{dprime-is-metric}
    $(\calP', d')$ is a metric space.
\end{claim}

\begin{proof}

Intuitively, we can derive $d'$ by the metric of shortest path in the following graph.
For each $p \in \calP$ we have a node.
The weight of the edge $pq$ for $p,q \in \calU$ is $0$ and the weight of any other two point is their $d$ distance.
Now, $d'$ exactly equals to the metric derived by the shortest in this graph, which is well-known to be a metric.
By the way, we provide another proof here that does not depend in this fact and is self-contained.

Now, we proceed with the formal proof.
The only non trivial property about $d'$ is the triangle inequality. We need to show that for every $x,y,z \in \calP'$, we have
$$ d'(x,y) \leq d'(x,z) + d'(y,z). $$
If at least two of $x,y$ and $z$ are equal, the inequality is trivial. Otherwise, by symmetry between $x$ and $y$, we have one of the following three cases.

\noindent
\textbf{Case 1:} $z = u^*$ and $x,y \neq u^*$.

This case is trivial by the definition of $d'$ since
$$d'(x,y) = \min\{ d'(x,u^*) + d'(y,u^*), d(x,y) \} \leq d'(x,u^*) + d'(y,u^*) = d'(x,z) + d'(y,z). $$

\noindent
\textbf{Case 2:} $y = u^*$ and $x,z \neq u^*$.

Assume $u_x,u_z \in \calU$ are the projection of $x$ and $z$ on $\calU$ respectively with respect to metric $d$. So,
$d(x,\calU) = d(x,u_x)$, $d(z,\calU) = d(z,u_z)$ and $d(x,u_x) \leq d(x,u_z)$.
Note that
$d'(x,y) = d'(x,u^*) = d(x,\calU) = d(x,u_x)$ and
$d'(z,y) = d'(z,u^*) = d(z,\calU) = d(z,u_z)$.
As a result,
\begin{eqnarray*}
    d'(x,y)
    &=& d(x,u_x)  \\
    &\leq& d(x,u_z) \\
    &\leq& d(x,z) + d(z,u_z) \\
    &=& d(x,z) + d'(z,y).
\end{eqnarray*}
The first inequality is because of the choice of $u_x$ in $\calU$ and the second inequality is triangle inequality for metric $d$.
So, we have the following
\begin{equation}\label{eq:new-metic-case2-eq1}
    d'(x,y) \leq d(x,z) + d'(z,y).
\end{equation}

Since $d'$ is non negative, we also have
\begin{equation}\label{eq:new-metic-case2-eq2}
   d'(x,y) = d'(x,u^*) \leq \left( d'(x,u^*) + d'(z,u^*) \right) + d'(z,y) 
\end{equation}
Combining \Cref{eq:new-metic-case2-eq1} and \Cref{eq:new-metic-case2-eq2} we have
$$ d'(x,y) \leq \min \{ d'(x,u^*) + d'(z,u^*) , d(x,z) \}  + d'(z,y) = d'(x,z) + d'(z,y). $$

\noindent
\textbf{Case 3:} $x,y,z \neq u^*$.
        
According to the definition of $d'$ we have one of the following sub cases.
\begin{itemize}
    \item $d'(y,z) = d'(y,u^*) + d'(z,u^*)$.
    
    By the second case that we have already proved, we know that
    $$ d'(x,u^*) \leq d'(x,z) + d'(z,u^*). $$
    Hence,
    \begin{eqnarray*}
      d'(x,y) &=& \min \{ d'(x,u^*) + d'(y,u^*) , d(x,y) \} \\
      &\leq& d'(x,u^*) + d'(y,u^*) \\
      &\leq& d'(x,z) + d'(z,u^*) + d'(y,u^*) \\
      &=& d'(x,z) + d'(y,z).  
    \end{eqnarray*}
    
    \item $d'(x,z) = d'(x,u^*) + d'(z,u^*)$.
    
    By symmetry between $x$ and $y$, this sub case is similar to the previous sub case.
    
    \item $d'(y,z) \neq d'(y,u^*) + d'(z,u^*)$ and $d'(x,z) \neq d'(x,u^*) + d'(z,u^*)$.

    According to the definition of $d'(x,z)$ and $d'(y,z)$ we have
    $d'(x,z) = d(x,z)$ and $d'(y,z) = d(y,z)$.
    Then, be triangle inequality for $d$ we have
    \begin{eqnarray*}
        d'(x,y) &=& \min \{ d'(x,u^*) + d'(y,u^*) , d(x,y) \} \\
        &\leq& d(x,y) \\
        &\leq& d(x,z) + d(y,z) \\
        &=& d'(x,z) + d'(y,z).
    \end{eqnarray*}
\end{itemize}

In all cases, we proved $ d'(x,y) \leq d'(x,z) + d'(y,z) $ which concludes $(\calP', d')$ is a metric space.
\end{proof}

Next, we run any algorithm for $(s+1)$-median problem on $\calP'$ w.r.t.~metric $d'$ to find a $\calF$ of size at most $(s+1)$ which is a $\beta$ approximation for $\OPT_{s+1}(\calP')$.
Finally, we let $\tilde{\calU} = \calU + (\calF - u^*)$.
Note that all of the points of $\calF - u^*$ are present in $\calP$ and also the size of $\tilde{\calU}$ is at most $k+s$ since $u^* \in \calF$.

\begin{algorithm}[H]
  \DontPrintSemicolon
  $\calP' \gets (\calP - \calU) + u^*$.
 
  $w(u^*) \gets \beta n W \Delta$.
    \tcp{$W$ is maximum weight of points in $\calP$ and $\beta$ is defined in \Cref{parameters}}

   Consider $D'$ as an oracle to the metric $d'$ defined in \Cref{def:new-metric}.

   $\calF \gets $ any $\beta$ approximate solution for $\OPT_{s+1}(\calP')$ with access to distance oracle $D'$.

   $\tilde{\calU} \gets \calU + (\calF - u^*)$.

   \Return $\tilde{\calU}$.
  \caption{\DevelopCenters$(\calP, \calU, s)$}
\end{algorithm}

It is not obvious how the cost of a set w.r.t.~metrics $d$ and $d'$ relate to each other. But, we will show that solution $\calF$ w.r.t.~metric $d'$ is going to be a good solution for our purpose w.r.t.~metric $d$ as well.
So, we show the following guarantee for this subroutine.

\begin{lemma}\label{lem:guarantee-develop-centers}
    If $\tilde{\calU}$ is returned by \DevelopCenters$(\calP,\calU,s)$, then we have
    $$\cost{\tilde{\calU},\calP} \leq \beta\ \min_{\substack{\calF \subseteq \calP, \\ |\calF| \leq s}} \cost{\calU + \calF, \calP}.\footnote{Recall $\beta$ from \Cref{parameters}} $$
\end{lemma}

\begin{proof}

First, we show how the objective function of two solutions with respect to metrics $d$ and $d'$ relate to each other.
Denote $\costd{X,\calP}$ to be the cost of $X$ with respect to metric $d$ on $\calP$, i.e.~$\sum_{p \in \calP} w(p)d(p,X)$, and $\costdp{X,\calP'}$ to be the cost of $X$ with respect to metric $d'$ on $\calP'$, i.e.~$\sum_{p \in \calP'} w(p)d'(p,X)$.
The following claim relates metrics $d$ and $d'$.

\begin{claim}\label{lem:costd-costdp}
    Assume $\calF \subseteq \calP$. Then, the following holds.
    $$ \costd{\calF + \calU, \calP} = \costdp{\calF + u^*, \calP'}. $$
\end{claim}

\begin{proof}

Note that $\calP' = (\calP-\calU) + u^*$. Since $d'(u^* , \calF + u^*) = 0$ and $d(u,\calF+\calU)=0$ for every $u \in \calU$, it suffices to show that for each $x \in \calP-\calU$, we have $d(x, \calF + \calU) = d'(x, \calF + u^*) $.
Since
$d(x, \calF+\calU) = \min\{ d(x, \calF), d(x, \calU) \}$, we have one of the following two cases.

\noindent
\textbf{Case 1: $d(x, \calU) \leq  d(x, \calF)$.}
This means $d'(x,u^*) \leq d(x,\calF)$. As a result, for every $f \in \calF$ we have
\begin{equation}\label{eq:lem-costd'=costd-case1-eq1}
    d'(x,u^*) \leq d(x,f)
\end{equation}
We also have the following since $d'$ is non negative.
\begin{equation}\label{eq:lem-costd'=costd-case1-eq2}
    d'(x,u^*) \leq d'(x,u^*) + d'(f,u^*)
\end{equation}
Combining \Cref{eq:lem-costd'=costd-case1-eq1} and \Cref{eq:lem-costd'=costd-case1-eq2}, for every $f \in \calF$, we have
$$ d'(x,u^*)  \leq \min \{  d'(x,u^*) + d'(f,u^*), d(x,f)\} = d'(x,f). $$
Hence, $d'(x, \calF +u^*) = d'(x,u^*)$ implying that
$$ d'(x, \calF +u^*) = d'(x,u^*) = d(x,\calU) = d(x, \calF + \calU). $$
The last equality is because of the assumption of this case.

\noindent
\textbf{Case 2: $d(x, \calU) > d(x, \calF)$.}        
Assume $f_x \in \calF$ is the projection of $x$ onto $\calF$ with respect to metric $d$.
We are going to show that $f_x$ is also the projection of $x$ onto $\calF + u^*$ with respect to metric $d'$.

We have $d(x,f_x) = d(x,\calF) \leq d(x, \calU) = d'(x,u^*)$.
So,
\begin{equation}\label{eq:lem-costd'=costd-case2-eq1}
    d(x,f_x) \leq d'(x,u^*)
\end{equation}
By non negativity of $d'$, we have
\begin{equation*}
    d(x,f_x) \stackrel{(\ref{eq:lem-costd'=costd-case2-eq1})}{\leq}
     d'(x,u^*) \leq d'(x,u^*) + d'(f_x,u^*),
\end{equation*}
which concludes
$d(x,f_x) = \min \{ d'(x,u^*) + d'(f_x,u^*), d(x,f_x) \} = d'(x,f_x)$. So,
\begin{equation}\label{eq:lem-costd'=costd-case2-eq2}
  d(x,f_x) = d'(x,f_x). 
\end{equation}

Now, assume $f \in \calF$ is arbitrary. We have
\begin{equation}\label{eq:lem-costd'=costd-case2-eq3}
   d'(x,f_x)  \stackrel{(\ref{eq:lem-costd'=costd-case2-eq2})}{=} d(x,f_x) \stackrel{(\ref{eq:lem-costd'=costd-case2-eq1})}{\leq} d'(x,u^*) \leq  d'(x,u^*) + d'(f,u^*),  
\end{equation}
as well as
\begin{equation}\label{eq:lem-costd'=costd-case2-eq4}
   d'(x,f_x) \stackrel{(\ref{eq:lem-costd'=costd-case2-eq2})}{=} d(x,f_x) \leq d(x,f).
\end{equation}
The inequality in \Cref{eq:lem-costd'=costd-case2-eq4} is because of the choice of $f_x$.
Combining \Cref{eq:lem-costd'=costd-case2-eq3} and \Cref{eq:lem-costd'=costd-case2-eq4}, we have
$$ d'(x,f_x) \leq \min\{ d'(x,u^*) + d'(f,u^*), d(x,f) \} = d'(x,f). $$
As a result, for every $f \in \calF$, we have $d'(x,f_x) \leq d'(x,f)$. Combining with $d'(x,f_x) \leq d'(x,u^*)$ (because of \Cref{eq:lem-costd'=costd-case2-eq1} and \Cref{eq:lem-costd'=costd-case2-eq2}), we conclude
$d'(x,\calF + u^*) = d'(x,f_x)$. Hence,
$$ d(x,\calF + \calU) = d(x,\calF) = d(x,f_x) \stackrel{(\ref{eq:lem-costd'=costd-case2-eq2})}{=} d'(x,f_x) = d'(x,\calF + u^*). $$
The first equality is because of the assumption of this case that $d(x, \calU) > d(x, \calF)$.

\noindent
\textbf{Objective Function:}
In each case, we showed $d(x, \calF + \calU) = d'(x, \calF + u^*) $ which concludes the lemma as follows.
\begin{eqnarray*}
    \costd{\calF + \calU, \calP} &=& \sum\limits_{x \in \calP} w(x) \cdot d(x, \calF + \calU) \\
    &=& \sum\limits_{x \in \calP-\calU} w(x) \cdot d(x, \calF + \calU) \\
    &=& \sum\limits_{x \in \calP-\calU}w(x) \cdot d'(x, \calF + u^*) \\
    &=&
    \sum\limits_{x \in (\calP-\calU) + u^*}w(x) \cdot d'(x, \calF + u^*) \\
    &=&
    \costdp{\calF + u^*, \calP'}.
\end{eqnarray*}
Note that we used $ d'(u^*,\calF+u^*) = 0$ and $d(x,\calF+\calU) = 0 $ for all $x \in \calU$.

\end{proof}

Now, assume $\calF$ is the $\beta$ approximate solution found by \DevelopCenters \ in line 4. 
Let
$$ \calF^* = \arg\min_{\substack{\calF \subseteq \calP, \\ |\calF| \leq s}} \left\{ \cost{\calF + \calU, \calP} \right\}. $$
We have
\begin{eqnarray*}
   \costd{\calF + \calU, \calP}  &=&   \costdp{\calF + u^*, \calP'} \\
   &=& \costdp{\calF, \calP'} \\
   &\leq& \beta\ \OPT^{d'}_{s+1}(\calP') \\
   &\leq& \beta\  \costdp{\calF^* + u^*, \calP'} \\
   &=& \beta\  \costd{\calF^* + \calU, \calP},
\end{eqnarray*}
where $\OPT^{d'}_{s+1}(\calP')$ is the optimum $(s+1)$-median w.r.t.~metric $d'$ in space $\calP'$.
The first and the last equalities follow from \Cref{lem:costd-costdp}.
The second equality follows from $u^* \in \calF$ (since $w(u^*)$ is too large and since $\calF$ is $\beta$-approximate it should contain this point).
The first inequality follows from the approximation ratio of solution $\calF$.
The last inequality follows from the fact that $\calF^* + u^*$ is a set of size at most $s+1$.
Then we have the following as desired.
\begin{equation*}
   \cost{\tu, \calP} = \costd{\calF + \calU, \calP} \leq \beta\ \costd{\calF^* + \calU, \calP}
   = 
   \beta\ \min_{\substack{\calF \subseteq \calP, \\ |\calF| \leq s}} \cost{\calU + \calF, \calP}.
\end{equation*}
\end{proof}

\subsection{\FastOneMedian}\label{sec:fast-one-median}

The aim of this algorithm is to solve the $1$-median problem on a set $B$ as fast as possible.
In order to find a single center which is a $3$ approximate of $\OPT_1(B)$, we sample a set $S$ of $\Theta(\log n)$ points from $B$ independently.
The sampling distribution is such that every $p \in B$ is sampled proportional to its weight, i.e.~with probability $w(p)/w(B)$.
Then, we find the best center between these sampled points.
So, we have the following algorithm.

\begin{algorithm}[H]
  \DontPrintSemicolon
  $B \gets \{ b_1,b_2, \ldots , b_{|B|} \}$ an arbitrary ordering.
  
  $S \gets \emptyset$.
  
  \For{$\tilde{\Theta}(\log n)$ iterations}{
    $r \gets [0,1]$ uniformly at random.
    
    $i^* \gets $ smallest index such that $\sum_{j=1}^{i^*} w(b_j) \geq r \cdot w(B)$.
    
    $ S \gets S + b_{i^*}$.
  }

  $q \gets \arg\min\{ s \in S \mid \cost{s,B} \} $.

  \Return $q$.
  \caption{\FastOneMedian$(B)$}
\end{algorithm}

We show the following guarantee of this algorithm.

\begin{lemma}\label{lem:guarantee-fast-one-median}
    If $q$ is returned by \FastOneMedian$(B)$, then with high probability,
    $$ \cost{q,B} \leq 3 \ \OPT_1(B). $$
\end{lemma}

Note that in this algorithm, $q$ is a proper solution for $1$-median problem on $B$, but in this lemma $\OPT_1(B)$ stands for the optimum improper solution.
So, the solution of the algorithm is compared with the best improper solution, although it returns a proper solution.

\begin{proof}

Let $m = |B|$ and $z^*$ be the optimal solution to the improper $1$-median problem in $B$.
If $y$ is sampled from the distribution where the probability of choosing $p$ is $w(p)/w(B)$, then we have
\begin{eqnarray*}
   \mathbb E \left[ \cost{y,B} \right] 
   &=&
   \sum_{y \in B} \frac{w(y)}{w(B)} \ \cost{y,B}
   =
   \sum_{y \in B} \frac{w(y)}{w(B)} \sum_{x \in B} w(x) \cdot d(x,y) \\ 
   &=&
   \frac{1}{w(B)}
   \sum_{y \in B} \sum_{x \in B} w(y)w(x) \cdot d(x,y) \\ 
   &\leq&
   \frac{1}{w(B)}
   \sum_{y \in B} \sum_{x \in B} w(y)w(x) \cdot  (d(x,z^*) + d(z^*,y)) \\ 
   &=& \frac{1}{w(B)} \left[
    \sum_{y \in B} w(y)  \cdot  \sum_{x \in B} w(x) \cdot d(x,z^*) +
    \sum_{x \in B} w(x)  \cdot  \sum_{y \in B} w(y) \cdot d(y,z^*) \right] \\
   &=& \sum_{x \in B} w(x)  \cdot d(x, z^*) + \sum_{y \in B} w(y)  \cdot d(y, z^*) \\
   &=&2\ \OPT_1(B).
\end{eqnarray*}
Thus, the expected approximation ratio of the $1$-median solution $y$ is at most $2$. Hence, by Markov's inequality, the probability that the approximation ratio of sample $y$ is bigger than $3$, is at most $2/3$.
Finally, since \FastOneMedian \ samples a set $S$ of $\Theta(\log n)$ points from $B$ independently,
the probability that all of these samples fail to have approximation ratio less than $3$ is at most $(2/3)^{|S|} = 1/n^{\Theta(1)}$.
As a result, with high probability there exists some $y \in S$ such that $\cost{y,B} \leq 3 \ \OPT_1(B)$.

\end{proof}

\section{Description of the Algorithm}

In this section, we provide description of the algorithm as well as the main ideas and arguments on why the algorithm works correctly and how we achieve approximation, recourse and running time bounds.

The algorithm divides the input stream into some  epochs on the fly.
At the beginning of each epoch we have a solution $\calU^{(0)}$ of $k$ centers for the current space $\calP^{(0)}$.
The solution is $50\beta=O(1)$\footnote{Recall $\beta$ from \Cref{parameters}} approximate and is robust with respect to $\calP^{(0)}$.
Throughout the epoch we always maintain a $ O(\beta) = O(1) $ approximate solution.
Then, at the end of the epoch we build a solution using only the initial solution $\uo $ of this epoch.
So, this new solution is independent from the solutions maintained during the epoch.
This new solution is going to be $50\beta$ approximate and robust with respect to the new space and this is going to be the starting point of the next epoch. Now, we explain the algorithm on each epoch.

At the beginning of the epoch we find a value $l$ such that 
$$ \Omega(1) \ \OPT_{k-l}(\po) \leq \OPT_{k}(\po) \leq O(1) \ \OPT_{k+l}(\po). $$
Note that the value of $l$ could be zero.
Next, we will find a subset  of $\uo$ of size $ k-l $ like $\calU'$ whose cost is at most $O(1) \ \OPT_{k}(\po) $.
We call this part of the algorithm \RemoveCenters.
So, after this part, we have $l$ and $\calU' \subseteq \uo$ of size $k-l$ with above guarantees.

Next, for $l$ many iterations, we do lazy updates on $\calU'$ as follows.
If a point gets inserted, we add it to $\calU'$, and if a point gets deleted we do nothing (since we are dealing with the improper case, we can maintain a center in our main solution even if it gets deleted from the space).
Obviously, during these updates the size of $\calU'$ is at most $k$ which means our solution is valid.
We will prove for $l$ many iterations, this solution is a constant approximation for the value of $\OPT_k$
on the current space.
We call this part of the algorithm \LazyUpdates.

Then, we do one more update to get the final space $\calP^{(l+1)}$ in this epoch.
Now, we find a $10\beta$ approximate solution $\tuo$ for $\OPT_{k+l+1}(\po)$ by adding $O(l+1)$ many centers to $\uo$. Note that the solution might have more than $k+l+1$ centers, although our goal is to compete with $\OPT_{k+l+1}(\po)$.
This is not obvious why such a solution exists.
But, we show that it is possible and actually, the subroutine \DevelopCenters \ introduced in \Cref{sec:develop-centers} finds such a solution for this part.

Next, we add all of the inserted points during these $l+1$ updates to $\tuo$ i.e.~$\calV = \tuo + (\pl - \po)$.
This is because we want a good solution for the new set of points $\pl$ and newly added points might be costly for $\tuo$.
Note that if the weight of a point $p \in \po \cap \pl$ is changed (i.e.~it gets deleted and then inserted with another weight), we consider it in $\pl - \po$.
Next, we reduce the size of $\calV$ to get a subset of $k$ centers $\calW$ that is going to be a $ 32\beta$ approximate solution for $\OPT_k(\pl)$.
We do this part by calling subroutine \RandLocalSearch \ introduced in \Cref{sec:rand-local-search}.

So far, we have a solution $\calW$ for the new set of point $\pl$ that satisfies $ |\calW \oplus \calU^{(0)} | = O(l+1)$ and is $32\beta$ approximate for $\OPT_k(\pl)$.
Next, we make solution $\calW$ robust w.r.t.~the new space $\calP^{(l+1)}$ and argue that it is actually a $ 50\beta$ approximate solution.
We call this part of the algorithm \Robustify.
Finally, we can start the next epoch with this new solution of $k$ centers, since it is $50\beta$ approximate and robust w.r.t.~the new space.

So, we have the following pseudo-code for what we do in an epoch in a high level. $\calP^{(l)}$ in line 3 denotes the space after $l$ updates in this epoch. 

\begin{algorithm}[H]
  \DontPrintSemicolon
  $\left( \calU', l \right) \gets $ \RemoveCenters$(\uo)$.

  \LazyUpdates$(\calU', l)$.

  $\pl \gets \calP^{(l)} \pm p_{l+1}$. \tcp{read the $(l+1)^{th}$ update.}

  $\tuo \gets $ \DevelopCenters$(\po, \uo, (8C+2)(l+1))$.
  \tcp{$C = 12 \cdot 3 \cdot 10^5\gamma\beta^2$.}

  $\calV \gets \tuo + (\pl - \po)$.

  $\calW \gets $ \RandLocalSearch$(\calV, |\calV| - k)$.

  $\calU^{(l+1)} \gets $ \Robustify$(\po, \pl, \uo, \calW)$.

  \caption{\MainEp($\po, \uo$)}
\end{algorithm}

\subsection{\RemoveCenters}

Assume
\begin{equation}\label{def:l^*}
    l^* = \arg\max \left\{ 0 \leq l \leq k \;\middle|\;  \OPT_{k-l}(\po) \leq 400\gamma\beta \ \OPT_k(\po) \right\}.
\end{equation}
Note that $\OPT_k$ is considered in the improper case.
We will find a value $l'$ such that $ l'  = \Omega(l^*) $ and 
\begin{equation}\label{eq:cost-reduce-sol}
    \OPT_{k-l'}(\po) \leq 3 \cdot 10^5\gamma\beta^2 \ \OPT_k(\po).
\end{equation}
Let $\calR = \{ 0, 2^0, 2^1,2^2,\ldots,2^{\lfloor \log k \rfloor} \}$.
We iterate over $r \in \calR$ in increasing order starting from $0$, and call \RandLocalSearch$(\uo, r)$ \ introduced in \Cref{sec:rand-local-search}.
If the cost of the subset returned by this subroutine is less than $14 \cdot 400\gamma\beta \ \cost{\uo, \po}$, we take the next $r \in \calR$ and repeat the process.
We do this until reaching the maximum $r'$ such that $\calU_{k-r'} =$ \RandLocalSearch$(\uo, r')$ satisfies
$$ \cost{\calU_{k-r'}, \po} \leq 14 \cdot 400\gamma\beta \cost{\uo, \po}. $$
Then, we let $l' = r'$ (note that for simplicity, we have written this in the pseudo-code in a different way).
Later, we will show that this $l'$ satisfies \Cref{eq:cost-reduce-sol} and $l' \geq l^*/2$ with high probability.
Next, we let $l = \lfloor l' / C \rfloor$ where $C$ is $ 12 \cdot 400 \gamma\beta\alpha^2$ and is determined using the Double-Sided Stability \Cref{lem:double-sided-stability}.
Note that $ l+1 = \Omega(l'+1) = \Omega(l^*+1) $ since $C$ is a constant.
We will prove that for this value of $l$ we have the following guarantee.
$$ \frac{1}{3 \cdot 10^5\gamma\beta^2} \ \OPT_{k-l}(\po) \leq \OPT_{k}(\po) \leq 4 \ \OPT_{k+l}(\po). $$

Next, we find a subset $\calU' \subseteq \uo$ of size $k-l$ such that
\begin{equation}\label{eq:cost-Uprime-in-description-section}
  \cost{\calU', \po} \leq 2 \cdot 10^7\gamma\beta^2 \ \OPT_{k+l}(\po).  
\end{equation}
We do this by calling \RandLocalSearch$(\po, \uo, l)$.
So, we have the following algorithm for \RemoveCenters.

\begin{algorithm}[H]
  \DontPrintSemicolon
  $\calR \gets \{0, 2^0,2^1,2^2,\ldots,2^{\lfloor \log k \rfloor} \}$.

  \For{$r \in \calR$}{
    $\calU_{k-r} \gets $ \RandLocalSearch$(\po, \uo, r)$.

    \If{$\cost{\calU_{k-r}, \po} > 14 \cdot 400\gamma\beta \ \cost{\uo, \po}$}{
        break \textbf{for} loop.
    }
  }
  $l' \gets \lfloor r/2 \rfloor$.

  $l \gets \lfloor l' / C \rfloor$ where $C = 12 \cdot 3 \cdot 10^5\gamma\beta^2$.

  $\calU' \gets$ \RandLocalSearch$(\calU, l)$.

  \Return $\left( \calU', l\right)$.
  \caption{\RemoveCenters($\calU$)}
\end{algorithm}

\subsection{\LazyUpdates}
So far, we have $\calU'$ of size $k-l$.
Now we perform $l$ many updates in the input stream as follows.
Whenever a point $p$ gets inserted to space, we add it to $\calU'$ and if a point $p$ gets deleted, we keep $\calU'$ unchanged.
Note that since we deal with the improper dynamic $k$-median, we can maintain a center even if it gets deleted from the space.
In order to prove this solution is actually a constant approximation of $\OPT_k(\calP)$, where $\calP$ is the current space during these lazy updates, we combine the guarantee in \Cref{eq:cost-Uprime-in-description-section} and Lazy Updates \Cref{lem:lazy-updates}.
This concludes our solution is always a constant approximation of $\OPT_k(\calP)$ for the current space $\calP$ during the next $l$ updates.

\begin{algorithm}[H]
  \DontPrintSemicolon
   \For{$i=1$ to $l$}{
    Read an update in the input stream.
    
    \If{this update is insertion of $p$}{
        $\calU' \gets \calU' + p$.
    }
   }
  \caption{\LazyUpdates$(\calU', l)$}
\end{algorithm}

\subsection{\DevelopCenters}
This part starts after the $(l+1)$'th update in the input stream.
So, we have the final set of points $\pl$.
We will add at most $(8C+2)(l+1) = O(l+1)$ centers to $\uo$ to find a set $\tuo$ such that
\begin{equation}\label{eq:tuo-must-be-10-beta-approx-for-OPT-kl1}
  \cost{\tuo, \po} \leq 10\beta \ \OPT_{k+l+1}(\po).  
\end{equation}
We do this by calling \DevelopCenters$(\po,\uo, (8C+2)(l+1))$ introduced in \Cref{sec:develop-centers}.
According to the guarantee of this subroutine, we get a $\tuo$ such that
$$ \cost{\tuo, \po} \leq \beta\ \min_{\substack{\calF \subseteq \po, \\ |\calF| \leq (8C+2)(l+1)}} \cost{\uo + \calF, \po}. $$
In order to show that this $\tuo$ actually satisfies \Cref{eq:tuo-must-be-10-beta-approx-for-OPT-kl1}, it suffices to show that there exist $\calF \subseteq \po$ of size at most $(8C+2)(l+1)$ such that
$ \cost{\uo + \calF, \po} \leq 10 \ \OPT_{k+l+1}(\po). $
We conclude this existential analysis by constructing $\calF$ using an optimum solution $\calV^*_{k+l+1}$  for $(k+l+1)$-median problem on $\po$.
We show that it suffices to add $(8C+2)(l+1)$ many centers of $\calV^*_{k+l+1}$ to $\uo$ to get a constant approximation of $\OPT_{k+l+1}$.
These centers are elements of $\calV^*_{k+l+1}$ that does not form well-separated pair with any center in $\uo$.

After termination of \DevelopCenters$(\po, \uo, (8C+2)(l+1))$, we let
$\calV = \tuo + (\pl - \po)$. Note that we also consider points whose weight is changed. This concludes that $\calV$ is a good set of centers in the new space $\pl$ as well as the initial space $\po$.
The eventual conclusion is that
$$ \cost{\calV, \pl} \leq 10\beta \ \OPT_{k}(\pl). $$

Next, we reduce the size of $\calV$ to $k$ in order to provide a feasible solution for $\OPT_k(\pl)$.
This part of the algorithm is done by calling  \RandLocalSearch \ on $\calV$ to find a subset $\calW \subseteq \calV$ of size $k$.
Since we know that the cost of $\calV$ on $\pl$ is at most $10\beta\  \OPT_{k}(\pl)$, with the guarantee of \RandLocalSearch, we conclude that the cost of $\calW$ on $\pl$ is at most $ 32\beta \ \OPT_{k}(\pl)$ with high probability.

\subsection{\Robustify}\label{sec:robustify-implementation-full-version}

So far we have a solution $\calW$ which is $32\beta$ approximate for $\OPT_{k}(\pl)$.
The aim of this part of the algorithm is to make $\calW$, robust w.r.t.~the new space $\pl$.
Before we explain \Robustify, we define an integer $t(u)$ for every center $u$. Throughout the algorithm, if $\calU$ is our main solution, we always have an integer $t(u)$ together with center $u$ which indicates that  $u$ is $t(u)$-robust w.r.t.~the current space.
So, at the beginning of the epoch we know that every $u \in \uo$ is $t(u)$-robust w.r.t.~$\po$.

Now, we explain how we make a robust solution w.r.t.~$\pl$ by changing $\calW$.
Recall from \Cref{def:bounded-robust} that $\calW$ is robust w.r.t.~$\pl$, whenever for all $u \in \calW$, we have the following condition
\begin{equation}\label{condition-robust-center}
    u\ \text{is}\ t\text{-robust w.r.t.}\ \pl, \ \text{where}\  t \ \text{is the smallest  integer satisfying}\  
    10^t \geq d(u,\calW - u )/ 200.
\end{equation}
In this subroutine, we have a set $\calS$ that helps us keep track of which elements violate Condition \ref{condition-robust-center}.
First of all, we find some centers in $\calW$ that might not satisfy this condition and put them into set $\calS$.
We do this by calling \FindSuspects.
Centers detected by \FindSuspects \ are going to be all centers $u \in \calW$ that $u$ might not be $t(u)$-robust w.r.t.~the new space $\pl$, i.e.~if a center $u \in \calW$ is not detected by \FindSuspects, then we are sure that $u$ is $t(u)$-robust w.r.t.~$\pl$.
Note that these centers found by \FindSuspects \ are not necessarily all of the centers violating Condition~\ref{condition-robust-center}.

Then, we make all of the centers in $\calS$ robust, by calling \EmptyS.
In this subroutine, for every $w \in \calS$,
we determine the smallest integer $t$ such that $10^t \geq d(w,\calW - w ) / 100$
\footnote{Note that the denominator is 100 (not 200 as in Condition \ref{condition-robust-center}), this is because we want to prevent multiple calls to \MakeRbst \ for a center. We use this in \Cref{lem:robustify-calls-once}.}
and replace $w$ by a center $w_0$ via a call to \MakeRbst$(w, t)$.
Now, this $w_0$ satisfies Condition \ref{condition-robust-center}.
We also remove $w$ from $\calS$, define $t(w_0) = t$ and save this integer with center $w_0$.

At this point, we are sure that every $u \in \calW$ is $t(u)$-robust w.r.t.~$\pl$.
So, in order to figure out if a center $u \in \calW$ violates Condition \ref{condition-robust-center}, we make use of the integer $t(u)$.
If $t$ is the smallest integer satisfying $10^t \geq d(u,\calW - u ) / 200$ and $t \leq t(u)$, then $u$ definitely satisfies Condition \ref{condition-robust-center}.

Note that since $\calW$ is changing, for every $w \in \calW$, the value of $d(w,\calW -w )/ 200$ is changing as well.
As a result, after each change in $\calW$, we should check whether or not any center in $\calW$ violates Condition \ref{condition-robust-center}.
If such a center exists, we add it to $\calS$.
We keep adding new elements to $\calS$ and removing elements from $\calS$ (by \EmptyS), until all of the centers in the current $\calW$ satisfy Condition \ref{condition-robust-center} and $\calS$ in turn is empty.
So, we have the following algorithm.

\begin{algorithm}[H]
  \DontPrintSemicolon
  $\calS \gets$ \FindSuspects$(\po,\pl,\calW,\uo)$.

  \While{true}{
    \If{$\calS = \emptyset$}{
        \Return
    }
    \EmptyS$(\calS,\calW,\pl)$.\label{emptyS-line-robustify}
    
    \For{$u \in \calW$}{ \label{for-loop-robustify}
        $t \gets \lceil \log \left( d(u,\calW -u) / 200 \right) \rceil  $.

        \If{$t > t(u)$}{\label{line:add-condition}
            $\calS \gets \calS + u$.
        }
    }  
  }
  \caption{\Robustify$(\po,\pl,\calW,\uo)$}
\end{algorithm}

It is not obvious why this algorithm terminates.
But, we show that \Robustify \ calls \MakeRbst \ on each center $w \in \calW$ at most once during a single call to \Robustify.
% This is because we take the smallest integer $t$ satisfying
% $ 10^t \geq d(u,\calW -u) / 100$ for a call to \MakeRbst, but the condition that we need for $u$ is being $t$-robust where $t$ is the smallest integer satisfying
% $ 10^t \geq d(u,\calW -u) / 200$.
% Intuitively, the distance of $u$ to other centers can not

\begin{algorithm}[H]
  \DontPrintSemicolon
  pick $w$ arbitrarily from $\calS$.

  $\calS \gets \calS - w$.

  $ t \gets \lceil \log \left(  d(w,\calW -w) / 100 \right) \rceil  $.

  \If{$t \geq 0$}{
    $ w_0 \gets \text{\MakeRbst}(\pl, w, t)$.
  }\Else{
    $w_0 \gets w$.

    $t \gets 0$.
  }

  $\calW \gets \calW - w + w_0$.

  $ t(w_0) \gets t $.
  \caption{\EmptyS$(\calS,\calW,\pl)$}
\end{algorithm}

In order to analyse the recourse, we show that the number of points added to $\calS$ by calling \FindSuspects \ in the beginning of \Robustify \ is at most $\tilde{O}(l+1)$. Then, we proceed by proving that the total recourse of the algorithm is $\tilde{O}(1)$ in amortization.
Note that during a single call to \Robustify \ the recourse might be as large as $k$, although we show the amortized recourse is $\tilde{O}(1)$ throughout the total run of the algorithm (not just one epoch).

In order to analyse why the final solution is $50\beta $ approximate, we show that the cost of the solution $\calW$ is not going to increase by more than a factor of $3/2$ by calling \Robustify.

As a result, we get the final solution of this epoch $\calU$ which is robust and is also $32\beta \times \frac{3}{2} \leq 50\beta$ approximate for $\OPT_k(\pl)$.

\subsubsection{\FindSuspects}

In this subroutine, we find all centers $w \in \calW$ that might not be $t(w)$-robust w.r.t.~the new space $\pl$.
What we are looking for are the following centers.
\begin{itemize}
    \item
    New centers $u$ added to $\uo$ (i.e. $u \in \calW - \uo$).
    \item
    Centers $u \in \uo \cap \calW$ whose neighborhood of radius $2 \cdot 10^{t(u)}$ has changed during this epoch (i.e.~$\text{Ball}_{2\cdot 10^{t(u)}}^{\po}(u) \neq \text{Ball}_{2\cdot 10^{t(u)}}^{\pl}(u)$).
\end{itemize}

Note that in both cases, we also consider the weight of the points. For instance, there might be a point $p$ which is present in both $\text{Ball}_{2\cdot 10^{t(u)}}^{\po}(u)$ and $ \text{Ball}_{2\cdot 10^{t(u)}}^{\pl}(u)$, but its weight has changed.
In this case, we consider center $u$ satisfying the second condition.
So, we find all of the centers $u \in \calW$ satisfying at least one of the above conditions and add them to set $\calS$.

\begin{algorithm}[H]
  \DontPrintSemicolon
  $\calS \gets \calW - \uo$.

  \For{$p \in \pl \oplus \po$}{
        \For{$u \in \uo \cap \calW$}{
            \If{$d(u,p) \leq 2 \cdot 10^{t(u)}$}{
                $\calS \gets \calS + u$.
            }
        }
  }
  \Return $\calS$.
  \caption{\FindSuspects$(\po, \pl, \calW, \uo)$}
\end{algorithm}

Note that in line 2, we also consider points $p \in \po \cap \pl$ whose weight have changed.
According to the definition of $t(u)$, we know that $u \in \uo$ is $t(u)$-robust w.r.t.~$\po$.
We will show that if a center $u$ is not detected by \FindSuspects, i.e.~does not satisfy any of the two above conditions, then it remains $t(u)$-robust with respect to the new space $\pl$. 
This fact helps us argue that when the \Robustify \ subroutine is terminated, the solution $\calU = \Robustify(\po, \pl, \calW, \uo)$ is robust.

\subsubsection{\MakeRbst}
In this subroutine, we get a center $w$ together with a non-negative integer $t$ and return a center $w_0$, such that there exist a $t$-robust sequence like $(w_0,w_1,\ldots, w_{t})$ where $w_{t} = w$.
To do this, we simply use the definition of $t$-robust sequence.
Iterating over $i$ from $t$ to $1$, assume that we have found the point $w_i$ (initially $w_{t} = w$).
First, we check whether 
$$\avcost{w_i, \text{Ball}_{10^i}^\calP(w_i)} \geq 10^i / 5. $$
If this inequality holds, we simply let $w_{i-1} = w_i$.
Otherwise, we find a $3$ approximate solution $q_i$ for $1$-median problem on $\text{Ball}_{10^i}^\calP(w_i)$  using \FastOneMedian \ algorithm introduced in \Cref{sec:fast-one-median}.
Then, we let $w_{i-1} = q_i$ if 
$$ \cost{q_i, \text{Ball}_{10^i}^\calP(w_i)} < \cost{w_i, \text{Ball}_{10^i}^\calP(w_i)}.  $$
Otherwise, we let $w_{i-1} = w_i$.
With guarantees of \FastOneMedian \ algorithm, we prove that we achieve a $t$-robust sequence at the end of \MakeRbst. So, we have the following algorithm.

\begin{algorithm}[H]
  \DontPrintSemicolon
  $w_t \gets w$.

  \For{$i=t$ down to $1$}{
    $B_i \gets \text{Ball}_{10^i}^\calP(w_i)$.

    $z \gets \avcost{w_i, B_i}$.

    \If{$z \geq 10^i/5 $}{
        $w_{i-1} \gets w_i$.
    }\Else{
        $q_i \gets \text{\FastOneMedian}(B_i,w_i) $.

        \If{$\cost{w_i, B_i} \leq \cost{q_i, B_i}$}{
            $w_{i-1} \gets w_i$.
        }\Else{
            $w_{i-1} \gets q_i$.
        }
    }
  }
  \Return $w_0$.
  \caption{\MakeRbst$(\calP, w, t)$}
\end{algorithm}

\begin{remark}
    Instead of \FastOneMedian, we can use the algorithm of \cite{MettuP02} to find a $\beta$-approximate solution $q_i$ for $1$-median problem on $B_i$.\footnote{Recall $\beta$ from \Cref{parameters}}.
    But, if we do this, we lose another multiplicative factor of $\beta$ on the approximation ratio in our final algorithm.
    We can prevent it by performing this simple \FastOneMedian \ subroutine instead of \cite{MettuP02}.
    % Currently, the only place that we are using \cite{MettuP02} as a black box is inside \DevelopCenters.
\end{remark}

    \section{Correctness of the Algorithm}\label{sec:correctness}

Assume we have $\uo$ at the beginning of an epoch where $\uo$ is $50\beta$ approximate and it is also robust.
The subroutine \RemoveCenters, reduces the size of the solution to $k-l$, which means during the \LazyUpdates, we always maintain a solution of size at most $k$ which is feasible.
The approximation guarantee of \RandLocalSearch \ in \Cref{lem:cost-local-search} holds with high probability. But, since we call \RemoveCenters \ at most $T$ times during any $T$ updates of input stream, we conclude that (by a simple union bound) with high probability all of these guarantees hold.

Next, the $(l+1)$'th update occurs.
Then, algorithm adds some centers to $\uo$ and later reduces the size to $k$ to get $\calW$ which is feasible for the new space $\pl$.

\subsection{\Robustify \ Calls  \MakeRbst \ Once for Each Center}

First, we show that \Robustify \ always terminates.

\begin{lemma}[Lemma 6.1 in \cite{soda/FichtenbergerLN21}]\label{lem:robustify-calls-once}
    If \Robustify$(\calW)$ calls a \MakeRbst \ for a center $w \in \calW$ and replaces it with $w_0$, then it will not call \MakeRbst \ on $w_0$ until the end of the \Robustify \ subroutine.
\end{lemma}

\begin{proof}
Assume a \MakeRbst \ call to a center $w \in \calW$ is made and it is replaced by $w_0$. Let $\calW_1$ be the set of centers just before this call to \MakeRbst \ is made on $w$.
At this point of time, $t(w_0)$ is the smallest integer satisfying $10^{t(w_0)} \geq d(w, \calW_1 - w)/100$ (note that $t(w_0) \geq 0$, otherwise the algorithm should not have called \MakeRbst \ duo to line 10 in the \Robustify).
For the sake of contradiction, assume that another call to \MakeRbst \ is made on $w_0$.
We assume that this pair $w$, $w_0$ is the first pair for which this occurs.
Let $w' = \pi_{\calW_1 - w}(w)$ which means $d(w,\calW_1 - w)=d(w,w')$.
Then,
$10^{t(w_0)} \geq d(w,w')/100$ and $10^{t(w_0)} < d(w,w')/10$ by definition of $t(w_0)$. So, by \Cref{lem:robustness-property-1}, we have
\begin{equation}\label{eq:dw0w-leq-dwwprime-over-20}
   d(w_0,w) \leq 10^{t(w_0)}/2 \leq d(w,w')/20. 
\end{equation}
Assume $\calW_2$ is the set of centers just after the call to \MakeRbst \ is made on $w_0$.
Since we assumed $w,w_0$ is the first pair that another call to \MakeRbst \ on $w_0$ is made, we have one of the two following cases.
\begin{itemize}
    \item $w' \in \calW_2$.
    In this case, we have
    \begin{equation*}
        d(w_0,\calW_2 - w_0) \leq d(w_0,w') \leq d(w_0,w) + d(w,w') \leq \frac{21}{20}\ d(w,w') \leq 2 \ d(w,w').
    \end{equation*}
    The last inequality holds by \Cref{eq:dw0w-leq-dwwprime-over-20}.
    This is in contradiction with the assumption that a call to \MakeRbst \ on $w_0$ is made, because $w_0$ is $t(w_0)$-robust where
    $10^{t(w_0)} \geq d(w,w')/100 \geq d(w_0,\calW_2-w_0)/200$. So, $w_0$ must not be selected to be added to $\calS$ in \Robustify.
    
    \item There is $ w'_0 \in \calW_2$ that replaced $w'$ by a single call to \MakeRbst.
    When \MakeRbst \ is called on $w'$, the algorithm picks integer $t(w_0')$ that satisfies
    $ 10^{t(w_0')} \leq d(w',w_0) / 10 $ (otherwise, it should have picked $t(w_0')-1$ instead of $t(w_0')$).
    Hence,
    \begin{equation}\label{eq:dw0primewprime-leq-dwwprime-over-10}
       d(w_0',w') \leq 10^{t(w_0')}/2 \leq d(w',w_0) / 20 \leq \left( d(w_0,w) + d(w,w') \right)/20 \leq d(w,w')/10.  
    \end{equation}
    The first inequality holds by \Cref{lem:robustness-property-1} and the last inequality is \Cref{eq:dw0w-leq-dwwprime-over-20}.
    As a result,
    \begin{eqnarray*}
        d(w_0, \calW_2-w_0) &\leq& d(w_0,w_0') \\
        &\leq& d(w_0,w) + d(w,w') + d(w',w_0') \\
        &\leq& d(w,w')/20 + d(w,w') + d(w,w')/10 \\
        &=& \frac{23}{20}\ d(w,w') \\
        &\leq& 2 \ d(w,w').
    \end{eqnarray*}
    The last inequality holds by \Cref{eq:dw0w-leq-dwwprime-over-20} and \Cref{eq:dw0primewprime-leq-dwwprime-over-10}.
    This leads to the same contradiction.
\end{itemize}

\end{proof}

\subsection{Correctness of \Robustify}

Consider a call to \Robustify \ occurs on $\calW$.
First, in \FindSuspects, all of the centers $u \in \calW$ satisfying $u \in \calW - \uo$ or $d(p,u) \leq 2\cdot 10^{t(u)}$ for at least one $p \in\po \oplus \pl$ are added to $\calS$.
We know that each center $u \in \calW \cap \uo$ is $t(u)$-robust w.r.t.~$\po$ at the beginning of epoch. We now prove the following.

\begin{lemma}\label{lem:not-recognized-by-find-suspect-remain-robust}
    Assume $u \in \uo$ is $t(u)$-robust w.r.t.~$\po$ and $\calS$ is returned by \FindSuspects. If $u \notin \calS$, then $u$ is $t(u)$-robust w.r.t.~$\pl$ as well.
\end{lemma}

\begin{proof}
Let $(p_0,\ldots,p_{t(u)})$ be a $t(u)$-robust sequence w.r.t.~$\po$, where $p_0=u$.
According to \Cref{lem:robustness-property-1},
$d(p_0,p_{t(u)}) \leq 10^{t(u)}/2$ which concludes
$$\text{Ball}_{10^{t(u)}}^{\po}(p_{t(u)}) \subseteq \text{Ball}_{2 \cdot 10^{t(u)}}^{\po}(u). $$
Because, the ball around $u$ of radius $2 \cdot 10^{t(u)}$ is not changed (i.e.~$\text{Ball}_{2 \cdot 10^{t(u)}}^{\po}(u) = \text{Ball}_{2 \cdot 10^{t(u)}}^{\pl}(u)$ considering points by their weight), we conclude that ball around $p_{t(u)}$ of radius $10^{t(u)}$ is also not changed. Finally, with the definition of $t(u)$-robust sequence we conclude that $(p_0,\ldots,p_{t(u)})$ is $t(u)$-robust  w.r.t.~$\pl$ as well. So, $p_0=u$ is $t(u)$-robust  w.r.t.~$\pl$.
\end{proof}

Assume $\calU =$ \Robustify$(\po, \pl, \calW, \uo)$ is the final solution of the algorithm for $\pl$.
According to above lemma and description of the algorithm, we show that $\calU$ is robust.

By \Cref{lem:robustify-calls-once} we know that \Robustify \ calls \MakeRbst \ on each $u \in \calW$ at most once.
According to the proof of this lemma, if a center $u \in \calU$ is returned by a call to \MakeRbst \ on center $u' \in \calW$, then $u$ does not violate Condition \ref{condition-robust-center}.
Now, assume no call to \MakeRbst \ was made on $u \in \calW$.
Hence, $u$ was not found by \FindSuspects, and by \Cref{lem:not-recognized-by-find-suspect-remain-robust}, we conclude $u$ is $t(u)$-robust w.r.t.~the new set of points $\pl$.
Let $t$ be the smallest integer such that
$10^t \geq d(u,\calU - u  )/ 200$.
Since $u$ was not added to $\calS$, we conclude that $t(u) \geq t$ (See Line~\ref{line:add-condition} of \Robustify).
So, $u$ is $t$-robust w.r.t.~$\pl$ which means $u$ does not violate Condition \ref{condition-robust-center} for current $\calU$.
As a result, the final $\calU$ returned by \Robustify \ is robust.

Later, in \Cref{sec:approx}, we show that it is also $50\beta$ approximate.
This means that the final solution of the algorithm at the end of the epoch satisfies the conditions needed as the initial solution of the next epoch.

\subsection{Correctness of \MakeRbst}

Assume a call to \MakeRbst \ on $(\calP,w,t)$.
In order to show that the returned value $w_0$ is actually a $t$-robust point w.r.t.~$\calP$, we  check the definition of $t$-robust sequence.
According to the subroutine \MakeRbst \ and the definition of $t$-robust sequence, we only need to show that for every $1 \leq i \leq t$, the returned point $q_i$ by \FastOneMedian$(B_i,w_i)$ satisfies $\cost{q_i,B_i} \leq 3 \ \OPT_1(B_i).$
This is also the guarantee of \FastOneMedian \ in \Cref{lem:guarantee-fast-one-median}.

Later, in \Cref{sec:recourse}, we show that the total number of times that we call \MakeRbst \ throughout the first $T$ updates of the input stream is $\Tilde{O}(T)$.
Hence, with high probability all of these independent calls to \MakeRbst \ return correctly.

\section{Recourse Analysis}\label{sec:recourse}

In this section, we provide the recourse analysis of the algorithm. We show the following lemma.

\begin{lemma}\label{lem:final-recourse}
    The Total recourse of the algorithm throughout the first $T$ updates is at most
    $$O(T(\log\Delta)^2).$$
\end{lemma}

\subsection{Difference Between $\uo$ and $\calW$}

In \DevelopCenters, we are adding at most $ (8C+2)(l+1)$ centers to $\uo$ to make $\tuo$. Then, we add $\pl-\po$ to get $\calV$. As a result, the size of  $\calV$ is at most $k + (8C+3)(l+1)$. Next, we reduce the size of this set to $k$ by calling \RandLocalSearch \ on $\calV$. Since both of $\uo$ and $\calW$ are subsets of size $k$ of set $\calV$, we conclude that 
$$ |\calW \oplus \uo| \leq 2\cdot((8C+3)(l+1)) \leq 20C(l+1). $$
So, we have the following lemma.

\begin{lemma}\label{lem:dis-uo-calW}
    For every epoch of length $l+1$, we have $|\calW \oplus \uo| \leq 20C(l+1)$.
\end{lemma}

\subsection{Number of Centers Detected by \FindSuspects}

At the beginning of \FindSuspects, the algorithm sets $\calS = \calW - \uo$.
As a result, the size of this initial $\calS$ is at most $20C(l+1)$ by \Cref{lem:dis-uo-calW}.
Next, for every point $p$ that has gone through an update (either insertion or deletion) in this epoch (i.e. $p \in \po \oplus \pl$) the algorithm finds all centers $u \in \calW \cap \uo$ such that $d(u,p) \leq 2 \cdot 10^{t(u)}$.
We claim that for a fix $p \in \po \oplus \pl$, there are at most $O(\log \Delta)$ many centers $u$ that satisfy this property.

\begin{lemma}\label{lem:number-of-invalidation}
    For each $p \in \po \oplus \pl$, the number of $u \in \calW \cap \uo$ satisfying $d(p,u) \leq 2 \cdot 10^{t(u)}$ is at most $O(\log \Delta)$.
\end{lemma}

\begin{proof}
Assume $p \in \po \oplus \pl$.
Let $\mu$ be the number of centers $u \in \calW \cap \uo$ that satisfy 
$d(p,u) \leq 2 \cdot 10^{t(u)}$.
For each of these centers like $u$, there is a $u'$ that was added to $\calS$ and a call to \MakeRbst \ was made on $u'$ which returned $u$.
Assume $\{u_1,u_2,\ldots,u_\mu \}$ is the set of  centers $u$ satisfying $d(p,u) \leq 2 \cdot 10^{t(u)}$ ordered in decreasing order by the time they replaced $u_i'$ via a call to \MakeRbst.
So, when $u_i$ is added to the main solution, center $u_{i+1}$ was already present in the main solution which concludes
$$ 10^{t(u_i)} \leq d(u_i,u_{i+1})/10, $$
by the definition of $t(u_i)$ at that time
for every $1 \leq i \leq \mu-1$.
Hence,
$$ 10^{t(u_{i+1})} \geq d(p, u_{i+1})/2 \geq (d(u_i, u_{i+1}) - d(p, u_i))/2 \geq 5 \cdot 10^{t(u_i)} - 10^{t(u_i)} = 4 \cdot 10^{t(u_i)}. $$
The first and last inequality holds by the assumption that $d(p,u) \leq 2 \cdot 10^{t(u)}$ for every $u \in \{u_1,\ldots,u_\mu\}$.
Finally, this concludes 
$t(u_\mu)> \cdots > t(u_{2}) > t(u_{1})$. Since we assumed the distances between any two point in the space is between $1$ and $\Delta$, we know $ 0 \leq t(u_{i}) \leq \lceil \log \Delta \rceil$ for each $1 \leq i \leq \mu$ and so, $\mu \leq  (\log \Delta) + 2 = O(\log \Delta)$.

\end{proof}

A quick corollary of this lemma is the following.

\begin{corollary}\label{cor:size-calS-find-suspects}
    The size of final $\calS$ returned by \FindSuspects \ in an epoch of length $l+1$ is at most $O((l+1)\log \Delta)$.
\end{corollary}

\begin{proof}
    By \Cref{lem:dis-uo-calW} and \Cref{lem:number-of-invalidation}, we have
    $$|\calS| \leq 10C(l+1) + |\po \oplus \pl|\cdot O(\log \Delta) = O((l+1)\log \Delta). $$
\end{proof}

\subsection{Total Recourse}

Assume the sequence of first $T$ updates in the input stream. We show the total recourse of the algorithm during these $T$ updates is at most $O(T(\log \Delta)^2)$.
Consider that the algorithm performs $E$ epochs in total.
For every $1 \leq e \leq E$, let $l_e+1$, $\uo_e$ and $\calW_e$ be the length of, initial solution and the set of centers returned by \RandLocalSearch \ just before calling \Robustify \ in epoch $e$ respectively. 
For every $1 \leq e \leq E$ and every $1 \leq i \leq l_e+1$, let $\calU^{(i)}_e$ be the main solution of the algorithm after the $i$'th update in epoch $e$. As a result, the total recourse of the algorithm equals.
$$ \sum_{e=1}^E \sum_{i=0}^{l_e} |\calU^{(i+1)}_e - \calU^{(i)}_e|. $$
The recourse after each lazy update is at most $1$, which concludes for every $1 \leq e \leq E$, we have
\begin{eqnarray*}
   \sum_{i=0}^{l_e} |\calU^{(i+1)}_e - \calU^{(i)}_e| &\leq& l_e + |\calU^{(l_e+1)}_e - \calU^{(l_e)}_e| \\
   &\leq& l_e + | \calU^{(l_e)}_e - \calU^{(0)}_e| + | \calU^{(0)}_e - \calU^{(l_e+1)}_e| \\
    &\leq& 2l_e + |\calU^{(0)}_e - \calU^{(l_e+1)}_e|.
\end{eqnarray*}
Hence,
$$ \sum_{e=1}^E \sum_{i=0}^{l_e} |\calU^{(i+1)}_e - \calU^{(i)}_e| \leq \sum_{e=1}^E
\left( 2l_e +|\calU^{(0)}_e - \calU^{(l_e+1)}_e| \right)
\leq 2T + \sum_{e=1}^E |\calU^{(0)}_e - \calU^{(l_e+1)}_e|. $$
So, it suffices to prove $$\sum_{e=1}^E |\calU^{(0)}_e - \calU^{(l_e+1)}_e| = O(T(\log \Delta)^2).$$

Recall \Cref{lem:robustify-calls-once}.
This lemma states that the \Robustify \ calls a \MakeRbst \ on every center $w \in \calW$ at most once.
So, if a center $u$ is added to $\calS$ in \Robustify, it is not going to be added to $\calS$ in this call to \Robustify \ again.
According to the \Robustify, it is obvious that $\sum_{e=1}^E |\calU^{(0)}_e - \calU^{(l_e+1)}_e|$ is at most the number of times that elements are added to $\calS$.

Now, we create a graph during the run of the algorithm as follows. The graph is a union of disjoint directed paths.
Each node of the graph is alive or dead.
At any time, every center in the main solution of the algorithm is corresponding to an alive node in this graph.
A center might be corresponding to several nodes in the graph (since it can be added to $\calS$ several times during the total run of the algorithm), but it is corresponding to only one alive node.

Initially, the graph consists of $k$ alive nodes corresponding to the initial solution of $k$ centers at the very beginning of the algorithm.
Now, we explain how we change the graph during the algorithm.
At the end of the epoch, we mark all of the nodes corresponding to centers $\uo - \calW$ as dead.
Assume a center $u \in \calW$ at the end of an epoch.
We have the following cases.

\begin{itemize}
    \item $u \notin \uo$. In this case, we know that there is no alive node in the graph corresponding to $u$. We add a new alive node to the graph corresponding to $u$.
    This would be the starting node of a new path in the graph.

    \item $u \in \uo$ and $u$ is detected by \FindSuspects.
    In this case, we know that $u$ is corresponding to an alive node $V_u$ in the graph.
    We mark this node $V_u$ dead and add a new alive node corresponding to $u$ in the graph.
    This is also the starting node of a new path in the graph.

    \item $u \in \uo$ and $u$ is not detected by \FindSuspects, but $u$ is added to $\calS$ during the \Robustify.
    In this case, we know that $u$ is corresponding to an alive node $V_u$ in the graph and also $u$ is going to be replaced by $u_0$ by a call to \MakeRbst.
    We mark $V_u$ as dead, add a new node $V_{u_0}$ corresponding to $u_0$ in the graph and draw a directed edge from $V_u$ to $V_{u_0}$.

    \item $u \in \uo$ and $u$ is not detected by \FindSuspects, and $u$ is not added to $\calS$ during the \Robustify. In this case, we do not change the graph.
\end{itemize}

\begin{figure}[ht]
\caption{The constructed graph looks like this. If a center is detected by \FindSuspects, it would be the end of the path. If a \MakeRbst \ is called on a center, the path ending in the node corresponding to this center is extended.}
\centering

\tikzset{every picture/.style={line width=0.75pt}} %set default line width to 0.75pt        

\begin{tikzpicture}[x=0.75pt,y=0.75pt,yscale=-1,xscale=1]

%Straight Lines [id:da1429145369512932] 
\draw    (62.33,45) -- (183.33,45) ;
\draw [shift={(185.33,45)}, rotate = 180] [color={rgb, 255:red, 0; green, 0; blue, 0 }  ][line width=0.75]    (10.93,-3.29) .. controls (6.95,-1.4) and (3.31,-0.3) .. (0,0) .. controls (3.31,0.3) and (6.95,1.4) .. (10.93,3.29)   ;
%Straight Lines [id:da45538058561818073] 
\draw    (185.33,45) -- (306.33,45) ;
\draw [shift={(308.33,45)}, rotate = 180] [color={rgb, 255:red, 0; green, 0; blue, 0 }  ][line width=0.75]    (10.93,-3.29) .. controls (6.95,-1.4) and (3.31,-0.3) .. (0,0) .. controls (3.31,0.3) and (6.95,1.4) .. (10.93,3.29)   ;
%Straight Lines [id:da9888977452195133] 
\draw    (308.33,45) -- (429.33,45) ;
\draw [shift={(431.33,45)}, rotate = 180] [color={rgb, 255:red, 0; green, 0; blue, 0 }  ][line width=0.75]    (10.93,-3.29) .. controls (6.95,-1.4) and (3.31,-0.3) .. (0,0) .. controls (3.31,0.3) and (6.95,1.4) .. (10.93,3.29)   ;
%Straight Lines [id:da025774580248822287] 
\draw    (60.33,108) -- (181.33,108) ;
\draw [shift={(183.33,108)}, rotate = 180] [color={rgb, 255:red, 0; green, 0; blue, 0 }  ][line width=0.75]    (10.93,-3.29) .. controls (6.95,-1.4) and (3.31,-0.3) .. (0,0) .. controls (3.31,0.3) and (6.95,1.4) .. (10.93,3.29)   ;
%Straight Lines [id:da2819210909483707] 
\draw    (183.33,108) -- (304.33,108) ;
\draw [shift={(306.33,108)}, rotate = 180] [color={rgb, 255:red, 0; green, 0; blue, 0 }  ][line width=0.75]    (10.93,-3.29) .. controls (6.95,-1.4) and (3.31,-0.3) .. (0,0) .. controls (3.31,0.3) and (6.95,1.4) .. (10.93,3.29)   ;
%Straight Lines [id:da9325202256646938] 
\draw    (306.33,108) -- (427.33,108) ;
\draw [shift={(429.33,108)}, rotate = 180] [color={rgb, 255:red, 0; green, 0; blue, 0 }  ][line width=0.75]    (10.93,-3.29) .. controls (6.95,-1.4) and (3.31,-0.3) .. (0,0) .. controls (3.31,0.3) and (6.95,1.4) .. (10.93,3.29)   ;
%Straight Lines [id:da32386967632660246] 
\draw    (429.33,108) -- (550.33,108) ;
\draw [shift={(552.33,108)}, rotate = 180] [color={rgb, 255:red, 0; green, 0; blue, 0 }  ][line width=0.75]    (10.93,-3.29) .. controls (6.95,-1.4) and (3.31,-0.3) .. (0,0) .. controls (3.31,0.3) and (6.95,1.4) .. (10.93,3.29)   ;
%Straight Lines [id:da3196387950202497] 
\draw    (72.33,226) -- (193.33,226) ;
\draw [shift={(195.33,226)}, rotate = 180] [color={rgb, 255:red, 0; green, 0; blue, 0 }  ][line width=0.75]    (10.93,-3.29) .. controls (6.95,-1.4) and (3.31,-0.3) .. (0,0) .. controls (3.31,0.3) and (6.95,1.4) .. (10.93,3.29)   ;
%Straight Lines [id:da9867887309590495] 
\draw    (195.33,226) -- (316.33,226) ;
\draw [shift={(318.33,226)}, rotate = 180] [color={rgb, 255:red, 0; green, 0; blue, 0 }  ][line width=0.75]    (10.93,-3.29) .. controls (6.95,-1.4) and (3.31,-0.3) .. (0,0) .. controls (3.31,0.3) and (6.95,1.4) .. (10.93,3.29)   ;
%Shape: Circle [id:dp24189239170285504] 
\draw  [fill={rgb, 255:red, 0; green, 0; blue, 0 }  ,fill opacity=1 ] (62.33,45) .. controls (62.33,43.62) and (63.45,42.5) .. (64.83,42.5) .. controls (66.21,42.5) and (67.33,43.62) .. (67.33,45) .. controls (67.33,46.38) and (66.21,47.5) .. (64.83,47.5) .. controls (63.45,47.5) and (62.33,46.38) .. (62.33,45) -- cycle ;
%Shape: Circle [id:dp18546346801742164] 
\draw  [fill={rgb, 255:red, 0; green, 0; blue, 0 }  ,fill opacity=1 ] (60.33,108) .. controls (60.33,106.62) and (61.45,105.5) .. (62.83,105.5) .. controls (64.21,105.5) and (65.33,106.62) .. (65.33,108) .. controls (65.33,109.38) and (64.21,110.5) .. (62.83,110.5) .. controls (61.45,110.5) and (60.33,109.38) .. (60.33,108) -- cycle ;
%Shape: Circle [id:dp01637949544491213] 
\draw  [fill={rgb, 255:red, 0; green, 0; blue, 0 }  ,fill opacity=1 ] (180.83,108) .. controls (180.83,106.62) and (181.95,105.5) .. (183.33,105.5) .. controls (184.71,105.5) and (185.83,106.62) .. (185.83,108) .. controls (185.83,109.38) and (184.71,110.5) .. (183.33,110.5) .. controls (181.95,110.5) and (180.83,109.38) .. (180.83,108) -- cycle ;
%Shape: Circle [id:dp0849971496970503] 
\draw  [fill={rgb, 255:red, 0; green, 0; blue, 0 }  ,fill opacity=1 ] (182.83,45) .. controls (182.83,43.62) and (183.95,42.5) .. (185.33,42.5) .. controls (186.71,42.5) and (187.83,43.62) .. (187.83,45) .. controls (187.83,46.38) and (186.71,47.5) .. (185.33,47.5) .. controls (183.95,47.5) and (182.83,46.38) .. (182.83,45) -- cycle ;
%Shape: Circle [id:dp46150999783871427] 
\draw  [fill={rgb, 255:red, 0; green, 0; blue, 0 }  ,fill opacity=1 ] (305.83,45) .. controls (305.83,43.62) and (306.95,42.5) .. (308.33,42.5) .. controls (309.71,42.5) and (310.83,43.62) .. (310.83,45) .. controls (310.83,46.38) and (309.71,47.5) .. (308.33,47.5) .. controls (306.95,47.5) and (305.83,46.38) .. (305.83,45) -- cycle ;
%Shape: Circle [id:dp22002897771380736] 
\draw  [fill={rgb, 255:red, 0; green, 0; blue, 0 }  ,fill opacity=1 ] (428.83,45) .. controls (428.83,43.62) and (429.95,42.5) .. (431.33,42.5) .. controls (432.71,42.5) and (433.83,43.62) .. (433.83,45) .. controls (433.83,46.38) and (432.71,47.5) .. (431.33,47.5) .. controls (429.95,47.5) and (428.83,46.38) .. (428.83,45) -- cycle ;
%Shape: Circle [id:dp5374347121657728] 
\draw  [fill={rgb, 255:red, 0; green, 0; blue, 0 }  ,fill opacity=1 ] (549.83,108) .. controls (549.83,106.62) and (550.95,105.5) .. (552.33,105.5) .. controls (553.71,105.5) and (554.83,106.62) .. (554.83,108) .. controls (554.83,109.38) and (553.71,110.5) .. (552.33,110.5) .. controls (550.95,110.5) and (549.83,109.38) .. (549.83,108) -- cycle ;
%Shape: Circle [id:dp9492732150124463] 
\draw  [fill={rgb, 255:red, 0; green, 0; blue, 0 }  ,fill opacity=1 ] (426.83,108) .. controls (426.83,106.62) and (427.95,105.5) .. (429.33,105.5) .. controls (430.71,105.5) and (431.83,106.62) .. (431.83,108) .. controls (431.83,109.38) and (430.71,110.5) .. (429.33,110.5) .. controls (427.95,110.5) and (426.83,109.38) .. (426.83,108) -- cycle ;
%Shape: Circle [id:dp8323651724449772] 
\draw  [fill={rgb, 255:red, 0; green, 0; blue, 0 }  ,fill opacity=1 ] (303.83,108) .. controls (303.83,106.62) and (304.95,105.5) .. (306.33,105.5) .. controls (307.71,105.5) and (308.83,106.62) .. (308.83,108) .. controls (308.83,109.38) and (307.71,110.5) .. (306.33,110.5) .. controls (304.95,110.5) and (303.83,109.38) .. (303.83,108) -- cycle ;
%Shape: Circle [id:dp7774317633927406] 
\draw  [fill={rgb, 255:red, 0; green, 0; blue, 0 }  ,fill opacity=1 ] (69.83,226) .. controls (69.83,224.62) and (70.95,223.5) .. (72.33,223.5) .. controls (73.71,223.5) and (74.83,224.62) .. (74.83,226) .. controls (74.83,227.38) and (73.71,228.5) .. (72.33,228.5) .. controls (70.95,228.5) and (69.83,227.38) .. (69.83,226) -- cycle ;
%Shape: Circle [id:dp505974739576535] 
\draw  [fill={rgb, 255:red, 0; green, 0; blue, 0 }  ,fill opacity=1 ] (315.83,226) .. controls (315.83,224.62) and (316.95,223.5) .. (318.33,223.5) .. controls (319.71,223.5) and (320.83,224.62) .. (320.83,226) .. controls (320.83,227.38) and (319.71,228.5) .. (318.33,228.5) .. controls (316.95,228.5) and (315.83,227.38) .. (315.83,226) -- cycle ;
%Shape: Circle [id:dp18245258563743683] 
\draw  [fill={rgb, 255:red, 0; green, 0; blue, 0 }  ,fill opacity=1 ] (192.83,226) .. controls (192.83,224.62) and (193.95,223.5) .. (195.33,223.5) .. controls (196.71,223.5) and (197.83,224.62) .. (197.83,226) .. controls (197.83,227.38) and (196.71,228.5) .. (195.33,228.5) .. controls (193.95,228.5) and (192.83,227.38) .. (192.83,226) -- cycle ;
%Curve Lines [id:da18619257360193964] 
\draw  [dash pattern={on 0.84pt off 2.51pt}]  (431.33,45) .. controls (404.14,155.45) and (198.08,216.39) .. (71.73,225.86) ;
\draw [shift={(69.83,226)}, rotate = 355.92] [color={rgb, 255:red, 0; green, 0; blue, 0 }  ][line width=0.75]    (10.93,-3.29) .. controls (6.95,-1.4) and (3.31,-0.3) .. (0,0) .. controls (3.31,0.3) and (6.95,1.4) .. (10.93,3.29)   ;

% Text Node
\draw (77,19) node [anchor=north west][inner sep=0.75pt]   [align=left] {Make-Robust};
% Text Node
\draw (198,19) node [anchor=north west][inner sep=0.75pt]   [align=left] {Make-Robust};
% Text Node
\draw (324,19) node [anchor=north west][inner sep=0.75pt]   [align=left] {Make-Robust};
% Text Node
\draw (75,85) node [anchor=north west][inner sep=0.75pt]   [align=left] {Make-Robust};
% Text Node
\draw (196,82) node [anchor=north west][inner sep=0.75pt]   [align=left] {Make-Robust};
% Text Node
\draw (322,82) node [anchor=north west][inner sep=0.75pt]   [align=left] {Make-Robust};
% Text Node
\draw (292.05,179.54) node [anchor=north west][inner sep=0.75pt]  [rotate=-333.16,xslant=0.06] [align=left] {Find-Suspects};
% Text Node
\draw (441,82) node [anchor=north west][inner sep=0.75pt]   [align=left] {Make-Robust};
% Text Node
\draw (90,205) node [anchor=north west][inner sep=0.75pt]   [align=left] {Make-Robust};
% Text Node
\draw (205,204) node [anchor=north west][inner sep=0.75pt]   [align=left] {Make-Robust};
% Text Node
\draw (56,56) node [anchor=north west][inner sep=0.75pt]   [align=left] {{\tiny dead}};
% Text Node
\draw (177,54) node [anchor=north west][inner sep=0.75pt]   [align=left] {{\tiny dead}};
% Text Node
\draw (299,52) node [anchor=north west][inner sep=0.75pt]   [align=left] {{\tiny dead}};
% Text Node
\draw (426.83,54) node [anchor=north west][inner sep=0.75pt]   [align=left] {{\tiny dead}};
% Text Node
\draw (55,121) node [anchor=north west][inner sep=0.75pt]   [align=left] {{\tiny dead}};
% Text Node
\draw (174,119) node [anchor=north west][inner sep=0.75pt]   [align=left] {{\tiny dead}};
% Text Node
\draw (299,114) node [anchor=north west][inner sep=0.75pt]   [align=left] {{\tiny dead}};
% Text Node
\draw (420,118) node [anchor=north west][inner sep=0.75pt]   [align=left] {{\tiny dead}};
% Text Node
\draw (60,236) node [anchor=north west][inner sep=0.75pt]   [align=left] {{\tiny dead}};
% Text Node
\draw (186,234) node [anchor=north west][inner sep=0.75pt]   [align=left] {{\tiny dead}};
% Text Node
\draw (551.83,121) node [anchor=north west][inner sep=0.75pt]   [align=left] {{\tiny alive}};
% Text Node
\draw (312.83,235) node [anchor=north west][inner sep=0.75pt]   [align=left] {{\tiny alive}};

\end{tikzpicture}

\end{figure}
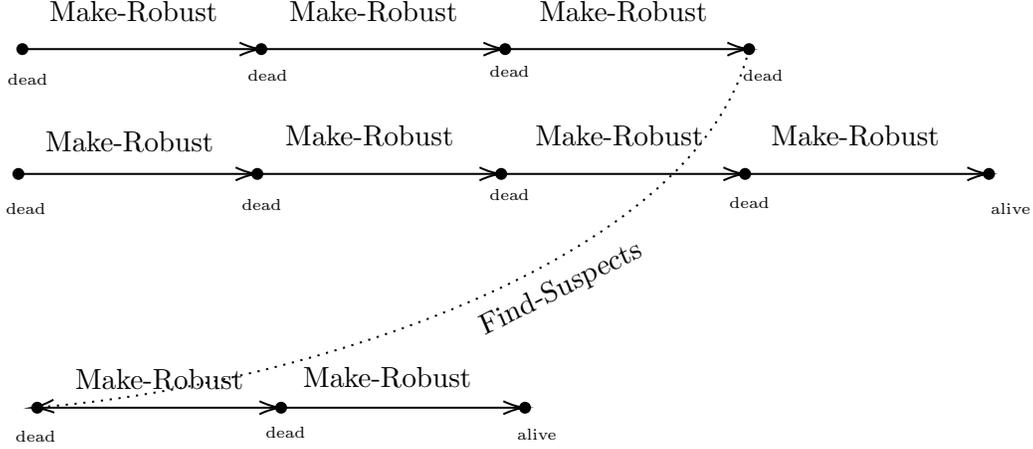

According to this construction of the graph, the total number of paths in the graph is equal to the total number of centers detected by \FindSuspects \ and the total number of nodes in the graph is equal to the total number of elements added to $\calS$ during the entire run of the algorithm.

By \Cref{cor:size-calS-find-suspects}, the number of centers detected by \FindSuspects \  at the end of epoch $e$ is at most $O((l_e+1)\log \Delta)$.
As a result, the total number of centers detected by \FindSuspects, during the run of the algorithm for the first $T$ updates of the input stream is at most $O(T\log\Delta)$.
This concludes the number of paths of the graph is $O(T\log\Delta)$.

Now, we show that the length of each path in the graph is at most $O(\log \Delta)$. This completes the analysis. 
Consider centers
$u_1,u_2,\ldots,u_r$ corresponding to the directed path
$V_{u_1},V_{u_2},\ldots,V_{u_r}$
in this graph.
First, note that for every $1 \leq i \leq r-1$, center $u_i$ is not detected by \FindSuspects \ during the algorithm. 
Otherwise, we would have marked $V_{u_i}$ as dead and $V_{u_i}$ would have been the end of the path.
The only reason that $u_i$ is added to $\calS$ is that
$10^{t(u_i)} < d(u_i, \calU-u_i)/200$.
In this case, the algorithm picks the smallest integer $t$ satisfying $10^t \geq d(u_i,\calU-u_i)/100$, replaces $u_i$ with $u_{i+1}$ by a call to \MakeRbst$(u_i,t)$ and sets $t(u_{i+1})=t$.
As a result, we have
$$ 10^{t(u_i)} < d(u_i, \calU-u_i)/200 < d(u_i,\calU-u_i)/100 \leq 10^{t(u_{i+1})} $$
for every $1 \leq i \leq r-1$.
Hence,
$t(u_{1}) < t(u_{2}) < \cdots < t(u_{r})$.
Since, we assumed the maximum distance between any two points in the metric space is at most $\Delta$, we conclude that
$t(u_{i}) \leq \lceil \log\Delta \rceil $ for every $1 \leq i \leq r$.
Together with $t(u_{1}) < t(u_{2}) < \cdots < t(u_{r})$, we conclude that $r \leq \lceil\log\Delta\rceil + 1 = O(\log \Delta)$.

Putting everything together, the number of paths in the graph is at most $O(T\log \Delta)$ and the length of each path is at most $O(\log\Delta)$, which concludes the total number of nodes in the created graph is at most $O(T(\log \Delta)^2)$. Hence, the total number of elements added to $\calS$ during the run of the algorithm is at most $O(T(\log \Delta)^2)$.
Finally, this implies
$$ \sum_{e=1}^E |\calU^{(0)}_e - \calU^{(l_e+1)}_e| = O(T(\log \Delta)^2), $$
which completes the analysis of recourse.

\section{Approximation Ratio Analysis}

In this section we prove that the approximation ratio of the algorithm is constant.

\noindent
\textbf{Important note:}
Since, subroutines \RandLocalSearch \ and \FastOneMedian \ are randomised and their approximation guarantee hold with high probability, some of the lemmas in this section occur with high probability.
But, we omit repeating this phrase over and over and provide the statements without it.
Assume $T$ many updates in the input stream.
According to the algorithm, we call \RandLocalSearch \ at most $O(T)$ times.
According to the recourse analysis in \Cref{sec:recourse}, we call \FastOneMedian \ at most $\tilde{O}(T)$ times.
As a result, we can conclude that all of these subroutines simultaneously return correctly with high probability.
So, we omit using w.h.p.~in the analysis in this section.

Assume that $\uo$ at the beginning of an epoch has approximation ratio $50\beta$. We will prove during the epoch, approximation ratio of our solution at any time is at most $ 2\cdot 10^7\gamma\beta^2$.
Then, we prove that the approximation ratio of final $\calU$ at the end of epoch is $50\beta$ which is the starting point of the next epoch.
Since the approximation ratio of the algorithm would be very large even if we optimize the parameters, we do not bother to use tight inequalities in the proofs.

\subsection{Analysis of \RemoveCenters \ and \LazyUpdates}
We start by analyzing \RemoveCenters.
Assume $\calR = \{ 0, 2^0,2^1,2^2,\ldots,2^{\lfloor \log k \rfloor} \}$ and
\begin{equation*}
    l^* = \arg\max \left\{ 0 \leq l \leq k \mid  \OPT_{k-l}(\po) \leq 400\gamma\beta \ \OPT_k(\po) \right\}.
\end{equation*}
According to \RemoveCenters, the value of $l'$ is the maximum number in $\calR$ such that
$$ \cost{\calU_{k-l'}, \po} \leq 14 \cdot 400\gamma\beta\ \cost{\uo, \po}. $$
Next, the algorithm sets $ l = \lfloor l' / C \rfloor$.
Now, we show good properties of this choice of $l$.

\begin{lemma}\label{lem:opt_{k-l'}<O(1)opt_k}
    We have
    $$ \OPT_{k-l'}(\po) \leq 3 \cdot 10^5\gamma\beta^2 \ \OPT_k(\po). $$    
\end{lemma}

\begin{proof}
According to the definition of $l'$, we have
\begin{eqnarray*}
    \OPT_{k-l'}(\po) &\leq& 
    \cost{\calU_{k-l'},\po} \\
    &\leq&
    14 \cdot 400 \gamma \beta \ \cost{\uo, \po} \\
    &\leq&
    50\beta \cdot 14 \cdot 400\gamma\beta \ \OPT_k(\po) \\
    &\leq&
    3 \cdot 10^5\gamma\beta^2 \ \OPT_k(\po).
\end{eqnarray*}

\end{proof}

\begin{lemma}\label{cor:(l+1)-bigger-than(lstar+1)/(2C)}
    We have $l+1 \geq (l^*+1)/(2C)$.
\end{lemma}

\begin{proof}
First, we show that $l' \geq l^*/2$.
We have three cases:
\begin{itemize}
    \item 
    $l^* = 0$. The inequality is obvious.
    
    \item 
    $l^* \geq 1$ and $l' = 0$. In this case, the value of $l'' = 1 \in \calR$ (which is bigger than $l'$) also satisfies
    \begin{eqnarray*}
        \cost{\calU_{k-l''}, \po} &\leq& 2 \ \cost{\uo,\po} + 12 \ \OPT_{k-l''}(\po) \\
        &\leq& 2 \ \cost{\uo,\po} + 12 \ \OPT_{k-l^*}(\po) \\
        &\leq&
        2 \ \cost{\uo,\po} + 12\cdot 400\gamma\beta \ \OPT_{k}(\po)  \\
        &\leq&
        (14\cdot 400\gamma\beta) \ \cost{\uo,\po}.
    \end{eqnarray*}
    The first inequality holds by the guarantee of \RandLocalSearch \ in \Cref{lem:cost-local-search}, the second one holds since $l'' \geq l^*$, and the third one holds by the definition of $l^*$.
    This is in contradiction with the choice of $l' \in \calR$ in the algorithm, since $l'' > l' $ and the algorithm must have picked $l''$ instead of $l'$.
    
    \item
    $l^* \geq 1$ and $l' \geq 1$. In this case if we have $l' < l^*/2$, then the above argument works for $l'' = 2l' \in \calR$ which leads to the same contradiction (note that since $l^*/2 \leq k/2$, then $l'' \leq k$ must be in $\calR$).
\end{itemize}
So, we have $l' \geq l^*/2$.
Since both of $l'$ and $C$ are integers, we have $l+1 \geq l'/C + 1/C$. Hence,
$ l+1 \geq (l'+1)/C \geq (l^*+1)/(2C) $.

\end{proof}

\begin{lemma}\label{optk-within-constant-factor}
    We have
    $$
  \frac{1}{3 \cdot 10^5\gamma\beta^2} \ \OPT_{k-l}(\po) \leq \OPT_{k}(\po) \leq 4 \ \OPT_{k+l}(\po).  $$
\end{lemma}

\begin{proof}
Using Double-Sided Stability \Cref{lem:double-sided-stability}, we easily conclude that the algorithm picks an $l$ with the property that both of $\OPT_{k-l}(\po)$ and $\OPT_{k+l}(\po)$ are within a constant factor of $\OPT_{k}(\po)$ as desired.
More precisely, let $\eta = 3 \cdot 10^5\gamma\beta^2$ and $r=l'$ in Double-Sided Stability \Cref{lem:double-sided-stability}.
Note that 
$$\lfloor r/(12\eta) \rfloor = \lfloor l'/(12\cdot 3 \cdot 10^5\gamma\beta^2) \rfloor = \lfloor l'/C \rfloor = l.$$
According to \Cref{lem:opt_{k-l'}<O(1)opt_k}, we have $\OPT_{k-r}(\po) \leq \eta \ \OPT_k(\po)$.
Hence,
$$\OPT_{k}(\po) \leq 4 \ \OPT_{k + \lfloor r / (12\eta) \rfloor }(\po) = 4 \ \OPT_{k + l}(\po). $$
Since $l' \geq l$, \Cref{lem:opt_{k-l'}<O(1)opt_k} concludes
$$ \OPT_{k-l}(\po) \leq \OPT_{k-l'}(\po) \leq 3 \cdot 10^5\gamma\beta^2 \ \OPT_k(\po). $$

\end{proof}

Next, we bound the cost of $\calU'$ returned by Rand-Local-Search$(\uo, l)$.

\begin{lemma}\label{cost-calU'}
    We have
    $$
   \cost{\calU', \po} \leq  2\cdot 10^7\gamma\beta^2 \ \OPT_{k+l}(\po). $$
\end{lemma}

\begin{proof}
According to \Cref{lem:cost-local-search} and \Cref{optk-within-constant-factor},
\begin{eqnarray*}
    \cost{\calU',\po} &\leq& 2\ \cost{\uo,\po} + 12 \ \OPT_{k-l}(\po) \\
    &\leq& 2\cdot(50\beta)\ \OPT_{k}(\po) + 12\cdot(3 \cdot 10^5\gamma\beta^2) \ \OPT_{k}(\po) \\
     &\leq&
     4\cdot 10^6\gamma\beta^2 \ \OPT_{k}(\po) \\
     &\leq&
     4\cdot (4\cdot 10^6\gamma\beta^2) \ \OPT_{k+l}(\po) \\
     &\leq&
     2\cdot 10^7\gamma\beta^2 \ \OPT_{k+l}(\po).
\end{eqnarray*}
\end{proof}

Finally, we show that the maintained solution 
during the lazy updates is a constant approximation of $\OPT_k$ on the current space.
Assume that $1 \leq i \leq l$ updates has happened to $\po$ and we have the current set of points $\calP^{(i)}$.
According to the algorithm, the solution for this set of points is $\calU^{(i)}$ such that $\calU' + (\calP^{(i)} - \po) \subseteq \calU^{(i)}$.
Note that there might be some point $p$ that has been added during these $i$ updates and then deleted. So, $p$ is not present in $\calP^{(i)} - \po$, but it is present in our solution $\calU^{(i)}$. As a result, \LazyUpdates \ \Cref{lem:lazy-updates} and \Cref{cost-calU'} conclude
\begin{eqnarray*}
    \cost{\calU^{(i)}, \calP^{(i)}} &\leq&
    \cost{\calU' + (\calP^{(i)} - \po), \calP^{(i)}} \\
    &\leq&
    \cost{\calU', \po} \\
    &\leq&
    2\cdot 10^7\gamma\beta^2 \ \OPT_{k+l}(\po) \\
    &\leq&
    2\cdot 10^7\gamma\beta^2 \ \OPT_{k}(\calP^{(i)}).
\end{eqnarray*}

So, we have the following.

\begin{corollary}\label{approx-ratio-during-lazy-updates}
    The approximation ratio of algorithm during \LazyUpdates \ is at most $2\cdot 10^7\gamma\beta^2$.
\end{corollary}

\subsection{Analysis of Develop-Centers}\label{sec:approx}

In this section, we show $\tuo$ returned by Develop-Centers is $10\beta$ approximate for $\OPT_{k+l+1}(\po)$.

We start by proving that it is possible to add at most $O(l+1)$ points to $\uo$ to obtain a set of centers $\tu$ which is $10$ approximate of $\OPT_{k+l+1}(\po)$.
Then, we show \DevelopCenters \ finds $\tuo$ that is $10\beta$ approximate of $\OPT_{k+l+1}(\po)$.
Assume $\calV^*_{k+l+1}$ is the optimum \textbf{proper} solution of size $k+l+1$ for $(k+l+1)$-median problem on $\po$ (we will use the assumption that this is a proper solution in \Cref{cor:existential-final}).
Note that $\OPT_{k+l+1}(\po)$ stands for the optimum improper solution.
According to \Cref{lem:proper-is-2-approx-of-improper}, $\calV^*_{k+l+1}$ is a $2$ approximate solution for the improper $(k+l+1)$-median on $\po$, i.e.
\begin{equation}\label{eq:calV-star-is-2-approx-for-improper-OPT}
  \cost{\calV^*_{k+l+1}, \po} \leq 2 \ \OPT_{k+l+1}(\po).  
\end{equation}

\begin{lemma}\label{cor:we-can-add-m+l+1-elements}
    If the number of well-separated pairs with respect to $(\uo, \calV^*_{k+l+1})$ is $k - m$, then there exists $\calV' \subseteq \calV^*_{k+l+1}$ of size at most $m+l+1$ such that
    $$ \cost{\uo + \calV', \po} \leq 10 \ \OPT_{k+l+1}(\po). $$
\end{lemma}

\begin{proof}
    Consider orderings $\uo = \{u_1,u_2,\ldots,u_k\}$ and $\calV_{k+l+1}^* = \{v_1,v_2,\ldots,v_{k+l+1}\}$ such that $(u_i,v_i)$ are all $k-m$ well-separated pairs w.r.t.~$(\uo, \calV_{k+l+1}^*)$ for $1 \leq i \leq k-m$. 
    Define
    $$\calV' = \{v_{k-m+1},v_{k-m+2},\ldots, v_{k+l+1}\}$$
    of size $m+l+1$.
    For simplicity, denote $C_{v_i}(\calV^*_{k+l+1}, \po)$ by $C_{v_i}$ for each $1 \leq i \leq k+l+1$. Note that these sets give a partition of $\po$.
    Now, consider $\uo + \calV'$.
    We show how to assign points of $\po$ to centers in $\uo + \calV'$.
    For every $1 \leq i \leq k-m$, we assign all of the points of $C_{v_i}$ to $u_i$ which is present in $\uo + \calV'$.
    This gives us, \begin{equation}\label{eq:i-from-1-to-k-m}
        \sum_{i=1}^{k-m} \cost{\uo + \calV', C_{v_i}}
        \leq \sum_{i=1}^{k-m} \cost{u_i, C_{v_i}}
        \leq
        \sum_{i=1}^{k-m} 5 \ \cost{v_i, C_{v_i}}.
    \end{equation}
    The last inequality is by \Cref{lem:cost-well-sep} for $\calP=\po$, $\calU = \uo$, $\calV = \calV^*_{k+l+1}$ and every $1 \leq i \leq k-m$.
    For each $k-m+1 \leq i \leq k+l+1$, we assign all of the points of $C_{v_i}$ to $v_i$ which is present in $\uo + \calV'$. This gives us, \begin{equation}\label{eq:i-from-k-m+1-to-k+l+1}
        \sum_{i=k-m+1}^{k+l+1} \cost{\uo + \calV', C_{v_i}} 
        \leq
        \sum_{i=k-m+1}^{k+l+1} \cost{v_i, C_{v_i}}.
    \end{equation}
    As a result,
    \begin{eqnarray*}
        \cost{\uo + \calV', \po} &=& \sum_{i=1}^{k+l+1} \cost{\uo + \calV', C_{v_i}} \\
        &=& \sum_{i=1}^{k-m} \cost{\uo + \calV', C_{v_i}} + \sum_{i=k-m+1}^{k+l+1} \cost{\uo + \calV', C_{v_i}} \\
        &\leq& \sum_{i=1}^{k-m} 5 \ \cost{v_i, C_{v_i}} + \sum_{i=k-m+1}^{k+l+1} \cost{v_i, C_{v_i}} \\
        &\leq& 5 \ \sum_{i=1}^{k+l+1} \cost{v_i, C_{v_i}} \\
        &=& 5 \ \cost{\calV_{k+l+1}^*, \po} \\
        &\leq& 10 \ \OPT_{k+l+1}(\po).
    \end{eqnarray*}
    The first inequality holds by \Cref{eq:i-from-1-to-k-m} and \Cref{eq:i-from-k-m+1-to-k+l+1} and the last inequality holds by \Cref{eq:calV-star-is-2-approx-for-improper-OPT}.     
\end{proof}

Next, we show that the number of well-separated pairs w.r.t.~$(\uo, \calV^*_{k+l+1})$ is large. 

\begin{lemma}\label{lem:m=O(l+1)}
    $m \leq (8C+1)(l+1)$.
\end{lemma}
    
\begin{proof}
    Let $b = \lfloor (m- (l+1)) / 4 \rfloor$.
    By plugging $\calP = \po$, $\calU = \uo$, $\calV = \calV^*_{k+l+1}$ and $r=l+1$ in \Cref{lem:generalize-lemma-7.3-in-FLNS21}, there exist a $\bar{\calU} \subseteq \uo$ of size at most $k - b$ such that
    $$ \cost{\bar{\calU}, \po} \leq 6\gamma\left( \cost{\uo, \po} + \cost{\calV^*_{k+l+1}, \po} \right). $$
    We also have
    $$\cost{\uo, \po} \leq 50\beta \ \OPT_k(\po) $$
    and
    $$ \cost{\calV^*_{k+l+1}, \po} \leq 2 \ \OPT_{k+l+1}(\po) \leq 2 \ \OPT_k(\po). $$
    Combining these three inequalities together, we have
    \begin{eqnarray*}
        \OPT_{k - b}(\po) &\leq& \cost{\bar{\calU}, \po} \\
        &\leq& 6\gamma\left( \cost{\uo, \po} + \cost{\calV^*_{k+l+1}, \po} \right) \\
        &\leq& 6\gamma\left( 50\beta \  \OPT_k(\po) + 2\  \OPT_k(\po) \right) \\
        &\leq& (6\gamma)\cdot(52\beta) \  \OPT_k(\po)  \\
        &\leq& 400\gamma\beta \ \OPT_k(\po).
    \end{eqnarray*}
    Now, with the maximality in the definition of $l^*$ in \Cref{def:l^*}, we have $\left\lfloor \frac{m-(l+1)}{4} \right\rfloor = b \leq l^*$ which concludes
    $$ m \leq 4(l^*+1)+(l+1).$$
    Together with \Cref{cor:(l+1)-bigger-than(lstar+1)/(2C)}, we have
    $m \leq (8C+1)(l+1)$.
\end{proof}

\begin{lemma}\label{cor:existential-final}
    Let $s = (8C+2)(l+1)$. We have
    $$ \min_{\substack{\calF \subseteq \po, \\ |\calF| \leq s}} \cost{\uo+\calF, \po}
   \leq 10 \ \OPT_{k+l+1}(\po). $$
\end{lemma}

\begin{proof}
    According to \Cref{cor:we-can-add-m+l+1-elements}, there exist a $\calV' \subseteq \calV^*_{k+l+1}$ of size at most $m+l+1$ such that
    $$\cost{\uo + \calV', \po} \leq 10 \ \OPT_{k+l+1}(\po).$$
    Note that $\calV' \subseteq \calV^*_{k+l+1} \subseteq \po$ since $\calV^*_{k+l+1}$ is a proper solution.
    We also have
    $$ |\calV'| \leq m+l+1 \leq (8C+2)(l+1) = s, $$
    where the last inequality holds by \Cref{lem:m=O(l+1)}.
    This concludes the lemma.
\end{proof}

\begin{lemma}\label{eq:calFstar+existential}
Assume $\tuo$ is returned by \DevelopCenters$(\po, \uo, s)$.
Then,
$$ \cost{\tuo, \po} 
    \leq 10\beta \  \OPT_{k+l+1}(\po). $$
\end{lemma}

\begin{proof}
According to \Cref{cor:existential-final} and the guarantee of \DevelopCenters \ in \Cref{lem:guarantee-develop-centers}, we have
\begin{eqnarray*}
    \cost{\tuo, \po} 
    \leq \beta\ 
    \min_{\substack{\calF \subseteq \po, \\ |\calF| \leq s}} \cost{\uo+\calF, \po}
   \leq 10\beta \ \OPT_{k+l+1}(\po).
\end{eqnarray*}
\end{proof}

\subsection{Analysis of Cost of $\calW$}

After termination of \DevelopCenters, the algorithm lets $\calV = \tuo + (\pl - \po)$ and calls \RandLocalSearch$ (\calV, |\calV| - k)$ to get $\calW$.

\begin{lemma}\label{lem:cost-calU-final}
    We have,
    $$
    \cost{\calV, \pl} \leq 10\beta\  \OPT_{k}(\pl).$$
\end{lemma}

\begin{proof}
According to \Cref{eq:calFstar+existential},
\begin{eqnarray*}
    \cost{\calV, \pl} &=&
    \cost{\tuo + (\pl - \po), \pl} \\
    &\leq&
    \cost{\tuo, \po} \\ 
    &\leq& 10\beta\ \OPT_{k+l+1}(\po) \\
    &\leq& 10\beta\  \OPT_{k}(\pl).
\end{eqnarray*}
\end{proof}

\begin{lemma}\label{eq:final-cost-of-calW}
    We have
    $$
    \cost{\calW, \pl} \leq  32\beta\  \OPT_k(\pl). $$
\end{lemma}

\begin{proof}
    
With guarantee of \RandLocalSearch \ in \Cref{lem:cost-local-search}, and with \Cref{lem:cost-calU-final}, we have
\begin{eqnarray*}
   \cost{\calW, \pl} &\leq& 2\ \cost{\calV, \pl} + 12\ \OPT_{k}(\pl) \nonumber \\
   &\leq& 20\beta\  \OPT_k(\pl) + 12\ \OPT_{k}(\pl) \nonumber \\
   &\leq& 32\beta\  \OPT_k(\pl)
\end{eqnarray*}

\end{proof}

\subsection{Analysis of Cost of \Robustify}
Finally, we analyze the cost of \Robustify\  on $\calW$.

\begin{lemma}[Lemma 7.5 in \cite{soda/FichtenbergerLN21}]\label{lem:cost-robustify}
    Assume $\calW$ is a set of $k$ centers. If $\calU$ is the output of \Robustify\ on $\calW$, then
    $$ \cost{\calU, \pl} \leq \frac{3}{2} \ \cost{\calW, \pl}. $$
\end{lemma}

\begin{proof}
By \Cref{lem:robustify-calls-once}, \Robustify \ calls \MakeRbst \ on each center at most once. 
Let $r$ be the number of centers in $\calW$ for which a call to \MakeRbst \ has happened.
Assume
$\calW = \{w_1,\ldots, w_k\}$ is an ordering by the time that \MakeRbst \ is called on the centers $w_1,w_2,\ldots,w_r$ and the last elements are ordered arbitrarily.
Assume also that $\calU =  \{w'_1,\ldots, w'_k\}$ where $w_i'$ is obtained by the call \MakeRbst \ on $w_i$ for each $1 \leq i\leq r$ and $w_i' = w_i$ for each $r+1 \leq i \leq k$.

For every $1 \leq j \leq r$, integer $t(w_j')$ is the smallest integer satisfying
$$ 10^{t(w_j')} \geq d(w_j,\{ w_1',\ldots,w_{j-1}',w_{j+1},\ldots,w_k \})/100. $$
Hence,
\begin{equation}\label{eq:proof-cost-robustify-dwjwiprime}
    10^{t(w_j')} \leq d(w_j,w_i')/10 \quad \forall i < j
\end{equation}
and
\begin{equation}\label{eq:proof-cost-robustify-dwjwi}
    10^{t(w_j')} \leq d(w_j,w_i)/10 \quad \forall i > j.
\end{equation}
Now, for every $1 \leq i < j \leq r$, we conclude
\begin{equation*}
    10\cdot 10^{t(w_j')} \leq d(w_i',w_j) \leq  d(w_i,w_j) + d(w_i,w_i') \leq d(w_i,w_j) + 10^{t(w_i')}/2 \leq \frac{11}{10} \ d(w_i,w_j).
\end{equation*}
The first inequality holds by \Cref{eq:proof-cost-robustify-dwjwiprime}, the third one holds by \Cref{lem:robustness-property-1}  (note that $w_i' = $ \MakeRbst$(w_i)$) and the last one holds by \Cref{eq:proof-cost-robustify-dwjwi} (by exchanging $i$ and $j$ since $i < j$).
This concludes for each $i < j$, we have $d(w_i,w_j) > 2 \cdot 10^{t(w_j')}$.
The same inequality holds for each $i>j$ by \Cref{eq:proof-cost-robustify-dwjwi}.
As a result,
$d(w_i,w_j) > 2 \cdot 10^{t(w_j')}$ for all $i \neq j$.
This concludes
$$\text{Ball}_{10^{t(w_j')}}^\calP(w_j) \subseteq C_{w_j}(\calW, \calP), $$
for every $1 \leq j \leq r$.
By applying \Cref{lem:robustness-property-2}, we have
\begin{equation}\label{eq:proof-cost-robustify-cost-wjprime-less-3-over-2}
   \cost{w_j',C_{w_j}(\calW, \calP)} \leq \frac{3}{2}\  \cost{w_j, C_{w_j}(\calW, \calP)} \quad \forall 1 \leq j \leq r. 
\end{equation}
Finally, since $C_{w_j}(\calW, \calP)$ for $1 \leq j \leq k$ form a partition of $\calP$, we conclude
\begin{eqnarray*}
    \cost{\calU, \calP} &=& 
    \sum_{j=1}^k \cost{\calU, C_{w_j}(\calW, \calP)} \\
    &\leq& 
    \sum_{j=1}^k \cost{w_j', C_{w_j}(\calW, \calP)} \\
    &=&
    \sum_{j=1}^r \cost{w_j', C_{w_j}(\calW, \calP)} + \sum_{j=r+1}^k \cost{w_j', C_{w_j}(\calW, \calP)} \\
    &\leq&
    \frac{3}{2}\  \sum_{j=1}^r \cost{w_j, C_{w_j}(\calW, \calP)} + \sum_{j=r+1}^k \cost{w_j, C_{w_j}(\calW, \calP)} \\
    &\leq&
    \frac{3}{2}\  \sum_{j=1}^k \cost{w_j, C_{w_j}(\calW, \calP)} \\
    &=&\frac{3}{2}\  \cost{\calW, \calP}.
\end{eqnarray*}
The second inequality holds by \Cref{eq:proof-cost-robustify-cost-wjprime-less-3-over-2}.

\end{proof}

\begin{corollary}\label{approx-final-end-of-epoch}
    If $\calU$ is the final solution of the algorithm for the epoch, then
    $$\cost{\calU, \pl} \leq 50\beta \ \OPT_{k}(\pl). $$
\end{corollary}

\begin{proof}
    According to \Cref{eq:final-cost-of-calW} and \Cref{lem:cost-robustify}, we have
    $$ \cost{\calU, \pl} \leq \frac{3}{2}\ \cost{\calW, \pl} \leq \frac{3}{2}\cdot 32\beta \ \OPT_{k}(\pl) \leq 50\beta \ \OPT_{k}(\pl). $$
\end{proof}

\subsection{Final Approximation Ratio}

According to \Cref{approx-ratio-during-lazy-updates}, the approximation ratio of the main solution during the epoch is at most $2\cdot 10^7\gamma\beta^2$. The approximation ratio of the main solution at the beginning of (and consequently end of) each epoch is at most $50\beta$ by \Cref{approx-final-end-of-epoch}.
As a result, the approximation ratio of the main solution is always at most 
$$2\cdot 10^7\gamma\beta^2 = \max \{2\cdot 10^7\gamma\beta^2, 50\beta \}. $$

    \section{Analysis of Update Time}

In this section, we prove that we can implement the algorithm that runs in amortized $\Tilde{O}(n)$ time.

\subsection{Sorted Distances to Main Centers}
Assume that $\calU$ is the current solution and $\calP$ is the current space.
For any point $p \in \calP + \calU$, we maintain a balanced binary search tree $\calT_p$ that contains the centers $\calU$ sorted by their distances to $p$.
Every time a new $p$ is inserted to $\calP$, we build the tree $\calT_p$ which takes $\tilde{O}(k)$ time (since we should sort the distance of $p$ to centers in $\calU$).
The total time needed to construct these trees after $T$ updates in the input stream is at most $\tilde{O}(Tk)$.

Whenever a point is deleted from $\calP$, if it is not present in our solution $\calU$, we remove the tree $\calT_p$, otherwise we keep the tree.
As a result, at any time the number of trees that we have is at most $n+k$, where $n$ is the size of the current space.

Every time a center $u$ is deleted from the main solution $\calU$ and $u'$ is inserted to $\calU$, for every tree $\calT_p$, we remove the element $u$ from $\calT_p$ and add $u'$ to $\calT_p$.
For each $p$, this can be done in $O(\log k)$ time.
So, the overall time consumed for updating these $n+k$ trees after each change in $\calU$ would be at most $\tilde{O}(n+k) = \tilde{O}(n)$, where $n$ is the size of the current space.
Since, we have already proved in \Cref{sec:recourse}, that the total change of the main solution during $T$ updates is at most $O(T(\log \Delta)^2)$, then the total time consumed for updating all of the trees throughout the first $T$ updates in the input stream is at most
$\tilde{O}(Tn)$,
where $n$ is the maximum size of the space during these $T$ updates.

As you can see, the total time needed for constructing and maintaining these trees throughout the first $T$ updates in the input stream is at most $\tilde{O}(Tn)$, where $n$ is the maximum size of the space during these $T$ updates.
Obviously, this is $\tilde{O}(n)$ in amortization.

So, we assume that we always have access to sorted distances of $p$ to all centers $\calU$, where $p$ is either present in the current space or in the main solution $\calU$.

\subsection{\RandLocalSearch}

In \cite{FOCS24kmedian}, the authors provide an implementation for this subroutine that runs in  $\tilde{O}(ns) $ time.
This is with the assumption that we have access to sorted distances of each $p \in \calP$ to centers in $\calU$.
We have these sorted distances by the argument in the previous section.
So, we conclude that we can run \RandLocalSearch, in $\tilde{O}(ns)$ time.
We refer the reader to Lemma 7.2 in \cite{FOCS24kmedian}, to see the full implementation and arguments.

\subsection{\RemoveCenters}

In this subroutine, we run a \RandLocalSearch \ for each $r \in \calR$ starting from the smallest value, until we reach the largest $r$ specified in the algorithm.
For each $r \in \calR$ that we call \RandLocalSearch$(\po, \uo, r)$,
the running time of \RandLocalSearch \ is $\tilde{O}(nr) $.
Since $|\calR| = O(\log k)$, by the description of the algorithm, it is obvious that the most costly \RandLocalSearch \ that we call is the last one which takes at most $\tilde{O}(n(2l')) = \tilde{O}(nl')$ time.
Finally, since $l+1 = \Omega(l'+1)$,
we can charge the total running time of this part to the $l+1$ updates in the epoch, which is $\tilde{O}(n)$ in amortization.
More precise, since we have not performed these $l+1$ updates yet at the beginning of the epoch, we can not charge the running time to these updates.
But, we can first charge them to the previous $l+1$ elements and at the end of the epoch, we take them back and charge them to $l+1$ updates in the epoch.
So, at any time during the algorithm, we charged $\tilde{O}(n)$ units to any update at most twice which is still $\tilde{O}(n)$.

\subsection{\DevelopCenters}

In this subroutine, we define a new space $\calP' = (\calP - \calU) + u^*$ with a new metric $d'$ as \Cref{def:new-metric}.
Since we have access to trees $\calT_x$, we know the value of $d(x,\calU)$ in constant time.
Note that in this subroutine, we only search for $\tu \subseteq \calP + \calU$.
So, we only need access to sorted distances of points in $\calP + \calU$ (instead of every point in $ \ground$) to centers in $\calU$ to be able to do this fast.

So, according to the definition of metric $d'$ in \Cref{def:new-metric}, for every $x \in \calP'-u^*$, we can evaluate $d'(x,u^*)$ in constant time using tree $\calT_x$ and for every $x,y \in \calP' - u^*$ we can evaluate $d'(x,y)$ in constant time using trees $\calT_x$ and $\calT_y$.
This means that we can design an oracle $D'$ that evaluates $d'(x,y)$ for every given $x,y \in \calP'$ in constant time.

Next, we run any fast $\beta$ approximate algorithm for \textbf{proper} $(s+1)$-median problem on $\calP'$ using oracle $D'$.
So, we have the desired $\calF$ in $\tilde{O}(|\calP'|(s+1))$ time which is $\tilde{O}(n(l+1))$ since $s = O(l+1)$ and $|\calP'|$ is at most the size of $\calP$.
Finally, we can charge this running time to the $l+1$ updates of the epoch which leads to an amortized running time of $\tilde{O}(n)$.

Note that in above arguments $n$ is the size of the current space $\calP$, not the underlying ground metric space $\ground$.

\subsection{\FastOneMedian}

The sampling process can be done easily.
First, consider an arbitrary ordering $B= \{ b_1,b_2, \ldots, b_{|B|}\}$.
Then sample a uniformly random $r \in [0,1]$.
Finally, we pick the smallest ${i^*}$ such that 
$$\sum_{j=1}^{i^*} w(b_j)/w(B) \geq r. $$
This ${i^*}$ indicates that our sample is $b_{i^*}$.
It is very easy to see that each $b_i$ is samples proportional to its weight $w(b_i)$ and also sampling a point takes $O(|B|)$ time.
The number of samples is $\tilde{O}(1)$.
So, the total time we need to sample these points is at most $\tilde{O}(|B|)$.
For each sample $s$, we iterate over all elements of $B$ to find $\cost{s,B}$ which takes $|B| = O(n)$ time.
Hence, the total running time of \FastOneMedian\ is at most $\tilde{O}(n)$.

\subsection{\FindSuspects}

Assume an epoch of length $l+1$.
We can find $\calW - \uo$ in $O(k)$ time by a simple iteration over centers $w \in \calW$ and check whether $w$ is in our main solution $\uo$ at the beginning of the epoch.
The number of points $p \in \po \oplus \pl$ is at most $l+1$.
The number of centers $u \in \uo \cap \calW$ is at most $k$ and these centers can be found in the same way as $\calW - \uo$.
For each $p$ and $u$ we can check $d(p,u) \leq 2 \cdot 10^{t(u)}$ in constant time.
In total, we can run \FindSuspects \ in $O(k(l+1))$ time.
Now, we charge this running time to $l+1$ updates of this epoch which gives an amortized update time of $O(k)$.

\subsection{\MakeRbst}

The main for loop is performed $t$ times where $t \leq \lceil \log \Delta \rceil = \tilde{O}(1)$ (since $\Delta$ is the aspect ratio of the metric space). For each $1 \leq i \leq t$, first we iterate over all points of $\calP$ to find $B_i$. This takes $O(n)$ time.
Finding $\avcost{w_i, B_i}$ also takes at most $|B_i| = O(n)$ time.
The most expensive part is calling \FastOneMedian, which takes $\tilde{O}(n)$ time.
Hence, the total running time of \MakeRbst\ is at most $\tilde{O}(n)$, where $n$ is the size of the current space.

\subsection{\Robustify}

At the beginning, a call to \FindSuspects \ is made which runs in amortized $O(k)$ time.
The analysis of \Robustify \ highly depends on the fact that total change of the main solution is $\tilde{O}(1)$ in amortization.
A single call to \Robustify \ might be costly, although we show that the amortized update time remains $\tilde{O}(n)$.

We show that throughout the first $T$ updates of the input stream, the running time of all calls to \Robustify \ is at most $\tilde{O}(Tn)$. 
\Robustify \ calls \EmptyS \ in Line~\ref{emptyS-line-robustify}, and have a loop in Line~\ref{for-loop-robustify}.
Assume the number of times that the algorithm calls \EmptyS \ is $N$. Note that $N$ is also the number of times \Robustify \  performs the for loop in Line~\ref{for-loop-robustify}.

$N$ is at most the number of elements added to $\calS$. This is because a call to \EmptyS \ is made only if $\calS$ is not empty, and when that happens, $\calS$ becomes empty through a call to \EmptyS.
So, to call \EmptyS \ again, at least one element should be added to $\calS$.
Finally, by the analysis in \Cref{sec:recourse}, the number of elements added to $\calS$ is at most $\tilde{O}(T)$ in total, which concludes $N = \tilde{O}(T)$.
The most costly part in \EmptyS \ is \MakeRbst.
Note that we have the value of $d(w,\calW-w)$ in constant time using tree $\calT_w$. The first element of the tree is $w$ since $d(w,w)=0$ and the second element is $w'$ such that $d(w,w') = d(w,\calW-w)$.
We will show that we can perform \MakeRbst \ in $\tilde{O}(n)$ time.
So, the total time spent by the algorithm for performing all of the calls to \EmptyS \ is at most  $N \cdot \tilde{O}(n) = \tilde{O}(Tn)$.

Now, consider the for loop in Line~\ref{for-loop-robustify}.
In this for loop, we are iterating over all elements of $\calW$ and the total running time to perform it is $O(k)$.
Note that we have the value of $d(u,\calW - u)$ in constant time for each $u \in \calW$ using tree $\calT_u$.
The total running time of this for loop is $ N \cdot O(k) = \tilde{O}(Tk)$.

In total, the running time of \Robustify \ throughout the first $T$ updates is at most $\tilde{O}(Tn)$ which is $\tilde{O}(n)$ in amortization.

\subsection{Total Running Time}

According to the implementation provided in this section, the total running time of the algorithm after $T$ updates in the input stream is at most $\tilde{O}(Tn)$, where $n$ is the maximum size of the space throughout these $T$ updates.
Note that $n$ is not the size of the underlying ground metric space $\ground$.
This is important since we are reducing the amortized running time to $\tilde{O}(k)$ in next section and we need the guarantee that $n$ is the maximum size of the space throughout the run of the algorithm.

    \section{Obtaining $\tilde{O}(k)$ Amortized Update Time}\label{sec:towards-update-time-O(k)}

In this section, we show how to reduce the amortized running time of the algorithm to $\tilde{O}(k)$ instead of $\tilde{O}(n)$.
Note that $n$ is the maximum size of the space at any time during the whole run of the algorithm (it is \textbf{not} the size of the underlying ground metric space $\ground$).
As a result, if we make sure that the size of the space is at most $\tilde{O}(k)$ at any point in time, then the amortized running time of the algorithm would be $\tilde{O}(k)$ as desired.

We use the result of \cite{ourneurips2023} to sparsify the input.
A simple generalization of this result is presented in Section 10 of \cite{FOCS24kmedian} which extends this sparsifier to weighted metric spaces.
The authors provided an algorithm to sparsify the space to $\tilde{O}(k)$ weighted points. 
More precisely, given a dynamic metric space $(\calP, w, d)$ and parameter $k \in \mathbb{N}$, there is an algorithm that maintains a dynamic metric space $(\calP', w', d)$ in $\tilde{O}(k)$ amortized update time such that the following holds.
\begin{itemize}
    \item $\calP' \subseteq \calP$  and the size of $\calP'$ at any time is $\tilde{O}(k)$.
    
    \item A sequence of $T$ updates in $(\calP, w, d)$, leads to a sequence of $\lambda \cdot T$ updates in $(\calP', w', d)$. In other words, the amortized recourse of $\calP'$ is at most $\lambda$.

    \item Every $\alpha$ approximate solution to the $k$-median problem in the metric space $(\calP', w', d)$ is also a $O(\alpha)$ approximate solution to the $k$-median problem in the metric space
    $(\calP,w,d)$ with probability at least $1- \tilde{O}(1/n^c)$.
\end{itemize}

The sparsifier of \cite{FOCS24kmedian} works by maintaining $O(\log \Delta)$ copies of the sparsifier of \cite{ourneurips2023}, feeding each one a subset of the metric space consisting of the points whose weights are in a given range, and taking the union of these spasifiers. Thus, it follows that the amortized recourse of this sparsifier is at most the amortized recourse of the sparsifier of \cite{ourneurips2023}. The authors of \cite{FOCS24kmedian} argue that the recourse of \cite{ourneurips2023} is at most $\tilde  O(1)$, and thus $\lambda = \tilde O(1)$. In \cite{simple-kcenter}, the authors give a more refined analysis of the recourse of \cite{ourneurips2023}, showing that it is constant. Thus, it follows that $\lambda = O(1)$.

Now, suppose that we are given a sequence of updates $\sigma_1, \sigma_2, \ldots, \sigma_T$ in a dynamic metric space $(\calP,w,d)$.
Instead of feeding the metric space $(\calP,w,d)$ directly to our algorithm, we can
perform this dynamic sparsification to obtain a sequence of updates $\sigma'_1, \sigma'_2, \ldots, \sigma'_{T'}$ for a metric space $(\calP',w',d)$, where $T' = O(T)$, and feed the metric space $(\calP',w',d)$ to our dynamic algorithm instead.
Since our algorithm maintains a $O(1)$ approximate solution $\calU$ to the $k$-median problem in $(\calP',w',d)$, then $\calU$ is also a $O(1)$ approximate solution for the $k$-median problem in $(\calP,w,d)$ with high probability.

Since the length of the stream is multiplied by $O(1)$, we would have a multiplicative overhead of $O(1)$ in the amortized update time and recourse (amortized w.r.t.~the original input stream). 
As a result, we have \Cref{main-theorem-2}.

    \newpage

    \bibliography{bibl.bib}
    \bibliographystyle{alpha}

\end{document}